\documentclass[reqno]{amsart}

\usepackage[centering,top=1in,bottom=1in,left=0.9in,right=0.9in]{geometry}

\numberwithin{equation}{section}

\usepackage{amsmath, amssymb, xcolor, amsthm, url, algorithm, algorithmic}
\usepackage{graphicx} %

\definecolor{dark-gray}{gray}{0.3}
\definecolor{dkgray}{rgb}{.4,.4,.4}
\definecolor{dkblue}{rgb}{0,0,.5}
\definecolor{medblue}{rgb}{0,0,.75}
\definecolor{rust}{rgb}{0.5,0.1,0.1}
\usepackage[colorlinks=true]{hyperref}
\hypersetup{linkcolor=dkblue}    
\hypersetup{citecolor=rust}      
\hypersetup{urlcolor=rust} 
\usepackage{bbm}
\usepackage{enumitem}
\usepackage[normalem]{ulem}
\usepackage{subcaption} 
\newcommand{\R}{\mathbb{R}}
\newcommand{\B}{\mathcal{B}}
\newcommand{\Bf}{\mathcal{B}f}
\newcommand{\E}{\mathbb{E}}
\renewcommand{\P}{\mathbb{P}}
\newcommand{\hatf}{\widehat{f}}
\newcommand{\hatfft}{\widehat{\fft}}

\newcommand{\fft}{f^{ft}}
\newcommand{\gft}{g^{ft}}
\newcommand{\pft}{p^{\text{ft}}}
\newcommand{\N}{\mathbbm{N}}
\newcommand{\cov}{\operatorname{Cov}}

\newcommand{\ft}{^{ft}}
\newcommand{\w}{\omega}

\usepackage{cancel}
\newcommand{\xxx}{E_3}
\newcommand{\xxone}{E_2^{(1)}}
\newcommand{\xxtwo}{E_2^{(2)}}
\newcommand{\xxthree}{E_2^{(3)}}
\newcommand{\xone}{E_1^{(1)}}
\newcommand{\xtwo}{E_1^{(2)}}
\newcommand{\xthree}{E_1^{(3)}}

\usepackage{comment}

\newcommand{\noisecov}{\mathsf{k}_\epsilon}
\newcommand{\fsup}{\|f\|_{\infty}}
\newcommand{\fLtwo}{\|f\|_2}
\newcommand{\Lipcov}{\|\nabla\noisecov\|_\infty}
\newcommand{\boundcov}{\|\noisecov\|_\infty}
\newcommand{\Lipf}{L_f}
\newcommand{\covLone}{\|\noisecov\|_{1}}

\newcommand{\obs}{Y}
\newcommand{\kerneldec}{\mathcal{K}}
\newcommand{\Kft}{\kerneldec^{\text{ft}}}
\newcommand{\dom}{R}
\newcommand{\sam}{N}

\newcommand{\nc}{\normalcolor}

\newtheorem{thm}{Theorem}[section]

\newtheorem{example}[thm]{Example}
\newtheorem{lem}[thm]{Lemma}
\newtheorem{rmk}[thm]{Remark}

\newtheorem{defn}[thm]{Definition}
\newtheorem{assump}[thm]{Assumption}
\newtheorem{prop}[thm]{Proposition}

\title{Functional Multi-Target Detection via Bispectrum Inversion}
\author{Anna Little$^*$, Daniel Sanz-Alonso$^\dagger$, Mikhail Sweeney$^\ddagger$, and Ruiyi Yang$^\S$ }
\thanks{$^*$Department of Mathematics, Utah Center for Data Science, University of Utah}
\thanks{$^\dagger$Department of Statistics, University of Chicago}
\thanks{$^\ddagger$Department of Mathematics, University of Utah}
\thanks{$^\S$Institute of Natural Sciences, School of Mathematical Sciences, MOE-LSC, Shanghai Jiao Tong University}

\begin{document}

\begin{abstract}
This paper develops a functional theory for multi-target detection, where a compactly supported signal is recovered from a single noisy observation containing many unknown translations of the signal. Our formulation allows continuous, off-grid translations and correlated stationary Gaussian process noise, extending beyond the discrete, grid-aligned, white-noise models common in prior work. We analyze two uninitialized recovery algorithms based on autocorrelation analysis; in particular,
both algorithms first estimate the signal's bispectrum via a debiased third-order empirical autocorrelation. The signal  is then recovered from the estimated bispectrum using either a functional frequency marching scheme or a Kotlarski-type deconvolution formula.
For both algorithms, we prove %
non-asymptotic recovery guarantees for compactly supported signals without band-limiting assumptions. The resulting error bounds depend on the smoothness of the signal and the accuracy of bispectrum estimation, with the latter governed by the noise characteristics and the number of signal occurrences. %
Numerical experiments validate our theory and demonstrate accurate recovery in low-SNR regimes.
\end{abstract}

\maketitle

\section{Introduction}
In multi-target detection (MTD), the goal is to recover an unknown signal that appears many times at unknown locations within a single large noisy observation.  In contrast to classical alignment problems, the latent translations cannot be reliably estimated in the low signal-to-noise ratio (SNR) regime, so the locations of the signal instances must be treated as nuisance parameters rather than explicitly inferred \cite{alignmentlimits}. This makes MTD a natural model for low-SNR inverse problems in which repeated signal instances are present but not individually detectable, including biological imaging problems such as cryo-electron microscopy (cryo-EM) \cite{DBLP:detectionlimit} and cryo-electron tomography (cryo-ET) \cite{bharat2016resolving}, as well as spike sorting \cite{spikesorting} and passive radar \cite{passiveradar}.

This paper studies the functional MTD problem in general dimension: recovering a signal function $f:\mathbb{R}^d\to\mathbb{R}$ from a single noisy observation of the form
\[
Y(t)=\sum_{i=1}^{N}f(t-x_i)+\epsilon(t),
\qquad
t\in[- R, R]^d,
\]
where \(f\) is supported on a compact domain \(D\subset [-R,R]^d\), the unknown translations \(x_1,\dots,x_N\) are latent variables, and \(\epsilon\) is a centered stationary Gaussian process. 
We assume that the translations are well separated: the shifted copies of $f$ are non-overlapping, separated by a fixed margin, and all supported within the observation window $[-R,R]^d$.
While this assumption can be relaxed to more general spacing distributions \cite{MTDapptoCryoEM,L:MTDarbitraryspacing}, it simplifies the analysis and is standard in the MTD literature. Unlike most existing MTD formulations, we work in a continuous setting: the translations are arbitrary (off-grid), the signal is modeled as a compactly supported function rather than a finite-dimensional vector, and the noise is allowed to be spatially correlated.  Our functional setting is therefore  motivated by
applications such as cryo-EM and cryo-ET, where the shifts are inherently off-grid \cite{B:MTDrotations}, the signal represents a physical object naturally described by a function, and the noise is often spatially correlated \cite{hazon2022noise}. 
Figure \ref{fig:introplot} illustrates the raw data and the recovered signal.

Our recovery approach follows the method of invariants.
A key object is the signal's bispectrum, a translation-invariant third-order frequency-domain quantity, defined in Equation \eqref{equ:bispectrum_def},  which measures how different frequencies in the signal interact. 
We first 
show that an appropriately debiased empirical third-order autocorrelation function yields an 
unbiased
estimator of the signal's bispectrum with nonasymptotic concentration guarantees.  
We then recover the signal by inverting the bispectrum using one of two procedures: a functional frequency marching algorithm \cite{ogfreqmarch} and a deconvolution method derived from Kotlarski’s identity \cite{kotlarski1967characterizing,rao1992identifiability}. Both algorithms are uninitialized, operate directly on the observed data, and recover the signal up to translation, which is the natural identifiability limit of the problem.

Our main contribution is the first explicit end-to-end recovery guarantee for MTD in a genuinely functional setting. Specifically,  Theorem \ref{thm:ErrorBoundDeconvolutionEstimator} makes precise the following informal statement:

\begin{thm}[Main result––informal statement]
    Under suitable assumptions and up to translation, the estimator \(\widehat f\) produced by either of our two algorithms satisfies, with high probability, 
\[
\|\widehat f-f\|_2
\;\lesssim\;
\rho_N^{\frac{\beta-d/2}{2\beta+1}},
\]
where \(\beta\) governs the smoothness of the true signal $f$, \(d\) is the dimension of its domain, and \(\rho_N\) denotes the error in estimating the bispectrum.
\end{thm}

Proposition~\ref{prop:bispectrumconcentration} provides a high-probability, nonasymptotic bound on $\rho_N$.  Consequently, the estimators produced by our algorithms converge as $N\to\infty$ to the true signal with a quantitative rate depending on the smoothness $\beta$ and the dimension $d$. To our knowledge, this is the first result for MTD that closes the loop from a single noisy observation to a recovered signal with a provable signal-level error guarantee.  

\begin{figure}[tb]
    \centering
    \begin{subfigure}[t]{0.32\textwidth}
        \centering
        \includegraphics[width=\textwidth]{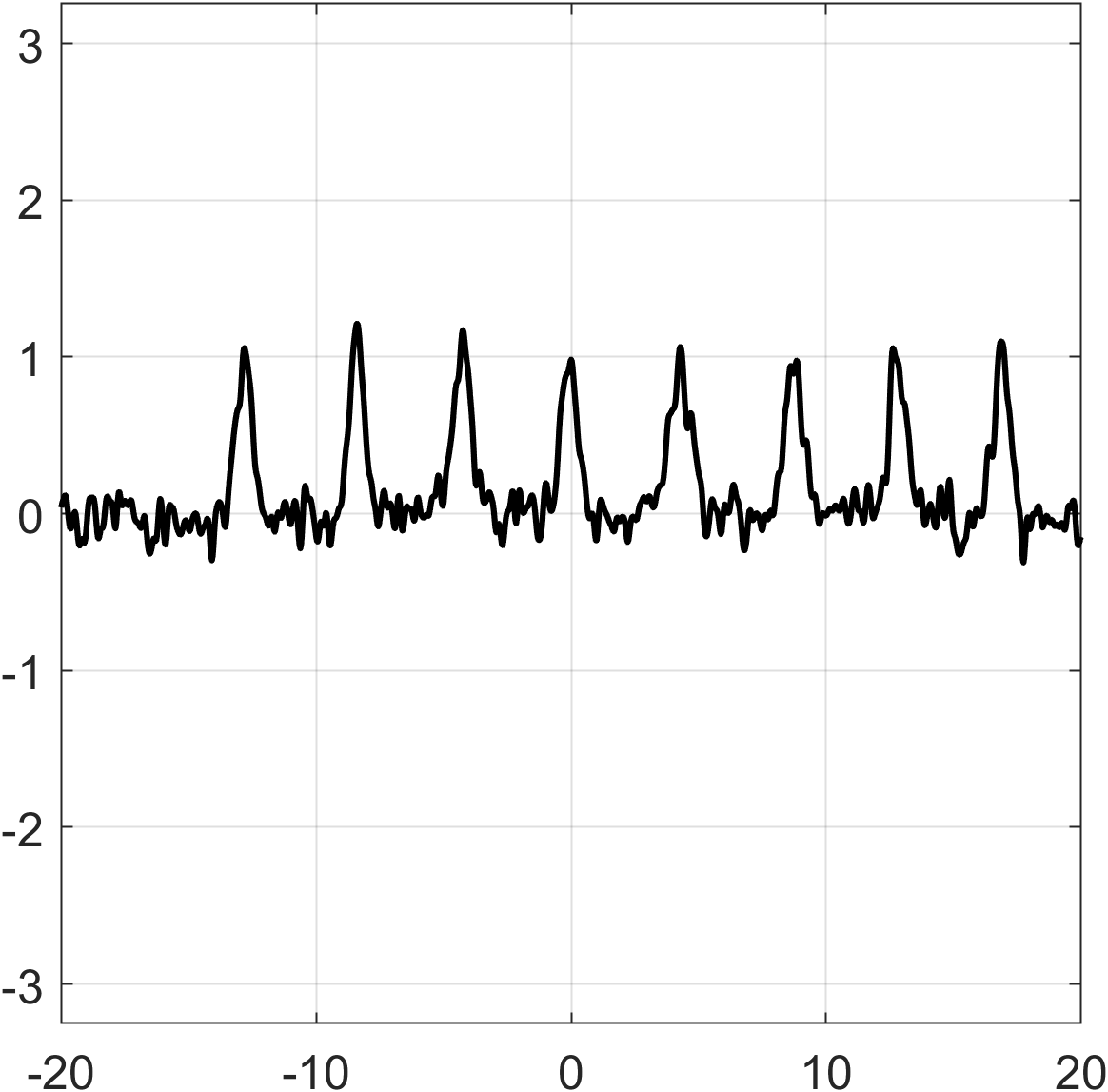}
        \caption{An MTD observation $\obs$ with $N=8$ and $\sigma = 0.1$.}
        \label{fig:introplotleft}
    \end{subfigure}
    \hfill
    \begin{subfigure}[t]{0.32\textwidth}
        \centering
        \includegraphics[width=\textwidth]{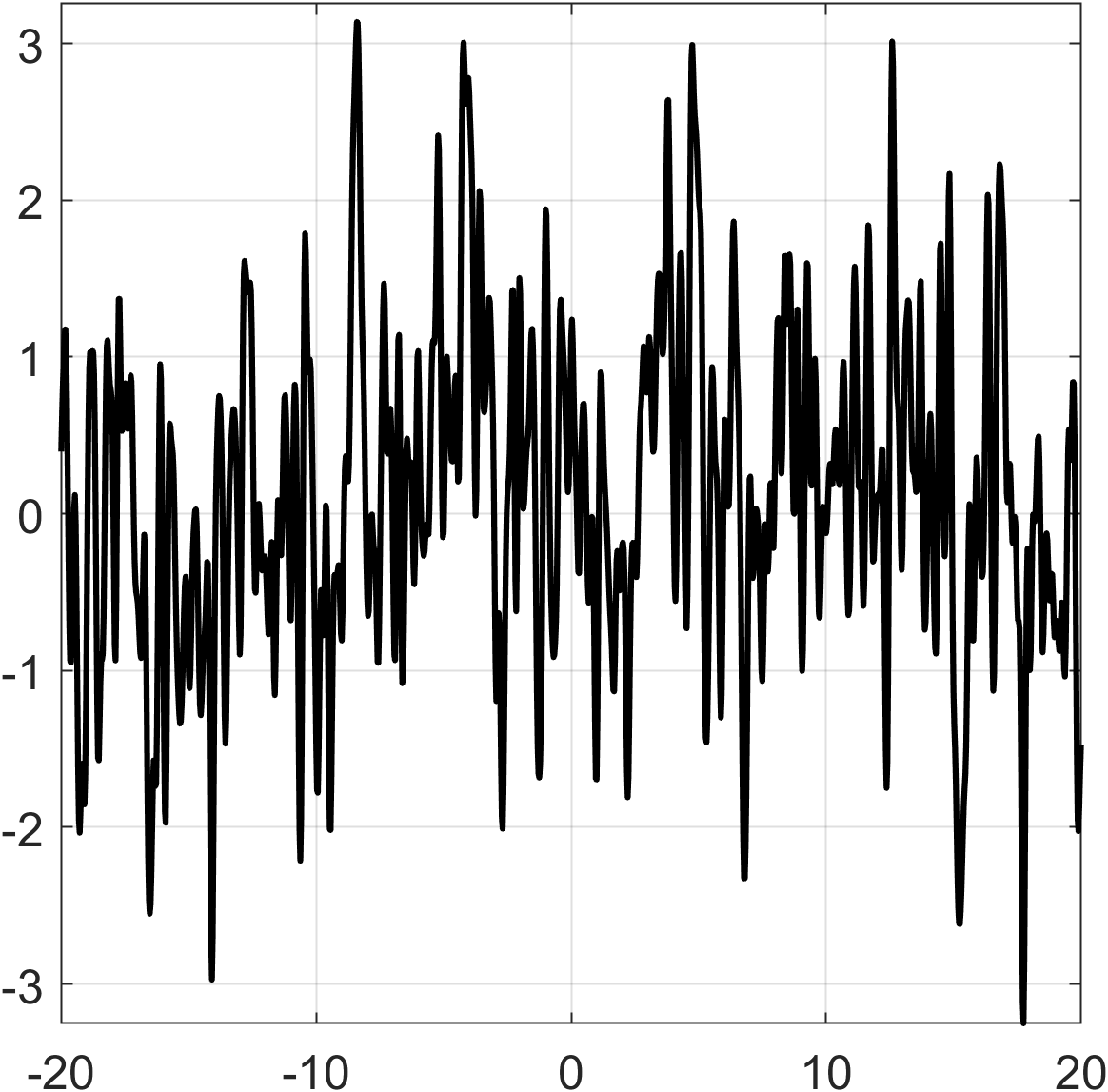}
        \caption{An MTD observation $\obs$ with $N=8$ and $\sigma = 1$.}
        \label{fig:introplotcenter}
    \end{subfigure}
    \hfill
    \begin{subfigure}[t]{0.32\textwidth}
        \centering
        \includegraphics[width=\textwidth]{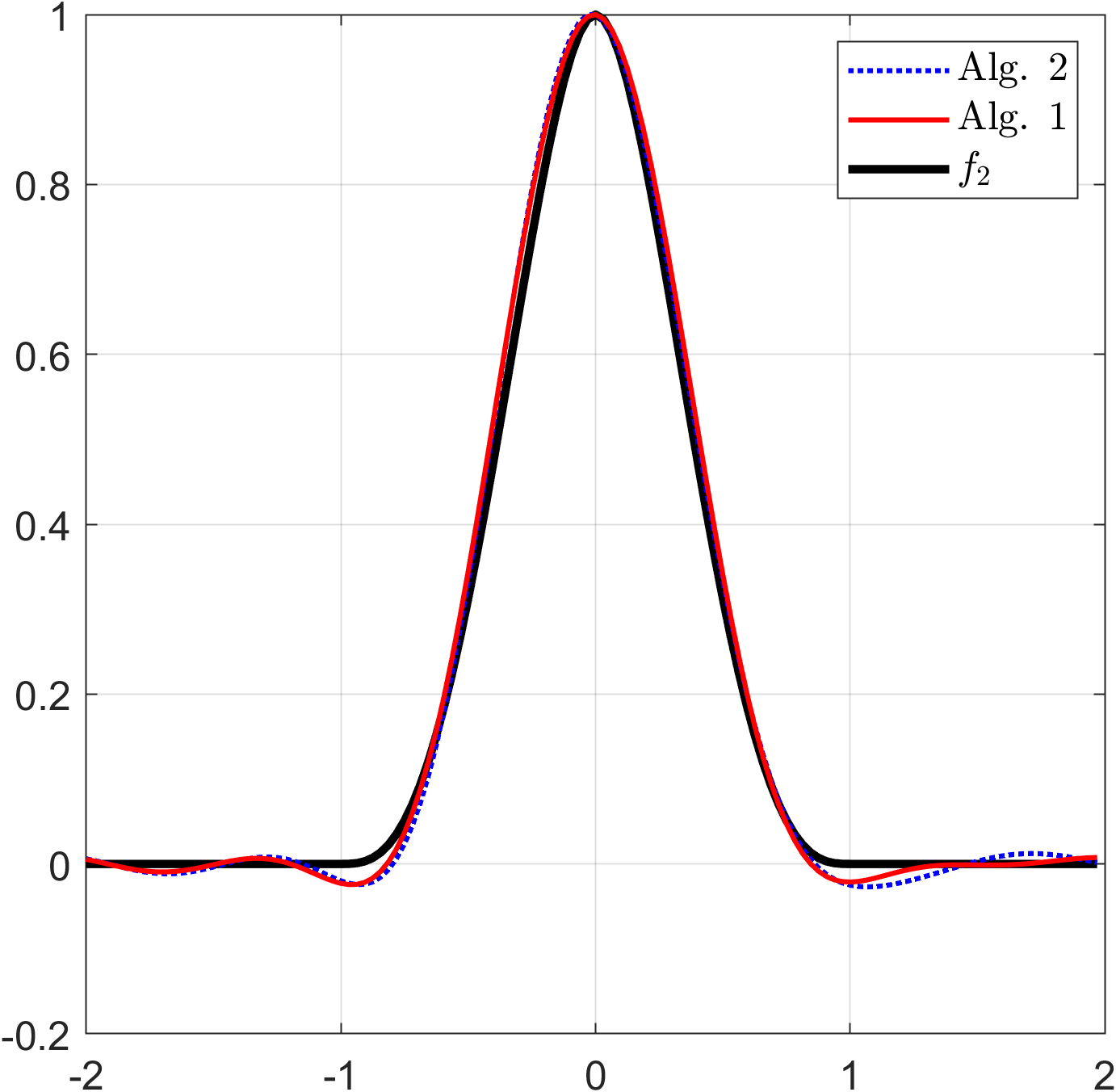}
        \caption{Signal recovered from a $\sigma=1$ observation $\obs$ with $N=2^{15}$.}
        \label{fig:introplotright}
    \end{subfigure}
    \caption{An illustration of the MTD problem and our recovery of the signal $f_2$ defined in \eqref{eq:signals}. The rightmost panel is recovered via Algorithms \ref{alg:mainalgFM} and \ref{alg:mainalg} from an observation $\obs$ with noise level $\sigma =1$, as in the center panel. Note the differing $y$-axes for the rightmost panel. See Section \ref{sec:numerics} for more details.} 
    \label{fig:introplot}
\end{figure}

\subsection{Motivation and Related Problems}
In cryo-EM, one seeks to recover the 3D structure of a molecule from a micrograph containing many randomly rotated, projected, and translated copies of the molecule, along with additive noise. In cryo-ET, a tilt series is acquired and reconstructed into a 3D tomogram of the specimen, which contains many copies of the molecule of interest at 
unknown locations and orientations. Subtomogram averaging then yields a 3D structure---and, critically, unlike cryo-EM, the signal model involves no tomographic projection, making it a more natural, albeit less studied, setting
for MTD-style recovery methods. In the standard cryo-EM and cryo-ET workflows, the translations are first estimated in a step known as \emph{particle picking}. 

The \emph{multi-reference alignment (MRA)} problem was developed as a simplified model for the cryo-EM workflow under the assumption of successful particle picking.
In MRA, one assumes access to many independently shifted noisy copies of the signal, rather than a single long observation containing many copies at unknown locations. 
For instance, in the functional MRA model of \cite{fMRA}, one seeks to recover a signal function $f: \R^d \to \R$ from $N$ independent observations
\[
y_i(t)=f(t-x_i)+\epsilon_i(t),
\qquad
t\in [-R,R]^d,\qquad 1\le i\le N,
\]
where $f$ is supported on a compact domain $D \subset [-R,R]^d,$  the shifts \(x_i\) are unknown, and the noise processes \(\epsilon_i\) are independent. Here, it is assumed that all shifted signals $f(t-x_i)$ are supported on the observation window $[-R,R]^d$, and the length $R$ of the observation window is fixed and independent of $N.$
As with MTD, MRA has been studied primarily in discrete settings, including settings with group actions beyond translations.
The sample complexity of MRA is studied in \cite{perry2019sample}. Solution methods for the low-SNR regime  ---avoiding explicit signal alignment--- include the method of moments \cite{hansen1982large, kam1980reconstruction, sharon2019method, abas2022generalized} (and in particular invariant methods \cite{bandeira2020non, bendory2017bispectrum, collis1998higher, hirn2023power, hirn2021wavelet}) and expectation maximization (EM) algorithms \cite{abbe2018multireference, dempster1977maximum, balanov2025expectation}. 
Extensions of the discrete MRA problem allowing additional group actions, such as rotations, have been studied in \cite{provableSO2MRA, SO3MRA, SEorbitrec}, and MRA with projections has been considered in \cite{provableMRAprojections}.
Altogether, the post-particle-picking cryo-EM problem has benefited substantially from the extensive MRA literature. 

However, it has recently been shown that particle-picking can be unreliable in the low-SNR regime and can contaminate reconstruction. Classically, particle-picking in cryo-EM  is done automatically with \emph{template matching}, in which signal locations are chosen as regions of the micrograph with increased correlation to a given template. This process can degrade the reconstruction in an unpredictable way, since cross-correlating noise with a template systematically produces spurious picks that resemble the template \cite{einsteinfromnoiseearly, einsteinfromnoisenew}.
Moreover, even template-free and aggregated particle picking approaches such as difference-of-Gaussian methods \cite{dogpicker}, deep-learning methods \cite{topazpicker}, and consensus methods \cite{cameron2024reliable} are unreliable in the low-SNR regime \cite{DBLP:detectionlimit}. Specifically, for molecules below a molecular weight of 40 kDa, it has been shown that particle picking is impossible \cite{CryoEmexcitement}; for cryo-EM to be viable for imaging small molecules, the bottleneck of particle picking must be avoided. Consequently, the works \cite{DBLP:detectionlimit,MTDapptoCryoEM} proposed the MTD computational framework as a simplified (projection-free) model of the cryo-EM problem without particle picking. MTD and MRA are thus closely related problems, with the sample complexity of MRA lower bounding the sample complexity of MTD \cite{abraham2025sample}.

\subsection{Previous Work on MTD}
The MTD problem was first introduced in \cite{MTDapptoCryoEM} as a complement to the proposal in \cite{DBLP:detectionlimit} of treating shifts in cryo-EM as nuisance parameters. The paper \cite{MTDapptoCryoEM} studied the one-dimensional discrete MTD problem and showed that the second- and third-order autocorrelation functions were enough to uniquely determine a signal vector $ z\in\R^L$. Moreover, the authors constructed unbiased estimators of the autocorrelation functions directly from the MTD observation $ \obs $. This was a continuation of the autocorrelation analysis methods proposed by Zvi Kam in 1980 \cite{Kam}, which laid out the groundwork for recovering signals corrupted by group actions from a method of moments approach, leveraging the invariance of the autocorrelation functions to translations. 
Shortly after, \cite{L:MTDarbitraryspacing} extended the previous results to hold for not only the well-separated model, but also the non-overlapping model. The authors also introduced an expectation-maximization (EM) algorithm as a multi-pass alternative to the autocorrelation-based method of moments / invariants approach and demonstrated a sample complexity of SNR$^{-3}$ in the low-SNR regime. \cite{S:generalizedMTD} placed autocorrelation analysis in the generalized method of moments framework and used this to construct an estimator for the autocorrelation functions of the hidden signal that attained minimum asymptotic variance.

The works \cite{marshall2020image,B:MTDrotations} extended the well-separated MTD formulation to two dimensions and added the additional group action of rotations, under a uniform rotation distribution assumption. While these papers assumed that the hidden signal was a compactly supported band-limited function, the MTD observation was assumed to be discrete, with shifts occurring on the discretization grid. The authors constructed unbiased estimators of the hidden signal's autocorrelation functions and computed the variance of these estimators. They also proposed an end-to-end recovery algorithm for their MTD problem based on inverting the bispectrum, as in this work. \cite{approxEMfor2dMTD} studied the same MTD problem formulation and constructed an approximate EM algorithm that attained comparable performance to the full EM algorithm proposed in \cite{L:MTDarbitraryspacing} but with much less computational complexity, requiring fewer passes through the data set.  Subsequently, \cite{diffusionpriorsforMTD} improved the approximate EM algorithm by incorporating learned score-based diffusion priors into the M-step via a regularized gradient ascent update, improving recovery quality
substantially in the low-SNR regime. The work
\cite{K:2dMTD} extended the 2D MTD with rotations problem formulation to the separated case, constructing pairwise separation functions to unbias the autocorrelation function estimators constructed from the MTD observation $ \obs $ and solving the MTD problem with a least squares non-convex optimization approach. 

The state-of-the-art for the EM algorithm workflow is developed in \cite{EMdirectlyfromCryoEM}, which extends the approximate EM algorithm constructed in \cite{approxEMfor2dMTD} to the full 3D cryo-EM problem, handling random shifts, rotations, and projections---demonstrating the first end-to-end algorithm for signal recovery from a cryo-EM micrograph without particle picking. The authors constructed and utilized a stochastic mini-batch variant of the EM algorithm. Meanwhile, formal sample complexity bounds for the MTD problem were first studied in \cite{noteonsamplecomplexity}, which lower bounded the MTD sample complexity by the sample complexity of the corresponding MRA problem and upper bounded the MTD sample complexity by an identifiability result for the autocorrelation functions. The paper \cite{abraham2025sample} also studied the sample complexity of MTD in both one and two dimensions, showing that, for method of moments approaches, the sample complexity of MTD matches that of a higher-dimensional i.i.d. MRA problem, up to a constant. That is, the mean squared error (MSE) of the MTD empirical moment estimator converges at the same rate as the i.i.d. MRA empirical moment estimator. Most recently, \cite{SEorbitrec} studied the functional well-separated MTD with rotations problem in arbitrary dimension and, under the assumption of non-vanishing antipodal correlation, provided upper and lower bounds for the sample complexity of the MTD problem.

\subsection{Our Contributions}
The principal contribution of this paper is to introduce and analyze two end-to-end uninitialized algorithms for signal recovery in the functional MTD setting. We provide concentration results for our estimators and explicit high-probability recovery guarantees for our algorithms. Our functional formulation enables natural modeling of the signal, the shifts, and the observational noise.

\subsubsection{Signal Modeling}
In many of the scientific problems that motivate MTD, the hidden signal represents a physical object, and is thus ideally modeled as a compactly supported function rather than a finite-dimensional vector. Moreover, the assumption of compact support in space implies infinite support in frequency, and thus the standard band-limited assumption in the MTD literature becomes unrealistic for physically relevant signals. With the exception of \cite{SEorbitrec}, which studies the full functional MTD problem but is predominantly interested in the sample complexity and does not characterize rates of recovery, this work is the only other paper in the MTD literature to study the functional MTD problem without a band-limited assumption on the hidden signal. Moreover, to the best of our knowledge, our proposed recovery rates are the first in the MTD literature to depend explicitly on properties of the hidden signal, such as signal smoothness.

\subsubsection{Shift Modeling}
A key feature of the functional formulation is that our recovery algorithms are sampling-rate independent. All proposed MTD recovery algorithms thus far make the assumption that the shifts inherent to the problem setup are multiples of the sampling rate; the discretized signal is always perfectly shifted on the grid. While discretization is necessary for practical implementation, in the functional setting discretization can occur after the signals have been acted upon by an \emph{infinite, continuous} group action. This is a more realistic assumption for the translations occurring in the motivating problems, such as cryo-EM. As \cite{B:MTDrotations} states in their Remark 3.3, \emph{``One potential extension of this model is
	to consider off grid translations... [this] would be a necessary extension if the presented approaches are adapted for an application problem involving real data.'' }

\subsubsection{Noise Modeling}
In addition, our problem formulation and recovery guarantees are set up for an arbitrary centered, stationary Gaussian noise process. In many applications, including cryo-EM, the noise is known to exhibit spatial correlations~\cite{hazon2022noise, bejjanki2017noise, huang2009noise}. As resolution increases, the assumption that pixelwise noise is uncorrelated becomes increasingly implausible; consequently, modeling the noise as a Gaussian process with short-range correlations (small lengthscale) may offer a more accurate representation than the universally used white-noise assumption in the MTD literature, as mentioned in the conclusion of \cite{EMdirectlyfromCryoEM}.

\subsubsection{Algorithms and Recovery Guarantees} 
The existing MTD literature has made substantial progress along two largely separate tracks. The algorithmic track has produced effective numerical methods such as autocorrelation analysis, approximate EM, and their variants \cite{MTDapptoCryoEM, L:MTDarbitraryspacing, B:MTDrotations, K:2dMTD, approxEMfor2dMTD, EMdirectlyfromCryoEM}, but offers no theoretical recovery guarantees: results consist of mean law-of-large-numbers convergence results for the autocorrelations, at best supplemented by variance scaling arguments on those intermediate quantities, rather than on the
recovered signal itself. The theoretical track establishes when recovery is information-theoretically possible using sample complexity bounds \cite{noteonsamplecomplexity, abraham2025sample, SEorbitrec}, but does not show
how to attain recovery algorithmically: the bounds are for moments, not for any specific algorithm, and no paper in the literature closes the loop from observed data to recovered signal with a quantitative error guarantee. 

This paper closes that loop for
the first time. We provide an explicit, end-to-end recovery guarantee for MTD using autocorrelation analysis in a genuinely functional setting, for compactly supported continuous signals with no band-limiting assumption and sampling-rate independent shifts.
Moreover, along the way we identify an infinite sampling-rate limit for the classical frequency marching algorithm \cite{ogfreqmarch}, constructing a novel functional version of the classic algorithm. Using this functional frequency marching algorithm, we prove a new explicit stability guarantee for bispectrum inversion via a Kotlarski-type integral formula, quantifying how perturbations in the estimated bispectrum propagate to errors in the recovered signal. This extends the analysis of bispectrum inversion in \cite{ogfreqmarch, bispecinversionMRA} to the functional setting. A distinguishing feature of our analysis is that it applies to signals with spectral decay rate $\beta > d/2$, in contrast to the analogous functional bispectrum inversion results of \cite{dou2024rates}, which require $\beta \in (0,1/2)$, a condition that excludes $L^2$ functions entirely. Our framework therefore extends the previous literature on bispectrum inversion. 
Additionally, in discretizations of the functional setting, our proposed algorithms for bispectrum inversion outperform other uninitialized algorithms for bispectrum inversion, such as the spectral algorithm proposed in \cite{spectralbispecinversion}, which is not robust to a fine sampling grid. 

\subsection{Outline}
The remainder of the paper is organized as follows. In Section~\ref{section:setup_algs}, we formalize the functional MTD model and present our two recovery algorithms. In Section~\ref{section:theory}, we establish concentration of the bispectrum estimator and prove end-to-end recovery guarantees. In Section~\ref{sec:numerics}, we present numerical experiments illustrating the performance of the proposed methods and validating the theoretical results. We close in Section \ref{sec:conclusions}.

\section{MTD: Functional Formulation and Inversion Algorithms}\label{section:setup_algs}
Let $f:\R^d \to \R$ be a compactly supported function on $D:= [-\pi,\pi]^d$, and  
let $\epsilon:\R^d \to\R$ be a centered stationary Gaussian process with covariance function $\noisecov$. We assume access to a large observation $\obs$ consisting of random shifts of the signal $f$ corrupted by the noise $\epsilon$; that is, we are given a single observation of the form
\begin{equation}\label{eq:M=F+noise}
\obs(t) = \sum_{i=1}^{\sam}f(t-x_i)+\epsilon(t) =: F(t) + \epsilon(t) \, , \qquad t \in [-\dom,\dom]^d,
\end{equation}
where $\{x_i\}_{i=1}^\sam$ are unknown latent variables, $F(t)=\sum_{i=1}^{\sam}f(t-x_i)$ is the noiseless component of the signal, and typically $\dom \gg 1$.  Our goal is to recover the hidden signal $f$ from the observation $\obs$. We assume that, as $\dom$ increases, the hidden signal continues to appear in the observation at the same rate, i.e. we consider a fixed density ratio $\gamma = \sam/(2\dom)^d$. In this setting, our proposed algorithms yield estimators of $f$ with error decaying to zero as $\dom,\sam \rightarrow \infty$.   

We adopt a two-stage approach to recovering $f$.  First, we estimate the bispectrum of $f$ from autocorrelation statistics. Second, we recover $f$ from its estimated bispectrum.

\subsection{Stage 1: Bispectrum Estimation}\label{ssec:bispectrumestimation}
Recall that the bispectrum of a function $g: \R^d \to \R$  is the function  $\B g:\R^d \times \R^d\to\mathbb{C}$  given by
\begin{equation}\label{equ:bispectrum_def}
   \B g(\omega_1,\omega_2):= \gft(\omega_1)\gft(-\omega_2)\gft(\omega_2-\omega_1) , 
        \qquad \omega_1, \omega_2 \in \R^d \, . 
\end{equation}
Here, the Fourier transform of a function $g$ is given by $g^\text{ft}(\omega) := \int_{\R^d} g(t) e^{-i\omega \cdot t} \, dt$. In this subsection we construct an estimator of the bispectrum of the signal $f$ using the third-order autocorrelation function of the data. To that end,  recall that the third-order autocorrelation of a function $g:\R^d\rightarrow \R$ on $[-\dom,\dom]^d$ is given by
\begin{align}\label{equ:A3_def}
A_3g(z_1,z_2) &= \int_{[-\dom,\dom]^d} g(t)g(t+z_1)g(t-z_2)\, dt.
\end{align}
As shown in Lemma \ref{lem:bispectrum}, for any function $g$ that is compactly supported on $[-(\dom-2\pi),\dom-2\pi]^d$, the Fourier transform of $A_3g$ produces the bispectrum of $g$, i.e.
\begin{align}\label{eq:bispectrum and A3}
    (A_3g)^{ft}(\omega_1,\omega_2) =\B g(\omega_1,\omega_2).
\end{align}

Since $f$ is compactly supported on $[-\pi,\pi]^d$, $A_3f(z_1,z_2)=0$ unless $(z_1,z_2) \in [-2\pi,2\pi]^{2d}$. Thus the only information relevant for recovering $f$ from autocorrelation statistics is contained in the domain $Z:= [-2\pi,2\pi]^{2d}$, and henceforth we will always consider $A_3g:Z\rightarrow \R$.

A key observation is that, under the well-separatedness condition in  Assumption \ref{assump:standing}(iv), $A_3F = \sam A_3f$, while the noisy observation $\obs$ satisfies
\begin{align}\label{eq:importantidentity}
\E\biggl[\frac{A_3\obs(z_1,z_2)}{\sam} \biggr] = A_3f(z_1,z_2)+\noisecov(z_2)+\noisecov(z_1)+\noisecov(z_1+z_2)  \, .
\end{align}
The proof of this identity can be found in Lemma \ref{lemma:A3M}. 
Leveraging this identity, we propose the following unbiased estimators for $A_3f$ and $\Bf$:
\begin{align}
    \label{equ:empirical_estimator}
    \widehat{A_3f}(z_1,z_2) &:= \frac{A_3\obs(z_1,z_2)}{\sam} - \noisecov(z_2)- \noisecov(z_1)-\noisecov(z_1+z_2)\, ,\quad\quad \widehat{ \Bf}:=(\widehat{A_3f})^{ft}\,.
\end{align} 

The estimator $\widehat{ \Bf}$ is not only unbiased, but it concentrates around the true bispectrum $\Bf$ as $\sam \to \infty$, as established in Proposition \ref{prop:bispectrumconcentration}.
Note that this estimator relies on knowledge of the number of occurrences $\sam$. This is not an issue, since $N$ can be estimated from the first- and second-order autocorrelation functions, as we describe in Remark \ref{rmk:n_est} and numerically demonstrate in Figure \ref{fig:n_est}. However, going forward we assume oracle knowledge of $\sam$ in order to simplify our error analysis.

\subsection{Stage 2: Bispectrum Inversion}\label{ssec:bispectruminversion}
We consider two algorithms for recovering the unknown signal from its bispectrum. The first leverages a well-known recursion formula to perform frequency marching on the Fourier series coefficients, while the second uses Kotlarski's integral formula.

\subsubsection{Functional Frequency Marching}
Frequency marching is an algorithm that exploits the definition of the bispectrum to recursively define increasingly high frequencies \cite{ogfreqmarch, bendory2017bispectrum}. Since by definition of the bispectrum \eqref{equ:bispectrum_def},
\[ \fft(\omega_1) = \frac{\B f(\omega_1,\omega_2)}{\fft(-\omega_2)\fft(\omega_2-\omega_1)}\, , \]
one can obtain $\fft(\omega+\Delta\omega)$ from $\fft(\omega)$ and $\fft(\Delta\omega)$ via
\begin{equation}\label{equ:FMupdate}
 \fft(\omega+\Delta\omega) = \frac{\B f(\omega+\Delta\omega,\omega)}{\fft(-\omega)\fft(-\Delta\omega)}\, .    
\end{equation}
Note that because $f$ is real-valued, $\fft(-\omega)=\overline{\fft(\omega)}$.
Since $f$ can be recovered from its Fourier series coefficients, i.e.
\[ f(t)=\frac{1}{(2\pi)^d}\sum_{k\in\mathbb{Z}^d} \fft(k)e^{i k \cdot t}, \]
it is sufficient to define $\fft(k)$ for $k\in\mathbb{Z}^d$, so typically \eqref{equ:FMupdate} is utilized with $|\Delta\omega|=1$. 
Algorithm \ref{alg:mainalgFM} summarizes the procedure for hidden signal recovery via frequency marching on the series coefficients. To reach frequency $k=(k_1,\dots,k_d)$, one first defines $\fft$ at $k=0$ and $k=\Delta \omega$ for all $|\Delta\omega|=1$, and then one iteratively marches along coordinate $r$ until reaching the terminal frequency $k_r$. We note that when $d=1$ and $\Delta\omega=1$, the recursive formula \eqref{equ:FMupdate} reduces to the simple form
\begin{align}\label{equ:FM1d}
\fft(j+1)
            =
            \frac{\Bf(j+1,j)}
            {\fft(-j)\fft(-1)} \, ,
\end{align}
as by symmetry only recovery of the positive frequencies is required.
If we represent $\fft(j) = r_je^{i\phi_j}$ and $\Bf(j,\ell)=r_{j,\ell}e^{i\Psi_{j,\ell}}$, then \eqref{equ:FM1d} implies the phase relationship
\begin{equation}\label{equ:FM_phase}
   \phi_{j+1} = \Psi_{j+1,j} + \phi_j + \phi_1 \, . 
\end{equation}
Since the magnitude of $\fft$ can be recovered in other ways (e.g. second-order autocorrelation function), frequency marching is classically just a procedure to recover the \textit{phase} of the unknown signal via \eqref{equ:FM_phase} or similar recursive formulas. However, we find it convenient to write this relationship in its full form (including the magnitudes) to (1) simplify its stability analysis and (2) derive a clear connection with Kotlarski's identity which we now describe. 
We work under the normalization $\fft(0)=1$; see Assumption \ref{assump:standing}(iii) and the subsequent discussion.  Under this normalization, we have that $|\fft(1)|=Bf(1,1)^{\frac{1}{2}}.$ The initialization step of Algorithm \ref{alg:mainalgFM} assumes that $\fft(1)=|\fft(1)|$, i.e. that $\phi_1=0$; this fixes a shift for the recovery, i.e. the algorithm will output $f(t-\phi_1)$.
Notice that since the bispectrum is translation invariant, any bispectrum inversion algorithm can only recover $f$ up to a shift.
                       
\begin{algorithm}[htp]
\caption{\label{alg:mainalgFM}Functional MTD via Frequency Marching on Fourier Series Coefficients}
\begin{algorithmic}[1]
    \STATE {\bf Input:} Functional data $\obs;$ noise covariance $\noisecov$; truncation parameter $K\in\mathbb{Z}^{+}$.
    
    \STATE {\bf Estimate $\Bf$:} Set
    \begin{align*}
    \widehat{A_3f}(z_1,z_2)
    &:=
    \frac{A_3 \obs(z_1,z_2)}{\sam}
    - \noisecov(z_2)- \noisecov(z_1)-\noisecov(z_1+z_2), \, \quad\quad 
    \widehat{\Bf}(\omega_1,\omega_2)
    :=
    (\widehat{A_3f})^{ft}(\omega_1,\omega_2).
    \end{align*}
    \STATE {\bf Initialize:} Let $e_1,\dots,e_d$ denote the standard basis vectors in $\mathbb{Z}^d$, and let
    \[
    \Omega_K^{(d)}:=\{k\in\mathbb{Z}^d:|k|_{\infty}\le K\}.
    \]
    Set
    \[
    \widehat{\fft}(0)=1,
    \qquad
    \widehat{\fft}(e_r)=|\widehat{\Bf}(e_r,e_r)|^{1/2},
    \qquad
    \widehat{\fft}(-e_r)=\overline{\widehat{\fft}(e_r)},
    \qquad r=1,\dots,d.
    \]

    \STATE {\bf Estimate Fourier coefficients:} For each $k=(k_1,\dots,k_d)\in\Omega_K^{(d)}$, let $\sigma_r:=\text{sgn}(k_r)$ for $r=1,\dots,d$ and set $j=(j_1,\ldots,j_d)\leftarrow \sigma_1e_1 \in\mathbb{Z}^d$.
    \FOR{$r=1,\dots,d$}
        \WHILE{$|j_r| < |k_r|$}
            \STATE
            \[
            \widehat{\fft}(j+\sigma_r e_r)
            =
            \frac{\widehat{\Bf}(j+\sigma_r e_r,j)}
            {\widehat{\fft}(-j)\widehat{\fft}(-\sigma_r e_r)},
            \qquad
            j\leftarrow j+\sigma_r e_r.
            \]
         \ENDWHILE
    \ENDFOR

    \STATE {\bf Invert:} Set
    \[
    \hatf(t):=\frac{1}{(2\pi)^d}\sum_{k\in\Omega_K^{(d)}} \widehat{\fft}(k)e^{i k \cdot t}.
    \]

    \STATE {\bf Output:} Approximation $\hatf$ to the hidden signal $f$.
\end{algorithmic}
\end{algorithm}

\subsubsection{Kotlarski's Identity}
Here we consider an alternative approach for bispectrum inversion based on Kotlarski's identity, motivated by the following result: 
\begin{prop}[Identification from bispectrum] 
\label{prop:kotlarski} 
Assume 
$f, |t|_2f(t) \in L^1(\R^d)$, $\fft(0)=1$, $\nabla\fft(0)=0$, and $\fft(\omega) \ne 0$ for all $\omega \in \R^d$. Then
    \begin{align}\label{eq:formulaft}
    \fft(\omega) = 
    \exp \left ( 
    \int_0^1 \frac{\nabla_{1} \Bf(\alpha \omega, \alpha \omega)}{\Bf(\alpha \omega,\alpha \omega)} \cdot \omega ~ d \alpha 
    \right ),
    \end{align}
    where $\nabla_1$ denotes the gradient with respect to the first argument.
\end{prop}

\begin{proof}
The result follows from  \cite[Theorem 2.1]{fMRA} and the observation that, in their notation, $\Psi = \Bf$ if $\pft_{\zeta}=\overline{\fft}$ (i.e. if $p_\zeta(t)=f(-t)$). We note that the proof of Theorem 2.1 in \cite{fMRA} does not require $p_\zeta$ to be a density, only that $\nabla \pft_{\zeta}(0)=0$.
\end{proof}

Kotlarski's identity originates in the deconvolution literature, where it is used to identify latent distributions from \emph{replicated measurements} \cite{kotlarski1967characterizing, rao1992identifiability}. In the classical setting, one observes two noisy replicates $Z_1=X_0+X_1$ and $Z_2=X_0+X_2$ of a latent variable $X_0$, with independent errors $X_1,X_2$, and the joint characteristic function of $(Z_1,Z_2)$ factors into a product involving the characteristic functions of the latent variable and the errors. Under a nonvanishing assumption, Kotlarski's formula recovers one of these characteristic functions from a logarithmic derivative of the joint characteristic function; see also \cite{li1998nonparametric, comte2015density, meister2007Deconvolution, kurisu2022uniform} for statistical developments and applications in deconvolution. In our previous work \cite{fMRA}, we used this perspective to recover $\fft$ from second-order statistics in functional MRA. Proposition~\ref{prop:kotlarski} shows that bispectrum inversion in the present setting admits an analogous structure, so that $\fft$ can again be recovered by the same integral mechanism.
Algorithm \ref{alg:mainalg} summarizes the procedure for hidden signal recovery via Kotlarski's identity.

\begin{rmk}[Regularization]\label{rmk:reg_conc}
In practice, the bispectrum estimator $\widehat{\Bf}$ constructed in \eqref{equ:empirical_estimator} can be regularized to an estimator $\widetilde{\Bf}$, defined as
\[
\widetilde{\Bf}(\omega_1,\omega_2) := \frac{\widehat{\Bf}(\omega_1,\omega_2)}{1 \land \left( \sqrt{\sam}\cdot|\widehat{\Bf}(\omega_1,\omega_2)|\right)} \, .
\]
This regularization is considered in \cite{kurisu2022uniform, foucart:mathCS2013} and helps prevent blow-up in Algorithms \ref{alg:mainalgFM} and \ref{alg:mainalg} due to dividing by small values of $\widehat{\Bf}$. 
\end{rmk}

\begin{algorithm}[htp]
\caption{\label{alg:mainalg}Functional MTD via Kotlarski's Identity}
\begin{algorithmic}[1]
     \STATE {\bf Input:} Functional data $\obs;$ kernel $\kerneldec;$ noise covariance $\noisecov$; bandwidth parameter $h>0.$ 
    \STATE  {\bf Estimate  $\Bf$:}
    Set \begin{align*}
    \widehat{A_3f}(z_1,z_2) := \frac{A_3 \obs(z_1,z_2)}{\sam} - \noisecov(z_2)- \noisecov(z_1)-\noisecov(z_1+z_2), \, \quad\quad \widehat{ \Bf}(\omega_1,\omega_2):=(\widehat{A_3f})^{ft}(\omega_1,\omega_2)\,.
\end{align*}
    \STATE  {\bf Estimate $\fft$:} Set
    \begin{align*}
       \widehat{\fft}(\omega)
        :=
        \exp \left ( 
            \int_0^1 \frac{\nabla_1 \widehat{\Bf}(\alpha \omega, \alpha \omega)}{\widehat{\Bf}(\alpha \omega,\alpha \omega)} \cdot \omega ~d\alpha
        \right) .
        \end{align*}
    \STATE  {\bf Deconvolve:} Set
      \begin{align*}
        \hatf(t) := \frac{1}{(2\pi)^d} \int_{\R^d} e^{i \omega \cdot  t} \hatfft(\omega) \Kft(h \omega) \, d\omega \,.
        \end{align*}
    \STATE {\bf Output}: Approximation $\hatf$ to the hidden signal $f.$ 
\end{algorithmic}
\end{algorithm}

\subsubsection{Kotlarski as Infinitesimal Limit of Frequency Marching} Although the frequency marching update \eqref{equ:FMupdate} and
Kotlarski's identity \eqref{eq:formulaft} appear different, they express the
same underlying idea at different scales. Frequency marching propagates
$\fft$ across a discrete frequency grid using finite bispectrum ratios, while
Kotlarski's identity propagates $\log \fft$ continuously by integrating a
logarithmic derivative of the bispectrum. In this sense, Kotlarski's formula
can be viewed informally as the infinitesimal limit of frequency marching.

To see this connection, consider first the one-dimensional setting and write
\[
    g(\omega):=\log \fft(\omega),
\]
on an interval where $\fft$ is nonvanishing and a branch of the logarithm has
been fixed. Formally taking the logarithm of the frequency marching update \eqref{equ:FMupdate} gives
\[
    g(\omega+\Delta\omega)
    =
    \log \Bf(\omega+\Delta\omega,\omega)
    -
    \log \fft(-\omega)
    -
    \log \fft(-\Delta\omega).
\]
Since $g(\omega)=\log \Bf(\omega,\omega)-\log \fft(-\omega)$, one obtains
\[
    g(\omega+\Delta\omega)-g(\omega)
    =
    \log \Bf(\omega+\Delta\omega,\omega)
    -
    \log \Bf(\omega,\omega)
    -
    \log \fft(-\Delta\omega).
\]
Assuming $\fft(0)=1$ and $(\fft)'(0)=0$ as in Proposition \ref{prop:kotlarski}, $\log \fft(-\Delta\omega)=O(|\Delta\omega|^2)$, and thus
\[
    \frac{g(\omega+\Delta\omega)-g(\omega)}{\Delta\omega}
    =
    \frac{
        \log \Bf(\omega+\Delta\omega,\omega)
        -
        \log \Bf(\omega,\omega)
    }{\Delta\omega}
    +
    O(|\Delta\omega|).
\]
Letting $\Delta\omega\to 0$ gives the differential relation
\[
    g'(\omega)
    =
    \partial_1 \log \Bf(\omega,\omega)
    =
    \frac{\partial_1 \Bf(\omega,\omega)}{\Bf(\omega,\omega)} .
\]
Since $g(0)=0$, integration yields
\[
    \fft(\omega)
    =
    \exp\left(
        \int_0^\omega
        \frac{\partial_1 \Bf(\xi,\xi)}
             {\Bf(\xi,\xi)}
        \, d\xi
    \right),
\]
which is the one-dimensional form of Kotlarski's identity. In higher dimensions, the analogous calculation is performed along the ray $\alpha\omega$, $0\leq \alpha \leq 1$, recovering \eqref{eq:formulaft}. 

\section{Recovery Guarantees}\label{section:theory}

We now establish performance guarantees for Algorithms \ref{alg:mainalgFM} and \ref{alg:mainalg}. We work under the following assumptions on the hidden signal, shifts, and noise.  We use $|\cdot|_2$ to denote the Euclidean norm and $\|\cdot\|_2$ to denote the $L^2$ function norm throughout the article. 

\begin{assump}[Model assumptions]\label{assump:standing} Suppose that:
    \begin{enumerate}[label=(\roman*)]
        \item The signal $f:\R^d\rightarrow\R$ is Lipschitz with constant $\Lipf$ and is compactly supported on $[-\pi,\pi]^d.$%
        \item There exists a constant $\beta>d/2$ and positive universal constants $C\ge c$ such that, for all $\omega\in\R^d,$
\[
c(1+|\omega|_2)^{-\beta}
\le
|\fft(\omega)|
\le
C(1+|\omega|_2)^{-\beta}.
\]
Furthermore, there is a positive universal constant $C_\nabla$ such that $\sup_{\omega \in \R^d} \frac{|\nabla \fft(\omega)|_2}{|\fft(\omega)|}  \leq C_\nabla.$
    \item The signal satisfies $\fft(0)=1$. Additionally, for Algorithm \ref{alg:mainalgFM}, $\fft(e_r)\in\mathbb{R}_+$ for standard basis vectors $e_r$, $r=1,\ldots,d$ and for Algorithm \ref{alg:mainalg}, $\nabla\fft(0)=0$.
    \item The shifts satisfy $|x_i-x_j|_2 >4\pi\sqrt{d} $ for $i\not=j$ 
and $F(t)$ has support contained in $[-(\dom-2\pi),\dom-2\pi]^d$.
    \item Signal occurrences have a fixed density ratio $\gamma = \sam/(2\dom)^d$ as $\dom,\sam \rightarrow \infty$.
    \item The noise $\epsilon$ is a centered and stationary Gaussian process with covariance function $\noisecov.$
    \item The noise covariance $\noisecov$ and its gradient $\nabla \noisecov$ both belong to $L^1(\R^d)\cap L^\infty(\R^d)$.
    \end{enumerate} 
\end{assump}

Examples of signals satisfying Assumption \ref{assump:standing} will be considered in the numerical experiments in Section \ref{sec:numerics}.
Note the parameter $\beta$ controls the smoothness of the hidden signal, and we restrict our attention to signals with non-vanishing Fourier transform and to shifts which are well separated. By Plancherel, Assumption \ref{assump:standing}(ii) guarantees that $\|f\|_2$ is upper and lower bounded by constants depending on $c,C,\beta,d$, and hence we regard $\|f\|_2$ as order one throughout our analysis. Additionally, note that we assume $\fft(0)=1$ only for ease of presentation: if $\fft(0)=a\neq 0$, 
the constant $a$ can be estimated by $A_1\obs / \sam$, where $\sam$ can be estimated via Remark \ref{rmk:n_est}, and the data can 
then be rescaled accordingly. 
The assumption of compact support on $[-\pi,\pi]^d$ is also for convenience (to obtain a Fourier series in terms of integer frequencies), and can be replaced with any compactly supported box $[-L,L]^d$.
Similarly, the purpose of Assumption \ref{assump:standing}(iii) is to fix a unique centering for the recovered signal, since $\Bf$ only defines $f$ up to translation; if not satisfied, Algorithm \ref{alg:mainalg} will return an approximation of $f(x+m)$ for $m = \int xf(x)\ dx$ and Algorithm \ref{alg:mainalgFM} will return an approximation of $f(x-\phi)$ for $\phi_r = \arg\fft(e_r), \phi = (\phi_1,\ldots,\phi_d)$. 

Algorithms \ref{alg:mainalgFM} and \ref{alg:mainalg} can also be extended to handle the case of vanishing Fourier transform. Although Assumption \ref{assump:standing}(ii) imposes a global nonvanishing condition on  $\fft$, Algorithm \ref{alg:mainalgFM} only requires that the Fourier coefficients used in the marching recursion are nonzero. 
Thus Algorithm \ref{alg:mainalgFM} can in principle tolerate $\fft$ having zeros
away from the integer lattice, in which case the stability constants depend on the minimum magnitude of the relevant coefficients. Likewise, Algorithm \ref{alg:mainalg} can be extended to the vanishing case via the algorithm discussed in \cite[Section 4]{fMRA}, which gives a generalized Kotlarski estimator relying on empirical estimation of the zeros of $\fft$. However to simplify our analysis and presentation, we consider only nonvanishing Fourier transforms here.

\subsection{Main Results}\label{subsec:main_results}

Our recovery guarantees rely on two key ingredients. Proposition \ref{prop:bispectrumconcentration} establishes concentration of the bispectrum estimator, while Proposition \ref{prop:FTDeviation} establishes the stability of 
bispectrum inversion under small perturbations of the bispectrum. Combining these two propositions yields our main result, Theorem \ref{thm:ErrorBoundDeconvolutionEstimator}, which provides an $L^2$ recovery guarantee for the hidden signal.

\begin{prop}[Concentration of $\widehat{\Bf}$]\label{prop:bispectrumconcentration}
Under Assumption \ref{assump:standing}, with probability at least $1-\delta$ the bispectrum estimator defined in \eqref{equ:empirical_estimator} satisfies
\begin{align*}
    \sup_{\omega_1,\omega_2 \in \R^d} | \widehat{\Bf}(\omega_1, \omega_2) - {\Bf}(\omega_1, \omega_2) | \leq \rho_\sam \, ,\quad
     \sup_{\omega_1,\omega_2 \in \R^d} | \nabla_1\widehat{\Bf}(\omega_1, \omega_2) - \nabla_1\Bf(\omega_1, \omega_2) | \leq \rho_\sam,
\end{align*}
where
\begin{small}
\begin{align}\label{equ:def_of_rhom}
    \rho_\sam= C_\B\left[(-\log \delta)^{3/2}\sqrt{\frac{\boundcov^2\|\nabla\noisecov\|_{1}+\boundcov \Lipcov\|\noisecov\|_{1}}{\sam}}+(-\log \delta)\sqrt{\frac{\Lipcov\fLtwo^2 \covLone}{\sam}}+(-\log \delta)^{1/2}\sqrt{\frac{\Lipf^2 \fsup^2 \|\noisecov\|_{1}}{\sam}}\right],
\end{align}
and $C_\B$ is a constant independent of $N$. 
\end{small}
\end{prop}

\begin{proof}
The result follows from Lemmas \ref{lemma:bispectrum < A3f}, \ref{lemma:error decomp}, and \ref{lem:concentration of A3}, which are stated and proved in Section \ref{subsec:bispectrum_conc}. 
\end{proof}

We note that the concentration results in Proposition \ref{prop:bispectrumconcentration} also hold for the regularized estimator $\widetilde{\Bf}$ defined in Remark \ref{rmk:reg_conc}. 

\begin{example}[Squared exponential noise covariance]\label{ex:squaredexponential}
Suppose that the noise is a centered Gaussian process with squared exponential covariance function
\begin{align}
\label{equ:squared_exp_kernel}
    \noisecov(h) = \sigma^2 \exp\left( -\frac{|h|_2^2}{2\ell^2} \right),
\end{align}
where $\sigma^2$ controls the noise level and $\ell$ is the correlation lengthscale. Then  
\begin{align*}
    \boundcov &= \sigma^2\, ,\qquad  \qquad 
    \|\noisecov\|_{1}  = \sigma^2(2\pi \ell^2)^{d/2}\, ,   \\
    \Lipcov  &= \frac{\sigma^2}{\ell \sqrt{e}}\, ,  \qquad 
    \|\nabla\noisecov\|_{1} = 2^{(d+1)/2}\sigma^2\pi^{d/2}\ell^{d-1}
    \frac{\Gamma(\frac{d+1}{2})}{\Gamma(\frac{d}{2})}\, .
\end{align*}
If $\delta\in(0,e^{-1})$ is fixed, then the value of $\rho_\sam$ in
\eqref{equ:def_of_rhom} scales with the noise
parameters $\sigma$ and $\ell$ as
\begin{align}\label{eq:rhosemisimplified}
\rho_\sam
\asymp
\frac{\ell^{(d-1)/2}\sigma^3}{\sqrt{\sam}}
+
\frac{\ell^{(d-1)/2}\sigma^2}{\sqrt{\sam}}
+
\frac{\ell^{d/2}\sigma}{\sqrt{\sam}}.
\end{align}
This scaling suggests that, as expected, larger noise level $\sigma$ and correlation lengthscale $\ell$ make bispectrum estimation harder. In particular, the concentration becomes slower as $\ell$ increases: for fixed estimation accuracy on a fixed window, longer-range correlations require a longer observation, i.e., larger $\sam$.

If, in addition, the noise level is large enough so that
$\ell^{1/2}\lesssim \sigma^2$ and $\sigma\gtrsim 1$, then \eqref{eq:rhosemisimplified} simplifies to
\[
\rho_\sam
\lesssim
\frac{\ell^{(d-1)/2}\sigma^3}{\sqrt{\sam}}.
\]
\end{example}

\begin{example}[White noise]
The white noise covariance model can be heuristically viewed as the limit of squared exponential kernels with $\ell\searrow 0$ and $\sigma^2 = \ell^{-d}({2\pi})^{-d/2}.$ Under this scaling, $\rho_N$ blows up as $\ell \searrow 0$ in any dimension. The fact that our concentration bounds and recovery guarantees deteriorate in this limiting regime is expected, since white noise is not supported in $L^2$ and its covariance is not represented by an $L^1\cap L^\infty$ function. 
\end{example}

\nc

We now turn to the second main ingredient of our signal recovery theory: stability of bispectrum inversion under small perturbations.

\begin{prop}[Stability of Algorithms \ref{alg:mainalgFM} and \ref{alg:mainalg}] \label{prop:FTDeviation}
Under Assumption \ref{assump:standing}, let \(\widehat{\Bf}\) be an
approximation of the bispectrum \(\Bf\). Then the following stability
estimates hold.
\begin{enumerate}[leftmargin=*]
\item 

  If  
  \[ \sup_{(k,m) \in\mathbb{Z}^{2d}}|\widehat{\Bf}(k,m) - \Bf(k,m)| \leq \epsilon \, , \]
  then for $K>1$ satisfying $\frac{2^{\beta} \sqrt{d}}{c^3}(1+\sqrt{d}K)^{2\beta+1} \epsilon \leq \frac12,$ where $c$ is the constant in Assumption \ref{assump:standing}(ii), 
  the approximate Fourier coefficients produced by Algorithm \ref{alg:mainalgFM} satisfy
  \[ \sup_{k \in \mathbb{Z}^d \cap [-K,K]^d}|\widehat{\fft}(k)-\fft(k)| \lesssim K^{\beta+1}\epsilon \, .\]

  \item If  
\begin{align*}
    \sup_{\omega \in \R^d} | \widehat{\Bf}(\omega, \omega) - {\Bf}(\omega, \omega) | \leq \epsilon\, , \quad
     \sup_{\omega \in \R^d} | \nabla_1\widehat{\Bf}(\omega, \omega) - \nabla_1\Bf(\omega, \omega) | \leq \epsilon \, ,
\end{align*}
then for $h \in (0,1)$ satisfying $\frac{1}{c^2}(1+\sqrt{d}h^{-1})^{2\beta+1}\epsilon\leq \frac{1}{2}$, where $c$ is the constant in Assumption \ref{assump:standing}(ii), 
the approximation $\hatfft$ produced by Algorithm \ref{alg:mainalg} satisfies
    \begin{align*}
        \sup_{\omega \in [-h^{-1},h^{-1}]^d}
        \left|\hatfft(\omega) - \fft(\omega) \right|
        \lesssim h^{-\beta-1}\epsilon.
    \end{align*}
    \end{enumerate}
    
\end{prop}

\begin{proof}
The proof is given in Section \ref{subsec:stability_of_bisp_inv}. 
\end{proof}

Having controlled the error of our bispectrum estimator in Proposition \ref{prop:bispectrumconcentration} and established the stability of Algorithms \ref{alg:mainalgFM} and \ref{alg:mainalg} under perturbations of the bispectrum in Proposition \ref{prop:FTDeviation}, we are ready to state and prove our main result on signal recovery. For Algorithm \ref{alg:mainalg}, we also impose the following assumption on the deconvolution kernel $\kerneldec$:

\begin{assump}[Deconvolution kernel]\label{assumption:kernel}
For a positive even integer $p$, the kernel $\kerneldec:\R^d \to \R$ satisfies:
\begin{enumerate}[label=(\roman*)]
    \item $\kerneldec\in L^1(\R^d)$ and  $  \int_{\R^d} \kerneldec(t) \, dt = 1.$
    \item $ \int_{\R^d} t^\nu \kerneldec(t) \, dt = 0 $ for all multi-indices $\nu \in \N^d$ with $1 \le |\nu| \le p-1.$ \label{assp:vanishing moment}
    \item $t \mapsto t^\nu \kerneldec(t) \in L^1(\R^d)$ for all multi-indices $\nu \in \N^d$ with $ |\nu| = p.$
    \item $  \int_{\R^d} |t|_2^p\, |\kerneldec(t)| \, dt \le C_\kerneldec $ where $C_\kerneldec$ is a finite constant allowed to depend on $\kerneldec, d$, and $p.$ 
    \item $ \Kft(\omega) = 0 $ for all $ |\omega|_\infty > 1 .$ 
\end{enumerate}
\end{assump}

\begin{thm}[Main result]\label{thm:ErrorBoundDeconvolutionEstimator}
Suppose Assumption~\ref{assump:standing} holds. Let $\rho_\sam$ be defined as in \eqref{equ:def_of_rhom}. For Algorithm \ref{alg:mainalg}, assume in addition that the deconvolution kernel $\kerneldec$ satisfies Assumption \ref{assumption:kernel} with $p>\beta-d/2$. Then, with probability at least $1-\delta$, the following statements hold.
    
\begin{enumerate}[leftmargin=*]
    \item The estimator $\hatf$ produced by Algorithm \ref{alg:mainalgFM} satisfies
    \begin{align*}
        \|\hatf-f\|_2
        \lesssim
        K^{d/2+\beta+1}\rho_\sam
        + K^{-\beta+d/2},
    \end{align*}
    for all $K>1$ satisfying $\frac{2^{\beta}\sqrt{d}}{c^3}(1+\sqrt{d}K)^{2\beta+1} \rho_N \leq \frac12$. In particular, choosing
    \[
        K=K^*\asymp \rho_\sam^{-1/(2\beta+1)}
    \]
    with the implicit constant sufficiently small and when $N$ is large enough so that $K^*>1$, we have
    \begin{align*}
        \|\hatf-f\|_2
        \lesssim
        \rho_\sam^{\frac{\beta-d/2}{2\beta+1}}.
    \end{align*}

    \item The estimator $\hatf$ produced by Algorithm \ref{alg:mainalg} satisfies
    \begin{align*}
       \|\hatf-f\|_2
       \lesssim
       h^{-d/2-\beta-1}\rho_\sam
       + h^{\beta-d/2},
    \end{align*}
    for all $h\in(0,1)$ satisfying $\frac{1}{c^2}(1+\sqrt{d}h^{-1})^{2\beta+1}\rho_\sam\leq \frac{1}{2}$. In particular, choosing
    \[
        h=h^*\asymp \rho_\sam^{1/(2\beta+1)}
    \]
    with the implicit constant sufficiently large and when $N$ is large enough so that $h^*\in(0,1)$, we have
    \begin{align*}
        \|\hatf-f\|_2
        \lesssim
        \rho_\sam^{\frac{\beta-d/2}{2\beta+1}}.
    \end{align*}
\end{enumerate}
\end{thm}

\begin{proof}
By Proposition \ref{prop:bispectrumconcentration}, with probability at least $1-\delta$,
\[
    \sup_{\omega_1,\omega_2\in\R^d}
    |\widehat{\Bf}(\omega_1,\omega_2)-\Bf(\omega_1,\omega_2)|
    \leq \rho_\sam,
\qquad 
    \sup_{\omega_1,\omega_2\in\R^d}
    |\nabla_1\widehat{\Bf}(\omega_1,\omega_2)-\nabla_1\Bf(\omega_1,\omega_2)|
    \leq \rho_\sam .
\]
We work on this event throughout the proof.

\textbf{Algorithm \ref{alg:mainalgFM}:} Proposition \ref{prop:FTDeviation} gives
\[
    \sup_{k\in\mathbb Z^d\cap[-K,K]^d}
    |\widehat{\fft}(k)-\fft(k)|
    \lesssim K^{\beta+1}\rho_\sam,
\]
provided that $\frac{2^{\beta}}{c^3}(1+\sqrt{d}K)^{2\beta+1} \rho_N \leq \frac12$.
Hence, by Plancherel's identity,
\begin{align*}
    \|\hatf-f\|_2^2
    &=
    \frac{1}{(2\pi)^d}\sum_{k\in \mathbb{Z}^d}
    |\widehat{\fft}(k)-\fft(k)|^2 \\
    &=
    \frac{1}{(2\pi)^d}\sum_{k\in\Omega_K^{(d)}}
    |\widehat{\fft}(k)-\fft(k)|^2
    +
    \frac{1}{(2\pi)^d}\sum_{k\in(\Omega_K^{(d)})^c}
    |\fft(k)|^2 \\
    &\lesssim
    \sum_{k\in\Omega_K^{(d)}} K^{2\beta+2}\rho_\sam^2
    +
    \sum_{k\in(\Omega_K^{(d)})^c} |k|_2^{-2\beta} \\
    &\lesssim
    K^{d+2\beta+2}\rho_\sam^2
    +
    K^{-2\beta+d}.
\end{align*}
Since $\beta>d/2$, taking square roots gives
\[
    \|\hatf-f\|_2
    \lesssim
    K^{d/2+\beta+1}\rho_\sam
    +
    K^{-\beta+d/2}.
\]
Setting $K\asymp \rho_\sam^{-1/(2\beta+1)}$ while keeping the implicit constant sufficiently small ensures that $\frac{2^{\beta}}{c^3}(1+\sqrt{d}K)^{2\beta+1} \rho_N \leq \frac12$ is satisfied, and when $N$ is large enough so that $K>1$, we obtain
\[
    \|\hatf-f\|_2
    \lesssim
    \rho_\sam^{\frac{\beta-d/2}{2\beta+1}}.
\]

\textbf{Algorithm \ref{alg:mainalg}:} Plancherel's identity and the definition
\[
    (\widehat f)^{ft}(\omega)=\hatfft(\omega)\Kft(h\omega)
\]
give
\begin{align*}
    \|\widehat{f}-f\|_2^2
    &=
    \frac{1}{(2\pi)^d}\int_{\R^d}
    |(\widehat{f})^{ft}(\omega)-\fft(\omega)|^2\,d\omega \\
    &=
    \frac{1}{(2\pi)^d}\int_{\R^d}
    |\hatfft(\omega)\Kft(h\omega)-\fft(\omega)|^2\,d\omega \\
    &\lesssim
    \int_{\R^d}
    |\Kft(h\omega)|^2
    |\hatfft(\omega)-\fft(\omega)|^2\,d\omega
    +
    \int_{\R^d}
    |\fft(\omega)|^2
    |\Kft(h\omega)-1|^2\,d\omega \\
    &=: I_1+I_2.
\end{align*}

For $I_1$, by Assumption \ref{assumption:kernel}(v), $\Kft(h\omega)=0$ unless $|h\omega|_\infty\le 1$, and hence unless $\omega\in[-h^{-1},h^{-1}]^d$. Therefore, using Proposition \ref{prop:FTDeviation},
\begin{align}\label{eq:spatial recovery I1}
    I_1
    &\lesssim
    \int_{\omega\in[-h^{-1},h^{-1}]^d}
    |\hatfft(\omega)-\fft(\omega)|^2\,d\omega \nonumber\\
    &\lesssim
    \int_{\omega\in[-h^{-1},h^{-1}]^d}
    \left(\rho_\sam h^{-\beta-1}\right)^2\,d\omega
    \lesssim
    \rho_\sam^2 h^{-d-2\beta-2},
\end{align}
provided that $\frac{1}{c^2}(1+\sqrt{d}h^{-1})^{2\beta+1}\rho_\sam\leq \frac{1}{2}$ is satisfied.

For $I_2$, Assumption \ref{assump:standing}(ii) gives
\begin{align*}
    I_2
    &\lesssim
    \int_{\R^d}
    (1+|\omega|_2)^{-2\beta}
    |\Kft(h\omega)-1|^2\,d\omega \\
    &\lesssim
    \int_{|\omega|_2\le h^{-1}}
    (1+|\omega|_2)^{-2\beta}
    |\Kft(h\omega)-1|^2\,d\omega
    +
    \int_{|\omega|_2>h^{-1}}
    (1+|\omega|_2)^{-2\beta}\,d\omega \\
    &=: E_1+E_2.
\end{align*}
For $E_1$, Taylor's theorem and Assumption \ref{assumption:kernel}(ii) imply that the derivatives of $\Kft$ at the origin up to order $p-1$ vanish. Moreover, Assumption \ref{assumption:kernel}(iv) implies that, for every multi-index $\nu$ with $|\nu|=p$,
\[
    D^\nu\Kft(\xi)
    =
    (-i)^{|\nu|}
    \int_{\R^d} t^\nu \kerneldec(t)e^{-it\cdot \xi}\,dt
\]
is uniformly bounded in $\xi$. Therefore,
\[
    |\Kft(h\omega)-1|
    \lesssim
    h^p|\omega|_2^p.
\]
Using $p>\beta-d/2$, we obtain
\begin{align*}
    E_1
    &\lesssim
    \int_{|\omega|_2\le h^{-1}}
    (1+|\omega|_2)^{-2\beta}
    h^{2p}|\omega|_2^{2p}\,d\omega \\
    &\lesssim
    h^{2p}
    \int_0^{h^{-1}}
    r^{2p-2\beta+d-1}\,dr
    \lesssim
    h^{2\beta-d}.
\end{align*}
Similarly, since $\beta>d/2$,
\begin{align*}
    E_2
    \lesssim
    \int_{h^{-1}}^\infty r^{-2\beta+d-1}\,dr
    \lesssim
    h^{2\beta-d}.
\end{align*}
Combining these bounds with \eqref{eq:spatial recovery I1}, we obtain
\[
    \|\widehat{f}-f\|_2^2
    \lesssim
    \rho_\sam^2 h^{-d-2\beta-2}
    +
    h^{2\beta-d}.
\]
Taking square roots gives
\[
    \|\widehat{f}-f\|_2
    \lesssim
    \rho_\sam h^{-d/2-\beta-1}
    +
    h^{\beta-d/2}.
\]
Setting $h\asymp \rho_\sam^{1/(2\beta+1)}$ with the implicit constant sufficiently large ensures that $\rho_N\leq \frac{c^2}{2}(1+\sqrt{d}h^{-1})^{-2\beta-1}$ is satisfied, and when $N$ is large enough so that $h\in(0,1)$, we get
\[
    \|\widehat{f}-f\|_2
    \lesssim
    \rho_\sam^{\frac{\beta-d/2}{2\beta+1}}.
\]
This completes the proof.
\end{proof}

\begin{example}[Signal-to-noise ratio under squared exponential noise]\label{ex:signaltonoise}
Suppose, as in Example \ref{ex:squaredexponential}, that the noise $\epsilon(t)$ is a centered Gaussian process with squared exponential covariance function. Assume further that we are in the high-noise regime where $\ell^{1/2}\lesssim \sigma^2$ and $\sigma\gtrsim 1.$ Then, Theorem \ref{thm:ErrorBoundDeconvolutionEstimator} along with the scaling of $\rho_N$ in Example \ref{ex:squaredexponential} imply that, for fixed confidence parameter $\delta\in(0,e^{-1}),$ 
\begin{equation}\label{eq:scaling}
     \|\hatf-f\|_2 \lesssim  \Biggl(
\frac{\ell^{(d-1)/2}\sigma^3}{\sqrt{\sam}}\Biggr)^{\frac{\beta-d/2}{2\beta+1}}\, .
\end{equation}
 Hence, $N \gtrsim \ell^{d-1} \sigma^6$ suffices for signal recovery in the high-noise regime.  

Since throughout our analysis we take the signal to be order one, the quantity 
$$  \mathsf{SNR} = \frac{1}{\ell^{(d-1)/3}\sigma^2}$$
can be viewed as an effective signal-to-noise ratio (SNR) in the high-noise regime. This effective SNR incorporates both the noise level and the correlation lengthscale. Using this terminology, Theorem \ref{thm:ErrorBoundDeconvolutionEstimator} shows that $N \gtrsim \mathsf{SNR}^{-3}$ suffices for signal recovery.
\end{example}

\begin{rmk}
Note that the assumption $\sup_{\omega \in \R^d} \frac{|\nabla \fft(\omega)|_2}{|\fft(\omega)|}  \leq C_\nabla$ is not used in the analysis of Algorithm \ref{alg:mainalgFM}. Without this assumption, Algorithm \ref{alg:mainalg} still achieves signal recovery, albeit with a slightly reduced convergence rate. 
Indeed, the stability guarantee for Algorithm \ref{alg:mainalg} in Proposition \ref{prop:FTDeviation} becomes
    \begin{align*}
        \sup_{\omega \in [-h^{-1},h^{-1}]^d}
        \left|\hatfft(\omega) - \fft(\omega) \right|
        \lesssim h^{-2\beta-1}\epsilon \, ,
    \end{align*}
    and the recovery guarantee in Theorem \ref{thm:ErrorBoundDeconvolutionEstimator} becomes
    \begin{align*}
       \|\hatf-f\|_2 \lesssim h^{-d/2-2\beta-1}\rho_\sam+ h^{\beta-d/2} \lesssim  \rho_\sam^{\frac{\beta - d/2}{3 \beta + 1}}
    \end{align*} 
    for $h= h^* \asymp \rho_\sam^{\frac{1}{3\beta + 1}}$. 
\end{rmk}

\subsection{Concentration of Bispectrum Estimator} \label{subsec:bispectrum_conc}
In this subsection, we analyze the concentration of the bispectrum estimator \eqref{equ:empirical_estimator} considered in Proposition \ref{prop:bispectrumconcentration}. 
We first control the error of the bispectrum estimator $\widehat{\Bf}$ by the error of the third-order autocorrelation estimator $\widehat{A_3f}$ in Lemma \ref{lemma:bispectrum < A3f}.
We then decompose the error on $\widehat{A_3f}$ into three terms in Lemma \ref{lemma:error decomp}, and finally prove a concentration result for $\widehat{A_3f}$ in Lemma \ref{lem:concentration of A3}. 
Recall that, since \(f\) is supported on \([-\pi,\pi]^d\), the third-order
autocorrelation \(A_3f\) is supported in
$
    Z=[-2\pi,2\pi]^{2d}.
$

\begin{lem}\label{lemma:bispectrum < A3f}
Under Assumption \ref{assump:standing}, there is a constant $C_Z$ depending only on $Z$ such that for any $w_1,w_2$
\begin{align*}
    |\widehat{\Bf}(w_1,w_2) - \Bf(w_1,w_2)| &\leq C_Z\underset{(z_1,z_2)\in Z}{\operatorname{sup}} 
    \left|\widehat{A_3f}(z_1,z_2)-A_3f(z_1,z_2)\right|\\
    \left|\nabla_1 \widehat{\Bf} (w_1,w_2)- \nabla_1 \Bf (w_1,w_2)\right|  & \leq C_Z \underset{(z_1,z_2)\in Z}{\operatorname{sup}} 
    \left|\widehat{A_3f}(z_1,z_2)-A_3f(z_1,z_2)\right|.
\end{align*}
\end{lem}
\begin{proof}
Using the relation \eqref{eq:bispectrum and A3}, we have 
\begin{align*}
    |\widehat{\Bf}(w_1,w_2) - \Bf(w_1,w_2)|& =\left| (\widehat{A_3f} )^{ft}(w_1,w_2)-(A_3f)^{ft}(w_1,w_2)\right| \\
    & \leq \left|\int_{Z} \left[\widehat{A_3f}(z_1,z_2)-A_3f(z_1,z_2)\right]e^{-iw_1\cdot z_1-iw_2\cdot z_2}dz_1dz_2\right|\\
    & \leq |Z| \underset{(z_1,z_2)\in Z}{\operatorname{sup}} 
    \left|\widehat{A_3f}(z_1,z_2)-A_3f(z_1,z_2)\right|, 
\end{align*}
where we have used the fact that $Z$ is a bounded domain in the last step. 
Similarly,
\begin{align*}
    \left|\nabla_1\widehat{ \Bf} (w_1,w_2)- \nabla_1 \Bf (w_1,w_2)\right|& = \left|\int_Z -iz_1\left[\widehat{A_3f}(z_1,z_2)-A_3f(z_1,z_2)\right]e^{-iw_1\cdot z_1-iw_2\cdot z_2}dz_1dz_2\right|\\
    &\leq C_Z \underset{(z_1,z_2)\in Z}{\operatorname{sup}} 
    \left|\widehat{A_3f}(z_1,z_2)-A_3f(z_1,z_2)\right|. 
\end{align*}
\end{proof}

We next seek to control the estimation error for the third-order correlation of $f$:
\begin{align}\label{eq:error of A3f}
    \underset{(z_1,z_2)\in Z}{\operatorname{sup}} 
    \left|\widehat{A_3f}(z_1,z_2)-A_3f(z_1,z_2)\right| \, .
\end{align}
The next lemma shows that this error can be decomposed into several components, each involving weighted autocorrelation functions of $f$ of different orders. 
Specifically, the statement uses the weighted second- and first-order quantities defined by
\begin{align}\label{equ:weighted_A2_A1_def}
    A_2[g,w](z) &:=  \int_{[-\dom,\dom]^d} g(t)g(t+z)w(t)\, dt \, , \quad A_1[g,w]:=   \int_{[-\dom,\dom]^d}g(t)w(t)\, dt \, . 
\end{align}

\begin{lem}\label{lemma:error decomp}
Under Assumption \ref{assump:standing},
for each \((z_1,z_2)\in Z\), the estimation error satisfies
\begin{align}\label{eq:concentration error}
    |A_3f(z_1,z_2) - \widehat{A_3f}(z_1,z_2)| &\leq \left|\text{Error}_1\right|+\left|\text{Error}_2\right|+\left|\text{Error}_3\right|,  
\end{align}
where
\begin{align*}
\text{Error}_1 &:= \frac{1}{\sam}\Big(A_1[\epsilon(t-z_2),F(t)F(t+z_1)] + A_1[\epsilon(t+z_1),F(t)F(t-z_2)] +A_1[\epsilon(t),F(t+z_1)F(t-z_2)]\Big)\, , \\
 \text{Error}_2 &:= \left(\frac{1}{\sam}A_2[\epsilon(t), F(t+z_1)](-z_2) - \noisecov(z_2)\right) + \left(\frac{1}{\sam}A_2[\epsilon(t),F(t-z_2)](z_1)-\noisecov(z_1)\right) \\
 &\qquad + \left(\frac{1}{\sam}A_2[\epsilon(t+z_1),F(t)](-z_2-z_1) -\noisecov(z_1+z_2) \right),\\
 \text{Error}_3 &:= \frac{A_3\epsilon(z_1,z_2)}{\sam} .
\end{align*}
\end{lem}
\begin{proof}
Recall our model for $\obs$ in \eqref{eq:M=F+noise}. 
A straightforward calculation then gives:
\begin{align*}
    A_3 \obs(z_1,z_2) &= A_3F(z_1,z_2) + A_3\epsilon(z_1,z_2) \\
    &\quad +A_2[\epsilon(t), F(t+z_1)](-z_2) + A_2[\epsilon(t),F(t-z_2)](z_1) + A_2[\epsilon(t+z_1),F(t)](-z_2-z_1) \\
    &\quad +A_1[\epsilon(t-z_2),F(t)F(t+z_1)] + A_1[\epsilon(t+z_1),F(t)F(t-z_2)] + A_1[\epsilon(t),F(t+z_1)F(t-z_2)] .
\end{align*}   
Since $A_3F=\sam A_3f$, from the definition of $\widehat{A_3f}$ in \eqref{equ:empirical_estimator} we obtain  
\begin{align*}
    \widehat{A_3f}(z_1,z_2)-A_3f(z_1,z_2) = \frac{A_3 \obs(z_1,z_2)}{\sam} - \noisecov(z_2)- \noisecov(z_1)-\noisecov(z_1+z_2) - \frac{A_3F(z_1,z_2)}{\sam}\, .
\end{align*}
Substituting the expansion for $A_3Y$ into the above and rearranging the terms proves the claim. 
\end{proof}

By bounding each error component, we arrive at the following concentration result. 

\begin{lem}[Concentration of $\widehat{A_3f}$]\label{lem:concentration of A3}
Under Assumption \ref{assump:standing}, with probability at least $1-\delta$, we have
\begin{small}
\begin{align*}
    &\underset{z\in [-2\pi,2\pi]^{2d}}{\operatorname{sup}}\, |A_3f(z)-\widehat{A_3f}(z)|\\
&\lesssim (-\log \delta)^{3/2}\sqrt{\frac{\boundcov^2\|\nabla\noisecov\|_1+\boundcov \Lipcov\|\noisecov\|_1}{\sam}}+(-\log \delta)\sqrt{\frac{\Lipcov\|f\|_2^2 \|\noisecov\|_1}{\sam}}+(-\log \delta)^{1/2}\sqrt{\frac{\Lipf^2 \|f\|_\infty^2 \|\noisecov\|_1 }{\sam}}.
\end{align*}
\end{small}
Here we recall that $\Lipf$ is the Lipschitz constant of $f.$
\end{lem}

\begin{proof}
We control the three error terms in \eqref{eq:concentration error}. 
We present the full details of bounding $\text{Error}_3$ and the other two can be bounded using similar arguments.
By Lemma \ref{lemma:moments gaussian chaos}, there exists a constant $a$ such that the rescaled random field  
\begin{align*}
    X(z) = \frac{a\xxx(z)}{\sqrt{\boundcov^2\|\nabla\noisecov\|_{1}+\boundcov \Lipcov\|\noisecov\|_{1}}}
\end{align*}
is a sub-$3$rd Gaussian chaos field (see Definition \ref{def:Gaussian_chaos}), where
\[ \xxx(z) = \frac{1}{\sqrt{\sam}}\int_{[-\dom,\dom]^d} \epsilon(t)\epsilon(t-z_1)\epsilon(t-z_2)dt = \frac{1}{\sqrt{\sam}}A_3\epsilon\]
is from Definition \ref{def:error_decomp}, which gives
\begin{align*}
    \frac{1}{\sam}A_3\epsilon = \sqrt{\frac{\boundcov^2\|\nabla\noisecov\|_{1}+\boundcov \Lipcov\|\noisecov\|_{1}}{\sam a^2}}X(z).
\end{align*}
By \cite[Corollary 3.3]{viens2007supremum} applied to $X$, we obtain
\begin{align*}
    \E\, \underset{z\in [-2\pi,2\pi]^{2d}}{\operatorname{sup}}\, X(z) \lesssim \int_0^{\operatorname{diam}([-2\pi,2\pi]^{2d})} \log [\mathcal{N}(\ell,[-2\pi,2\pi]^{2d},s)]^{3/2} d\ell ,
\end{align*}
where $\mathcal{N}(\ell,[-2\pi,2\pi]^{2d},s)$ is the smallest number of $\ell$ balls, with respect to the metric $s(z,w)=\sqrt{\E|X(z)-X(w)|^2},$ needed to cover $[-2\pi,2\pi]^{2d}$. 
By Lemma \ref{lemma:bound on canonical distance}, we have $s(z,w)\leq a|z-w|_2^{1/2}$ for some constant $a$, so that 
\begin{align*}
    \mathcal{N}(\ell,[-2\pi,2\pi]^{2d},s) \lesssim \mathcal{N}((\ell/a)^2,[-2\pi,2\pi]^{2d},|\cdot|_2) \lesssim  (\ell^2)^{-2d}.
\end{align*}
Therefore we have 
\begin{align*}
    \E \underset{z\in [-2\pi,2\pi]^{2d}}{\operatorname{sup}}\, X(z) \lesssim \int_0^{\operatorname{diam}([-2\pi,2\pi]^{2d})} (-\log \ell)^{3/2}d\ell <\infty. 
\end{align*}
Furthermore, \cite[Theorem 3.4]{viens2007supremum} implies that  
\begin{align*}
    \P\left\{\left|\underset{z\in [-2\pi,2\pi]^{2d}}{\operatorname{sup}}\,X(z)-\E \underset{z\in [-2\pi,2\pi]^{2d}}{\operatorname{sup}}X(z)\right|>u\right\}\leq 2\exp\left(-c'u^{2/3}\right)
\end{align*}
for some constant $c'$,
so that with probability $1-\delta/3$, we have 
\begin{align*}
    \underset{z\in [-2\pi,2\pi]^{2d}}{\operatorname{sup}}\,X(z) \leq \E \underset{z\in [-2\pi,2\pi]^{2d}}{\operatorname{sup}}\,X(z) + c(-\log \delta)^{3/2}\lesssim (-\log \delta)^{3/2}.
\end{align*}
The same argument applied to $-X$ gives the same bound on the supremum of $|X(z)|$.
Therefore, with probability $1-\delta/3$, 
\begin{align*}
    \underset{z\in [-2\pi,2\pi]^{2d}}{\operatorname{sup}}\,|\text{Error}_3(z)|&=\underset{z\in [-2\pi,2\pi]^{2d}}{\operatorname{sup}}\, \frac{1}{\sam}|A_3\epsilon(z_1,z_2)|  \\
    &= \sqrt{\frac{\boundcov^2\|\nabla\noisecov\|_{1}+\boundcov \Lipcov\|\noisecov\|_{1}}{ \sam a^2}}\underset{z\in [-2\pi,2\pi]^{2d}}{\operatorname{sup}}\,  |X(z)| \\
    &\lesssim (-\log \delta)^{3/2}\sqrt{\frac{\boundcov^2\|\nabla\noisecov\|_{1}+\boundcov \Lipcov\|\noisecov\|_{1}}{\sam}}.
\end{align*}
By similar arguments, Lemma \ref{lemma:bound on canonical distance} yields the next two bounds, each with probability at least $1 - \delta/3$: 
\begin{align*}
    \underset{z\in[-2\pi,2\pi]^{2d}}{\operatorname{sup}} |\text{Error}_2(z)| &\lesssim (-\log \delta)\sqrt{\frac{\Lipcov\fLtwo^2 \covLone}{\sam}},\\
    \underset{z\in[-2\pi,2\pi]^{2d}}{\operatorname{sup}}|\text{Error}_1(z)|&\lesssim (-\log \delta)^{1/2}\sqrt{\frac{\Lipf^2 \fsup^2 \|\noisecov\|_{1}}{\sam}}.
\end{align*}
The result follows by combining the bounds above. 
\end{proof}

\subsection{Stability of Bispectrum Inversion}\label{subsec:stability_of_bisp_inv} In this section, we prove stability results for bispectrum inversion via Algorithms \ref{alg:mainalgFM} and \ref{alg:mainalg} stated in Proposition \ref{prop:FTDeviation}. 
Note the stability for Algorithm \ref{alg:mainalg} essentially follows from Theorem 3.4 in \cite{fMRA}, but for completeness we rederive the result in our setting, which accounts for our use of the unregularized bispectrum estimator as well as Assumption \ref{assump:standing}(ii). 
 
\begin{proof}[Proof of Proposition \ref{prop:FTDeviation}]

\textbf{Algorithm \ref{alg:mainalgFM}:}
We shall first present the proof for $d=1$.    
In this case, Algorithm \ref{alg:mainalgFM} first defines $\widehat{\fft}(0)=1$ and $\widehat{\fft}(1) = |\widehat{ \Bf}(1,1)|^{\frac{1}{2}},$ and then defines $\widehat{\fft}(k)$ for $k=2,\ldots, K$ by the recursive formula
    \[\widehat{\fft}(k) = \frac{\widehat{\Bf}(k,k-1)}{\widehat{\fft}\bigl(-(k-1)\bigr)\widehat{\fft}(-1)} \, .\]
By expressing $\hatfft\bigl(-(k-1)\bigr)$ again using the recursion formula, we obtain 
\begin{align}
    \tag{$k$ odd}\widehat{\fft}(k) &= \frac{\widehat{\Bf}(k,k-1)\cdot \widehat{\Bf}(k-2,k-3)\cdots \widehat{\Bf}(3,2)\cdot |\widehat{\Bf}(1,1)|^{\frac{1}{2}}}{\overline{\widehat{\Bf}(k-1,k-2)}\cdot\overline{\widehat{\Bf}(k-3,k-4)} \cdots \overline{\widehat{\Bf}(2,1)}} \\
    \tag{$k$ even}\widehat{\fft}(k)& = \frac{\widehat{\Bf}(k,k-1)\cdot \widehat{\Bf}(k-2,k-3)\cdots \widehat{\Bf}(2,1)}
{\overline{\widehat{\Bf}(k-1,k-2)} \cdot \overline{\widehat{\Bf}(k-3,k-4)}\cdots \overline{\widehat{\Bf}(3,2)} \cdot |\widehat{\Bf}(1,1)|}.
\end{align}
Since the quantities involved are complex-valued, we interpret the logarithmic
differences below as logarithms of ratios close to one. More precisely, by Assumption \ref{assump:standing}(ii) and the assumption that $\frac{2^{\beta}(1+K)^{2\beta+1}}{c^3}\epsilon \leq \frac12$ we have
\begin{equation}\label{equ:term_bound}
    \frac{|\widehat{\Bf}(j,j-1)-\Bf(j,j-1)|}{|\Bf(j,j-1)|}
    \le \frac{2^{\beta}(1+j)^{2\beta}}{c^3}\epsilon \leq \frac{2^{\beta}(1+K)^{2\beta}}{c^3}\epsilon \leq \frac12 \quad \text{for} \quad j=2,\ldots, K,
\end{equation}
and so we can write
\[
    \log \widehat{\Bf}(j,j-1)-\log \Bf(j,j-1)
    :=
    \log\left(
    1+
    \frac{\widehat{\Bf}(j,j-1)-\Bf(j,j-1)}
    {\Bf(j,j-1)}
    \right),
\]
where the logarithm on the right-hand side is taken using the principal branch. This is
well defined because the argument lies in the ball \(B(1,1/2)\), which does not
intersect the negative real axis. For the initialization,
note that Assumption \ref{assump:standing}(ii) and (iii) imply that $c \le 1,$ and hence the assumption $\frac{2^{\beta}(1+K)^{2\beta+1}}{c^3}\epsilon \leq \frac12$ for $K>1$
ensures that $
    \frac{4^\beta}{c^2}\epsilon \leq \frac12.$
Hence, we similarly have
\[
    \frac{| |\widehat{\Bf}(1,1)|^{\frac{1}{2}}-\Bf(1,1)^{\frac{1}{2}}|}{\Bf(1,1)^{\frac{1}{2}}} \leq \frac{||\widehat{\Bf}(1,1)| - \Bf(1,1)|}{\Bf(1,1)}
    \le \frac{4^\beta}{c^2}\epsilon \le \frac12 \, ,
\]
and $\log |\widehat{\Bf}(1,1)|^{\frac12}-\log \Bf(1,1)^{\frac12}$ is also well defined. 
As a result, for \(k\) odd,
\begin{align}\label{eq:algorithm 1 stability recursion}
   \left|\log\left(\frac{\widehat{\fft}(k)}{\fft(k)}\right)\right|
   \leq \sum_{j=2}^{k}
   \left|\log\left(\frac{\widehat{\Bf}(j,j-1)}{\Bf(j,j-1)}\right)\right|
   +
   \left|
   \log\left(
   \frac{|\widehat{\Bf}(1,1)|^{1/2}}
        {\Bf(1,1)^{1/2}}
   \right)
   \right|.
\end{align}
Notice that for $|z|<1/2,$ 
\begin{align*}
    |\log (1+z)| = \left|\sum_{k=1}^\infty \frac{(-1)^{k+1}z^k}{k}\right| \leq \sum_{k=1}^\infty |z|^k  = \frac{|z|}{1-|z|} \leq 2|z|.  
\end{align*}
Therefore we have 
\begin{align*}
\left|\log \widehat{\Bf}(j,{j-1})-\log \Bf(j,{j-1})\right| &\leq 2\frac{|\widehat{\Bf}(j,j-1)-\Bf(j,j-1)|}{|\Bf(j,j-1)|}\leq \frac{2^{\beta+1}(1+j)^{2\beta}}{c^3}\epsilon,\qquad j=2,\ldots,K,
\end{align*}
and similarly
\[ \left|\log |\widehat{\Bf}(1,1)|^{\frac12} - \log \Bf(1,1)^{\frac12}\right| \leq 2 \frac{| |\widehat{\Bf}(1,1)|^{\frac{1}{2}}-\Bf(1,1)^{\frac{1}{2}}|}{\Bf(1,1)^{\frac{1}{2}}} \leq \frac{2\cdot 4^{\beta}}{c^2}\epsilon \leq \frac{2^{\beta+1}\cdot 2^{2\beta}}{c^3}\epsilon\, . \]
We can thus bound the error on the log of the coefficients by:
\begin{align*}
  |\log\widehat{\fft}(k)-\log\fft(k)| &\leq  \epsilon\sum_{j=1}^{k} \frac{2^{\beta+1}(1+j)^{2\beta}}{c^3} \leq  \frac{2^{\beta+1}(1+k)^{2\beta+1}}{c^3} \epsilon.
\end{align*}
Now by the assumption $\frac{2^{\beta}}{c^3}(1+ K)^{2\beta+1} \epsilon \leq \frac12$ again, we have 
\begin{align*}
    \underset{k\in [-K,K]}{\operatorname{sup}} |\log\widehat{\fft}(k)-\log\fft(k)| \leq 1 \, .
\end{align*}
Then since $|e^z-1|\leq C|z|$ over $|z|\leq 1$ for some constant $C$ independent of $K$ and $\epsilon$, we obtain 
\begin{align*}
    \left|\frac{\widehat{\fft}(k)}{\fft(k)}-1\right| = \left|\exp\bigl(\log\widehat{\fft}(k)-\log\fft(k) \bigr)-1\right| \lesssim k^{2\beta+1}\epsilon,
\end{align*}
which further implies that 
\begin{align*}
    |\widehat{\fft}(k)-\fft(k)| \lesssim |\fft(k)|k^{2\beta+1}\epsilon\lesssim k^{\beta+1}\epsilon. 
\end{align*}
Taking the supremum over $k \in \mathbb{Z} \cap [-K,K]$ proves the claim.
The proof for the case when $k$ is even proceeds in the same manner except in \eqref{eq:algorithm 1 stability recursion} the last term would be $|\log |\widehat{\Bf}(1,1)|-\log \Bf(1,1)|$, which can be bounded with the same argument. 
Note a similar calculation holds in general dimension, with $|k|_1\leq dK$ replacing $k$ as the number of steps in the path, and also $\sqrt{d}K$ replacing $K$ as the bound for $|j|_2$ in \eqref{equ:term_bound}.

\textbf{Algorithm \ref{alg:mainalg}:}
We first notice that $\hatfft(\omega)= \fft(\omega) e^{\Delta(\omega)}$, where
\begin{align*}
    \Delta(\omega) = \int_0^1 \left[\frac{\nabla_1 \widehat{\Bf}(\alpha\omega,\alpha\omega)}{\widehat{\Bf}(\alpha\omega,\alpha\omega)}  - \frac{\nabla_1 \Bf(\alpha\omega,\alpha\omega)}{\Bf(\alpha\omega,\alpha\omega)}\right] \cdot \omega \,d\alpha. 
\end{align*}
Therefore, it suffices to bound the quantity $\Delta(\omega)$, which we decompose as 
\begin{small}
\begin{align*}
    \Delta(\omega) &= \int_0^1 \frac{\nabla_1 \widehat{\Bf}(\alpha\omega,\alpha\omega)-\nabla_1\Bf(\alpha\omega,\alpha\omega)}{\Bf(\alpha\omega,\alpha\omega)}\cdot \omega \, d\alpha \\
    &\quad + \int_0^1 \left[\frac{1}{\widehat{\Bf}(\alpha\omega,\alpha\omega)}-\frac{1}{\Bf(\alpha\omega,\alpha\omega)}\right] \nabla_1 \Bf(\alpha\omega,\alpha\omega)\cdot \omega\, d\alpha\\
    &\quad + \int_0^1 \left[\frac{1}{\widehat{\Bf}(\alpha\omega,\alpha\omega)}-\frac{1}{\Bf(\alpha\omega,\alpha\omega)}\right] \left[\nabla_1\widehat{\Bf}(\alpha\omega,\alpha\omega)-\nabla_1\Bf (\alpha\omega,\alpha\omega)\right]\cdot \omega\,d\alpha=: \Delta_1(\omega)+\Delta_2(\omega)+\Delta_3(\omega). 
\end{align*}
\end{small}

We will bound each $\Delta_i$ in turn. Before doing so, we obtain 
an estimate that will be used throughout. 
By the identity \eqref{equ:bispectrum_def} and Assumption \ref{assump:standing}, we have 
\begin{align*}
|\Bf(\alpha\omega,\alpha\omega)|= |\fft(\alpha\omega)|^2 \geq c^2 (1+\alpha|\omega|_2)^{-2\beta}\geq  c^2 (1+|\omega|_2)^{-2\beta},
\end{align*}
and 
\begin{align*}
    |\widehat{\Bf}(\alpha\omega,\alpha\omega)|&\geq |\Bf(\alpha\omega,\alpha\omega)| - |\widehat{\Bf}(\alpha\omega,\alpha\omega)-\Bf(\alpha\omega,\alpha\omega)|\\
    &\geq  |\Bf(\alpha\omega,\alpha\omega)| - \epsilon\gtrsim (1+|\omega|_2)^{-2\beta},
\end{align*}
where in the last step we have used the assumption that $\epsilon\leq \frac{c^2}{2}(1+\sqrt{d}h^{-1})^{-2\beta-1}$, which implies $\epsilon \leq \frac{c^2}{2} (1+|\omega|_2)^{-2\beta}$ for any $\omega\in [-h^{-1},h^{-1}]^d.$
Therefore, 
\[
    \min\{|\widehat{\Bf}(\alpha\omega,\alpha\omega)|,
    |\Bf(\alpha\omega,\alpha\omega)|\}
    \gtrsim (1+|\omega|_2)^{-2\beta}.
\]
Now fix $\omega \in [-h^{-1},h^{-1}]^d$. 
For $\Delta_1$, we have 
\begin{align}
    |\Delta_1(\omega)|&\leq |\omega|_2 \underset{0\leq \alpha\leq 1}{\operatorname{sup}}\,|\nabla_1 \widehat{\Bf}(\alpha\omega,\alpha\omega)-\nabla_1 \Bf(\alpha\omega,\alpha\omega)|_2 \cdot \underset{0\leq \alpha\leq 1}{\operatorname{sup}}\,\frac{1}{|\Bf(\alpha\omega,\alpha\omega)|} \nonumber\\
    &\lesssim |\omega|_2 \underset{z\in [-h^{-1},h^{-1}]^d}{\operatorname{sup}}\, |\nabla_1 \widehat{\Bf}(z,z)-\nabla_1 \Bf(z,z)|_2  \cdot \underset{0\leq \alpha\leq 1}{\operatorname{sup}}\,(1+\alpha |\omega|)^{2\beta} \nonumber\\
    &\leq |\omega|_2 \epsilon (1+|\omega|_2)^{2\beta}.  \label{eq:Delta_1}
\end{align}

For $\Delta_2$, we have 
\begin{align*}
    |\Delta_2(\omega)|& \leq \int_0^1 \left|\frac{\widehat{\Bf}(\alpha\omega,\alpha\omega)-\Bf(\alpha\omega,\alpha\omega)}{\widehat{\Bf}(\alpha\omega,\alpha\omega)}\right| |\omega|_2 \left|\frac{\nabla_1\Bf(\alpha\omega,\alpha\omega)}{\Bf(\alpha\omega,\alpha\omega)}\right|_2d\alpha \\
    &\leq |\omega|_2\underset{0\leq \alpha\leq 1}{\operatorname{sup}}\, \left|\frac{\widehat{\Bf}(\alpha\omega,\alpha\omega)-\Bf(\alpha\omega,\alpha\omega)}{\widehat{\Bf}(\alpha\omega,\alpha\omega)}\right| \cdot \underset{0\leq \alpha\leq 1}{\operatorname{sup}}\, \left|\frac{\nabla_1\Bf(\alpha\omega,\alpha\omega)}{\Bf(\alpha\omega,\alpha\omega)}\right|_2\\
    &\leq |\omega|_2 \epsilon (1+|\omega|_2)^{2\beta} \cdot \underset{0\leq \alpha\leq 1}{\operatorname{sup}}\, \left|\frac{\nabla_1\Bf(\alpha\omega,\alpha\omega)}{\Bf(\alpha\omega,\alpha\omega)}\right|_2.
\end{align*}
Using the identity \eqref{equ:bispectrum_def}, 
\begin{align*}
    \nabla_1 \Bf(w_1,w_2) &= \nabla \fft(w_1)\fft(-w_2)\fft(w_2-w_1)+ \fft(w_1)\fft(-w_2)\int (it)f(t)e^{i(w_2-w_1)t}dt,
\end{align*}
and so
\begin{align*}
    \frac{\nabla_1 \Bf(\alpha\omega,\alpha\omega)}{\Bf(\alpha\omega,\alpha\omega)} &= \frac{\nabla \fft(\alpha\omega)}{\fft(\alpha\omega)}+ \int (it)f(t)dt 
    = \frac{\nabla \fft(\alpha\omega)}{\fft(\alpha\omega)},
\end{align*}
where for the last equality we use that the assumption $\nabla\fft(0)=0$ implies that $\int tf(t)dt=0.$ 
By Assumption \ref{assump:standing}(ii), we have $|\frac{\nabla \fft(\alpha\omega)}{\fft(\alpha\omega)}|\leq C_{\nabla}$. 
Therefore,  
\begin{align}\label{eq:Delta_2}
|\Delta_2(\omega) |\leq  |\omega|_2 \epsilon (1+|\omega|_2)^{2\beta}.
\end{align}

For $\Delta_3$, we have 
\begin{align}\label{eq:Delta_3}
    |\Delta_3(\omega)|&\leq |\omega|_2 \underset{0\leq \alpha\leq 1}{\operatorname{sup}}\, \left|\frac{\widehat{\Bf}(\alpha\omega,\alpha\omega)-\Bf(\alpha\omega,\alpha\omega)}{\widehat{\Bf}(\alpha\omega,\alpha\omega)\Bf(\alpha\omega,\alpha\omega)}\right| \cdot \underset{0\leq \alpha\leq 1}{\operatorname{sup}}\,\left|\nabla_1\widehat{\Bf}(\alpha\omega,\alpha\omega)-\nabla_1\Bf(\alpha\omega,\alpha\omega)\right|\nonumber\\
    & \leq |\omega|_2 \epsilon^2 (1+|\omega|_2)^{4\beta}.
\end{align}
Since \(|\omega|_2\le \sqrt d\,h^{-1}\), the assumption on \(\epsilon\) implies
\[
\epsilon(1+|\omega|_2)^{2\beta+1}\lesssim 1.
\]
Combining this with \eqref{eq:Delta_1}, \eqref{eq:Delta_2}, \eqref{eq:Delta_3}, we obtain that 
\begin{align*}
    |\Delta(\omega)|\lesssim \epsilon |\omega|_2(1+|\omega|_2)^{2\beta}\lesssim 1, \qquad \forall \omega\in[-h^{-1},h^{-1}]^d.
\end{align*}
Since \(|\Delta(\omega)|\lesssim 1\), the bound
\(|e^z-1|\lesssim |z|\) for bounded complex \(z\) gives
\[
    |\hatfft(\omega)-\fft(\omega)|
    \leq |\fft(\omega)| |e^{\Delta(\omega)}-1|
    \lesssim |\fft(\omega)| |\Delta(\omega)|
    \lesssim
    (1+|\omega|_2)^{-\beta}
    \epsilon |\omega|_2(1+|\omega|_2)^{2\beta}
    \lesssim
    \epsilon(1+|\omega|_2)^{\beta+1}.
\]
Therefore, for \(\omega\in[-h^{-1},h^{-1}]^d\),
\[
    |\hatfft(\omega)-\fft(\omega)|
    \lesssim
    \epsilon h^{-\beta-1},
\]
and the desired result follows. 
\end{proof}

\section{Numerical Experiments}\label{sec:numerics}
In this section, we study the numerical performance of Algorithms \ref{alg:mainalgFM} and \ref{alg:mainalg} in a variety of experiments, including the recovery of the number of signal occurrences, bispectrum, and four signals that explicitly satisfy Assumption \ref{assump:standing}. Moreover, through the dependence of our numerical results on model parameters, we reinforce the validity of Proposition \ref{prop:bispectrumconcentration} and Theorem \ref{thm:ErrorBoundDeconvolutionEstimator}. Matlab code to reproduce all numerical experiments is publicly available at \url{https://github.com/msween11/fMTD/}. For computational simplicity, we only consider the one-dimensional MTD problem in our numerical experiments, although our theory holds in arbitrary dimension. All simulations in this section were conducted 20 times and averaged, and we report the average result along with error bars denoting one sample standard deviation across trials.

\subsection{Setting}\label{subsec:num_setup}

\subsubsection{Construction of the Observation $ \obs $}
In a slight break with Assumption \ref{assump:standing}, we assume that our hidden signal $ f $ is compactly supported on $ D:= [-1,1] $ in space. To construct the MTD observation $ \obs $ with $ N $ signal occurrences in the functional setting, we first initialize our hidden signal $ f $ as a MATLAB function handle. We then sample $ N $ i.i.d. $ \text{Unif}(0, 1/2) $ random variables $\{z_j\}_{j=1}^N$ and define the shifts
\begin{align*}
    x_1 = 0,\,x_j = x_{j-1} \pm 4 \pm z_j,\quad \, j =2,\ldots,N/2,
\end{align*}
where the shifts expand in both the positive and negative directions from 0. This construction ensures that the shifts are well-separated and and do not lie on any regular grid. The noise in our experiments was chosen to be a stationary Gaussian process with a squared-exponential covariance function defined in \eqref{equ:squared_exp_kernel}. We denote the lengthscale parameter $ \lambda $ and noise intensity parameter $ \sigma$. See Example \ref{ex:squaredexponential} for more details. Like the hidden signal, the noise is constructed as a MATLAB function handle that can be evaluated on any grid to produce a vector. 
To produce a discretization of our observation $ \obs $, we sample the interval $ [-2.5\cdot(N+1)-3,2.5\cdot(N+1)+3] $ at a rate of $ 2^{-5} $ . On this grid we then evaluate the sum of the noise and the $ N $ shifted function handles $ t\to f(t-x_j)$. Note that this discretization occurs \emph{after} the shifting, preserving the functional formulation of the problem. The empirical third-order autocorrelation function $ A_3\obs $ is then computed from the discretization of $ Y $ via matrix products. 

\subsubsection{Signal Function Choice}
We investigate our algorithms for four choices of functions as the hidden signal $ f $. Each of these functions is a compactly supported Wendland radial basis function \cite{wendland} $ \phi_{d,k} $, defined as
\[
  \phi_{d,k}(r) \;\propto\; \mathcal{I}^k\!\left[(1-r)_+^l\right], \qquad
  l = \left\lfloor\tfrac{d}{2}\right\rfloor + k + 1, \qquad
  \mathcal{I}[\psi](r) = \int_r^\infty t\,  \psi  (t)\,dt,
\]
and normalized so that $\phi_{d,k}(0)=1$. Thus for any $d,k,\, \|\phi_{d,k}\|_\infty = 1$. For our candidate functions, we take $r=|x|$ to get the following functions:
\begin{equation}\label{eq:signals}
	f_1 = \phi_{2,0}, \qquad  f_2 = \phi_{1,1}, \qquad f_3 = \phi_{1,2}, \qquad f_4 = \phi_{1,3}.
\end{equation}
Respectively, the functions $f_1,f_2,f_3,f_4$ are of smoothness class $C^0$, $C^2$, $C^4$, $C^6$, with Fourier decay $|f\ft(\omega)|\sim|\omega|^{-\beta}$ for $\beta=2,4,6,8$. These functions are compactly supported on $ [-1,1] $ in space, never vanish in frequency, and explicitly satisfy Assumption \ref{assump:standing}. 

\subsubsection{Algorithm Extensions}\label{ssec:algextensions}

For implementation of Algorithm \ref{alg:mainalg}, in order to increase the fidelity of the approximation of the integrals in Steps 3 and 4 of the algorithm, the discretized $ A_3\obs $ matrix is then zero-padded before Fourier transforming $ A_3\obs $ to estimate the bispectrum. As Algorithm \ref{alg:mainalgFM} operates only on the integer frequencies, the discretized $ A_3\obs $ matrix fed into this algorithm is not zero-padded before mapping into frequency.
For all experiments, we choose the infinite order deconvolution kernel $\kerneldec(x) = \text{sinc}(x)$. Additionally, a regularization step with constant $r = 10^{-3}$, was introduced, as in \cite{kurisu2022uniform} and \cite{li1998nonparametric}, to improve the performance of Algorithms \ref{alg:mainalgFM} and \ref{alg:mainalg}. In both algorithms, $\widehat{\Bf}$ was replaced with its regularized version: 
\[
\widetilde{\Bf}(\omega_1,\omega_2) := \frac{\widehat{\Bf}(\omega_1,\omega_2)}{1 \land \left(r\cdot \sqrt{\sam}|\widehat{\Bf}(\omega_1,\omega_2)|\right)} \, .
\]
However, for Algorithm \ref{alg:mainalg} the numerator was kept as $\widehat{\Bf}$; only the denominator was replaced with $\widetilde{\Bf}$. For both algorithms, this regularization step helps stabilize the construction of $\widehat\fft$ and improves performance across all signals and parameter regimes considered below.

Moreover, the frequency marching approach of Algorithm \ref{alg:mainalgFM} is based on a recursive formula of the bispectrum, see \eqref{equ:FM_phase} and the discussion preceding the algorithm's pseudocode. On the phases, this recursive formula is $ \phi_{j+1} = \Psi_{j+1,j} + \phi_j + \phi_1 $. However, for a given $\phi_{j+1}$ this recursive formula is not unique; for $ j \geq 1 $, there are $ \lfloor (j+1)/2 \rfloor $ distinct recursive formulas: 
\begin{align}\label{eq:fm1p}
	\phi_{j+1} = \phi_\ell + \phi_{j+1-\ell} + \Psi_{\ell,j+1} \quad \text{for any } 1 \leq \ell \leq \lfloor (j+1)/2 \rfloor.
\end{align}
While in the noiseless case all of these recursive formulas are equivalent and thus redundant, this is not true for the noisy empirical bispectrum. Indeed, following the discussion in Section IV-A of \cite{bendory2017bispectrum}, one can average the recovered phases over all of the recursive formulas to reduce the noise: 
\begin{align}\label{eq:mpfm}
	\phi_{j+1} = \text{phase}\left( \sum_{\ell=1}^{\lfloor (j+1)/2 \rfloor} e^{i(\phi_\ell + \phi_{j+1-\ell} + \Psi_{\ell,j+1})}\right).
\end{align}
Likewise, Algorithm \ref{alg:mainalg} also admits multiple integral formulas that are equivalent in the noiseless case. If we let $g(\omega)=\log \fft(\omega)$ for $f$ a real-valued signal, then by Lemma \ref{lem:kot_mp} the following recursion formula holds for any $0 \le \w_1 \le \w_2$:
\begin{align}\label{eq:kotmp}
	g(\w_2) = g(\w_1) +\int_0^{\w_2-\w_1} \partial_1 \log Bf(\xi+\w_1,\w_1)\ d\xi - \overline{g(\w_2-\w_1)}  .
\end{align}
Equation \eqref{eq:kotmp} gives a family of \emph{horizontal} integrals in the $\omega_1, \omega_2$ plane, starting on the diagonal point $(\w_1,\w_1)$ and ending at $(\w_2,\w_1)$. Numerically, if after discretization our frequency domain is of length $ J $, we can write our positive frequencies as $ \w_j, \, j = 1,\ldots,J/2 $ and denote $ g_\ell(j) $ as the recursion identity \eqref{eq:kotmp} with $ \w_2 = \w_j $ and $ \w_1 = \w_\ell $, for $ \ell = 1,\ldots,j $. We can then average to reduce noise:
\begin{align*}
	g(j):=\frac{1}{j} \sum_{\ell=1}^{j} g_\ell(j), \quad j = 1,\ldots,J/2.
\end{align*}
We perform our subsequent experiments with both the `single-path' formulas stated in Algorithms \ref{alg:mainalgFM} and \ref{alg:mainalg} and the averaged `multi-path' formulas \eqref{eq:mpfm} and \eqref{eq:kotmp}. It will be clearly indicated which version of Algorithms \ref{alg:mainalgFM} and \ref{alg:mainalg} are implemented.

Additionally, we note that the correctly unbiased second-order empirical autocorrelation function $ A_2\obs $ is generally a better estimator of $ A_2f $ than $ A_3\obs $ is of $ A_3 f $, as it contains only two copies of the noisy observation $ \obs $, rather than three. Moreover, the power spectrum of $ f $ can be recovered easily from $ A_2\obs $. Leveraging these facts, for all subsequent experiments, for whichever algorithm and recursion formula we implement, after obtaining our estimated Fourier transform $ \widehat{\fft} $, we then extract the phase information and update the magnitudes of $ \widehat{\fft} $ with the power spectrum recovered from $ A_2\obs $. This reduces recovery errors in both space and frequency for all the signals and  algorithms we consider.

\subsection{Bispectrum Estimation}
\begin{figure}[h]
	\centering
	\begin{subfigure}[t]{0.32\textwidth}
		\centering
		\includegraphics[width=\textwidth]{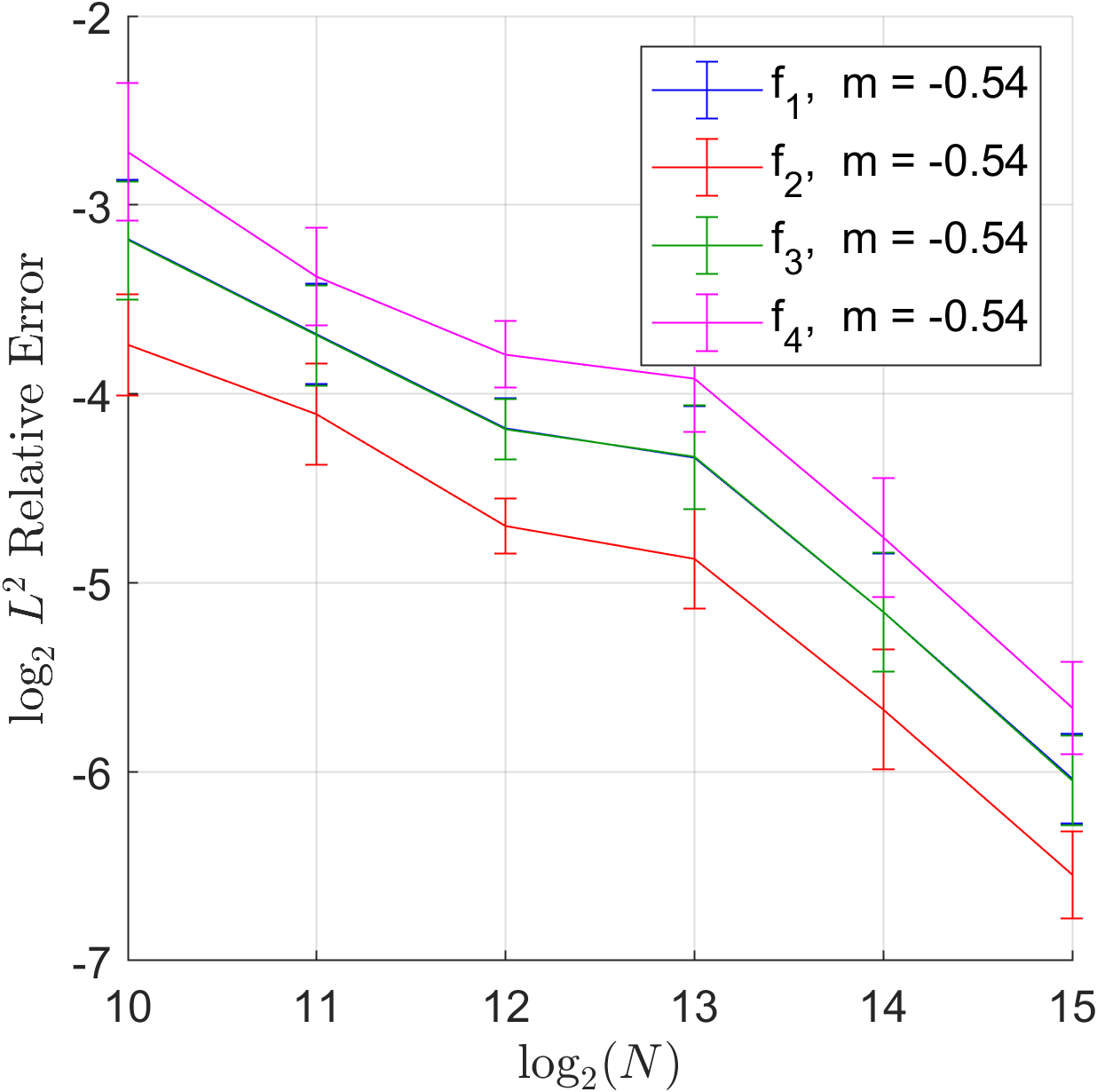}
		\caption{Estimating the number of signal occurrences $N$.}
		\label{fig:n_est}
	\end{subfigure}
	\hfill
	\begin{subfigure}[t]{0.32\textwidth}
		\centering
		\includegraphics[width=\textwidth]{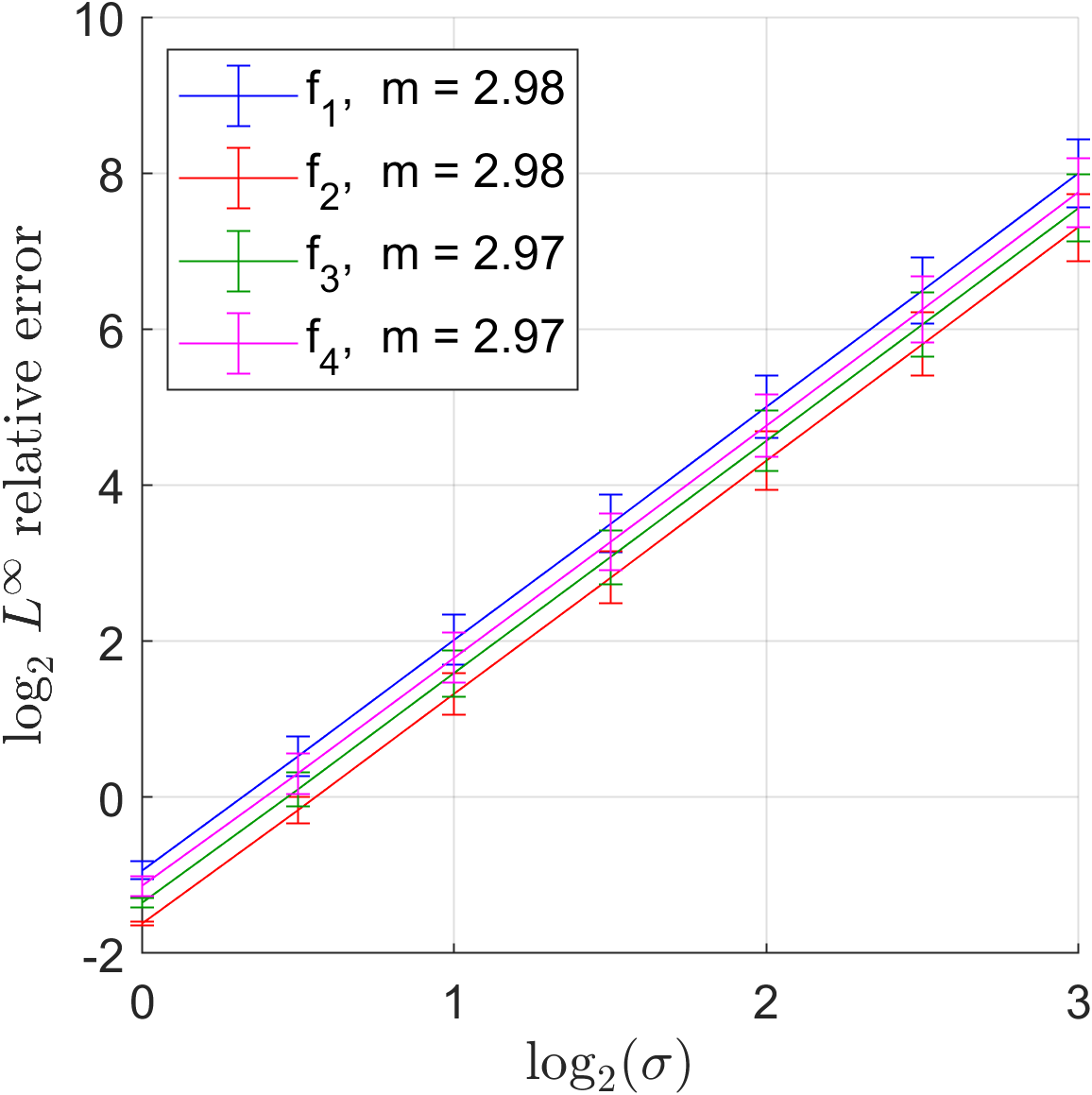}
		\caption{Estimating $ A_3f $.}
		\label{fig:a3y_est}
	\end{subfigure}
	\hfill
	\begin{subfigure}[t]{0.32\textwidth}
		\centering
		\includegraphics[width=\textwidth]{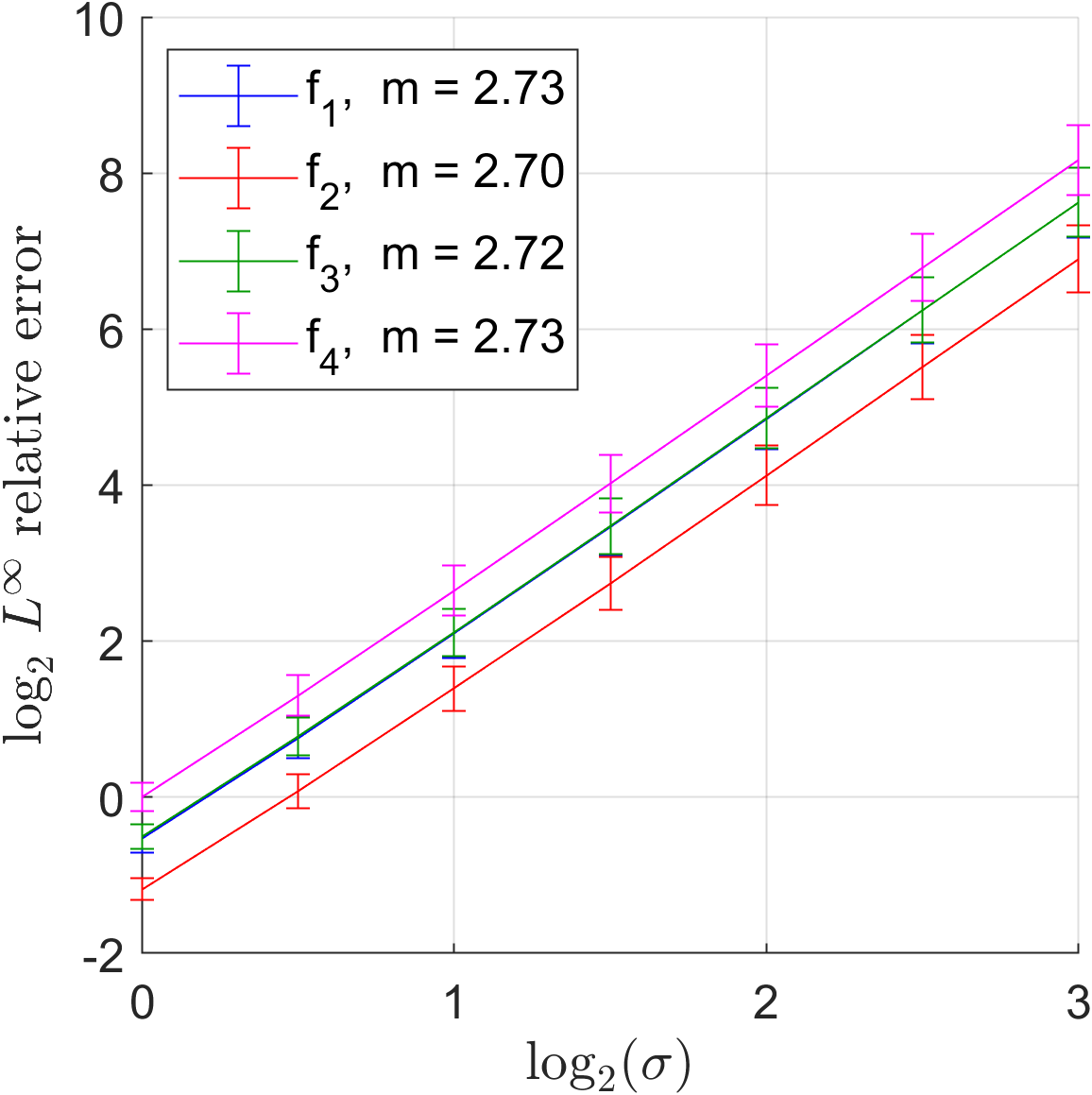}
		\caption{Estimating $ \Bf $.}
		\label{fig:subbispec_est}
	\end{subfigure}
	\caption{Three steps in estimating the bispectrum. The left panel plots error in estimating the true number of signal occurrences as a function of the number of occurrences. The middle and right panels plot error as a function of noise intensity $ \sigma $. Results are reported for all four hidden functions, with $m$ denoting the slope of the best linear fit for each function.}
	\label{fig:bispec_est}
\end{figure}

We first study the performance of our bispectrum estimator. We separately assess
the errors in estimating the number of signal occurrences and the third-order
autocorrelation, as well as the overall bispectrum estimation error.

Figure \ref{fig:n_est} shows the relative $ L^2 $ error in estimating the number of signal occurrences using the estimator proposed in Remark \ref{rmk:n_est} over a $ \log_2 $ evenly spaced range of sample sizes from $ 2^{10} $ to $ 2^{17} $. The noise parameters $ \sigma = 1, \lambda = .1 $ were held fixed for all sample sizes. Thus, all simulations were performed in the high noise regime. For all four functions, the relative error decreased monotonically as the number of signal occurrences increased. Moreover, for each function, the slope of the best fit line in log-log space was a constant $-0.54$. This matches well with the expected slope of $-0.5$ from the $N^{1/2}$ dependence on $\sam$ expected from extending the concentration in Proposition \ref{prop:bispectrumconcentration} to the $\sam$ estimator constructed in Remark \ref{rmk:n_est}. Going forward, to construct the estimator \eqref{equ:empirical_estimator}, one could first estimate the number of signal occurrences $\sam$ using the Remark \ref{rmk:n_est} estimator and use the estimated $\hat\sam$ in place of the true $\sam$. However, in order to control error dependence on model parameters, for all subsequent experiments we use the true number of signal occurrences $\sam$ in the construction of our bispectrum estimator. 

Figures \ref{fig:a3y_est} and \ref{fig:subbispec_est} show the relative $L^\infty$ error in estimating $A_3f$ and $\Bf$ using the estimators defined in \eqref{equ:empirical_estimator}. For both experiments, $\sam$ was held fixed at $2^{10}$ while the noise intensity $\sigma$ ranged evenly in $\log_2$ space from $2^0$ to $ 2^3 $. The bispectrum estimator was constructed without regularization but with zero-padding in space before mapping into frequency. The error curves for both experiments are almost perfectly linear for all functions, with slopes varying between $2.70$ to $2.73$ for the bispectrum estimation and $2.97$ to $2.98$ for the spatial autocorrelation estimation. This numerically validates the expected slope of 3 from the $\sigma^3$ dependence on $\sigma$ in Example \ref{ex:squaredexponential}, all other parameters being held constant. 

\subsection{Bispectrum Inversion} Now, we study the performance of our algorithms for inverting the bispectrum and recovering the hidden signal. In the following experiments, we emphasize the dependence on the model parameters $ \beta, N, \sigma $, as well as the spatial sampling rate. An example of a signal recovered using Algorithms \ref{alg:mainalgFM} and \ref{alg:mainalg} can be seen in Figure \ref{fig:introplotright}, where we recovered $f_2$ from a $\sigma=1$ observation with $N=2^{15}$ samples.

\subsubsection{Varying the Number of Signal Occurrences}
\begin{figure}[h]
    \centering
    \begin{subfigure}[b]{0.45\textwidth}
        \centering
        \includegraphics[width=\textwidth]{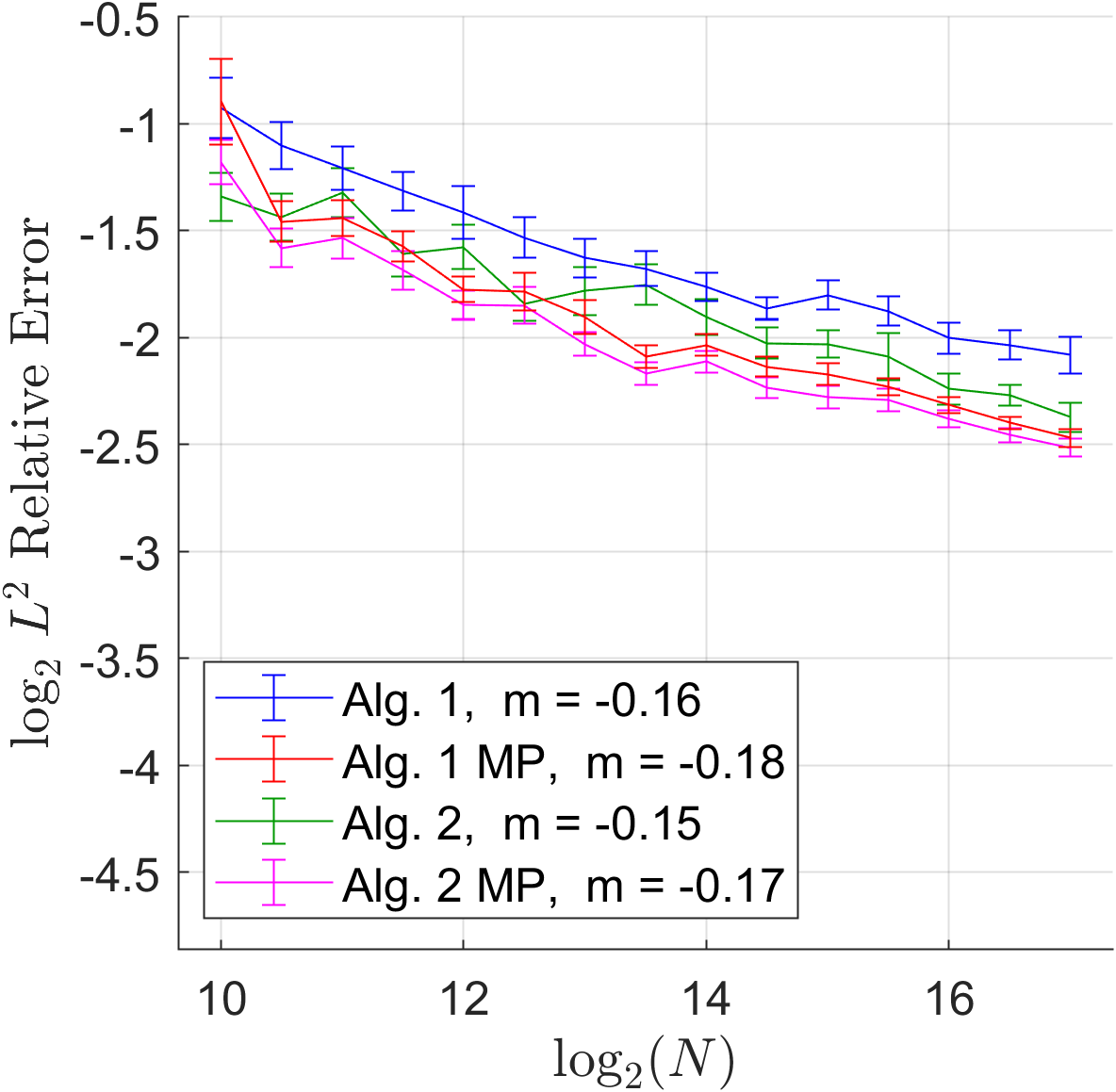}
        \caption{Recovery of $ f_1 $.}
        \label{fig:sub1}
    \end{subfigure}
    \hfill
    \begin{subfigure}[b]{0.45\textwidth}
        \centering
        \includegraphics[width=\textwidth]{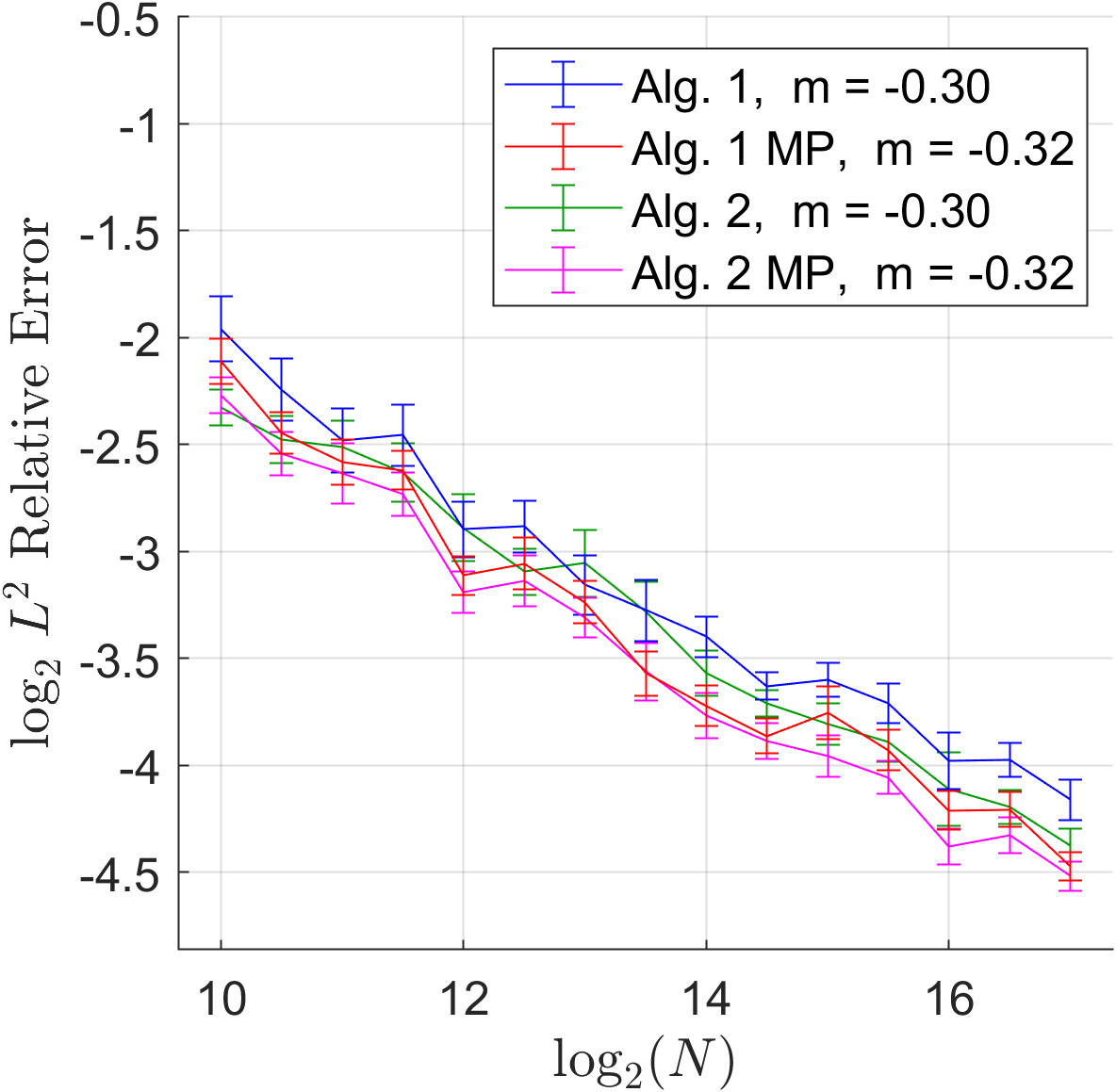}
		\caption{Recovery of $f_2$.}
		\label{fig:sub2}
    \end{subfigure}
    \\[1em]
    \begin{subfigure}[b]{0.45\textwidth}
        \centering
        \includegraphics[width=\textwidth]{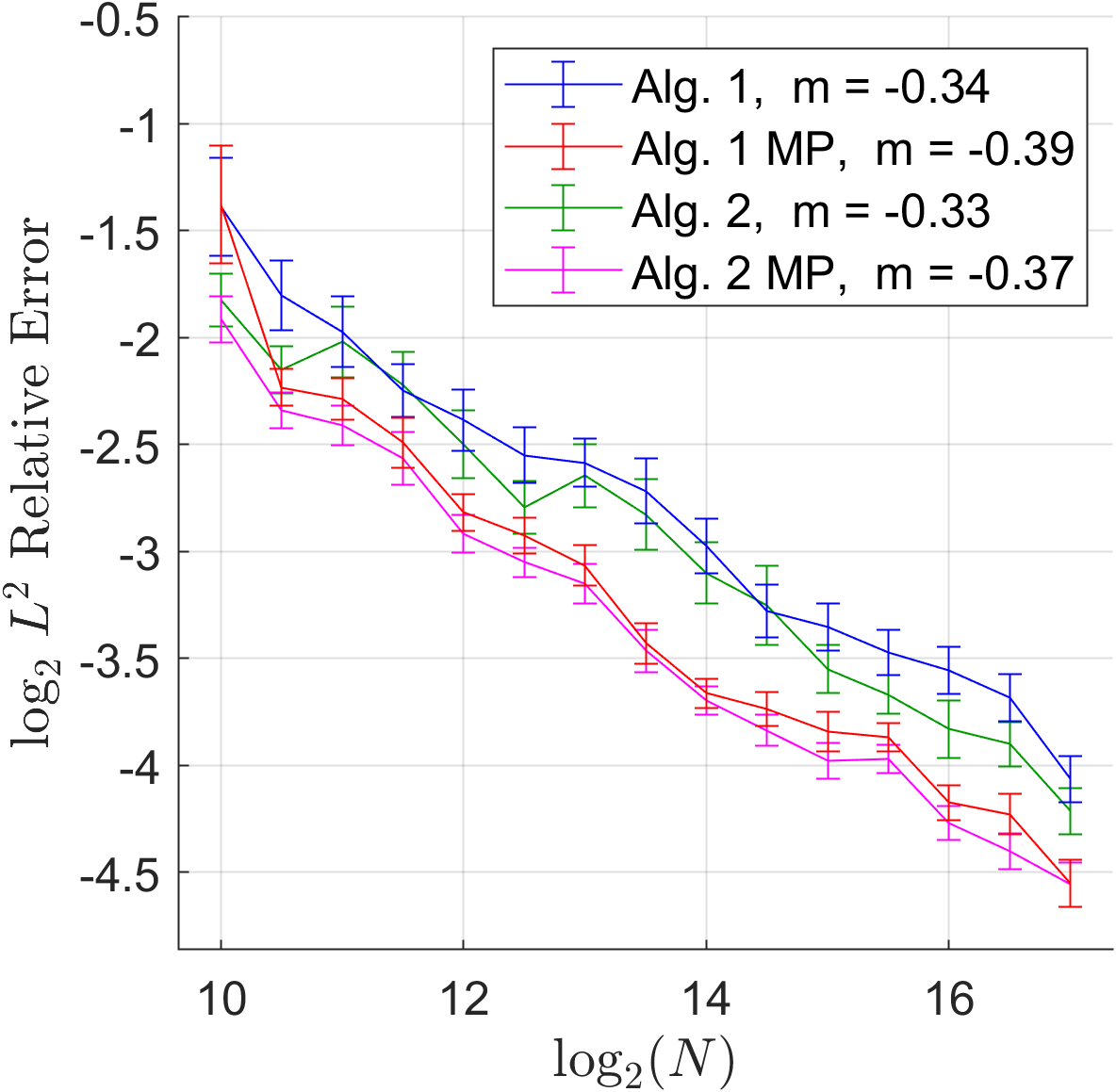}
		\caption{Recovery of $ f_3 $.}
		\label{fig:sub3}
    \end{subfigure}
    \hfill
    \begin{subfigure}[b]{0.45\textwidth}
        \centering
        \includegraphics[width=\textwidth]{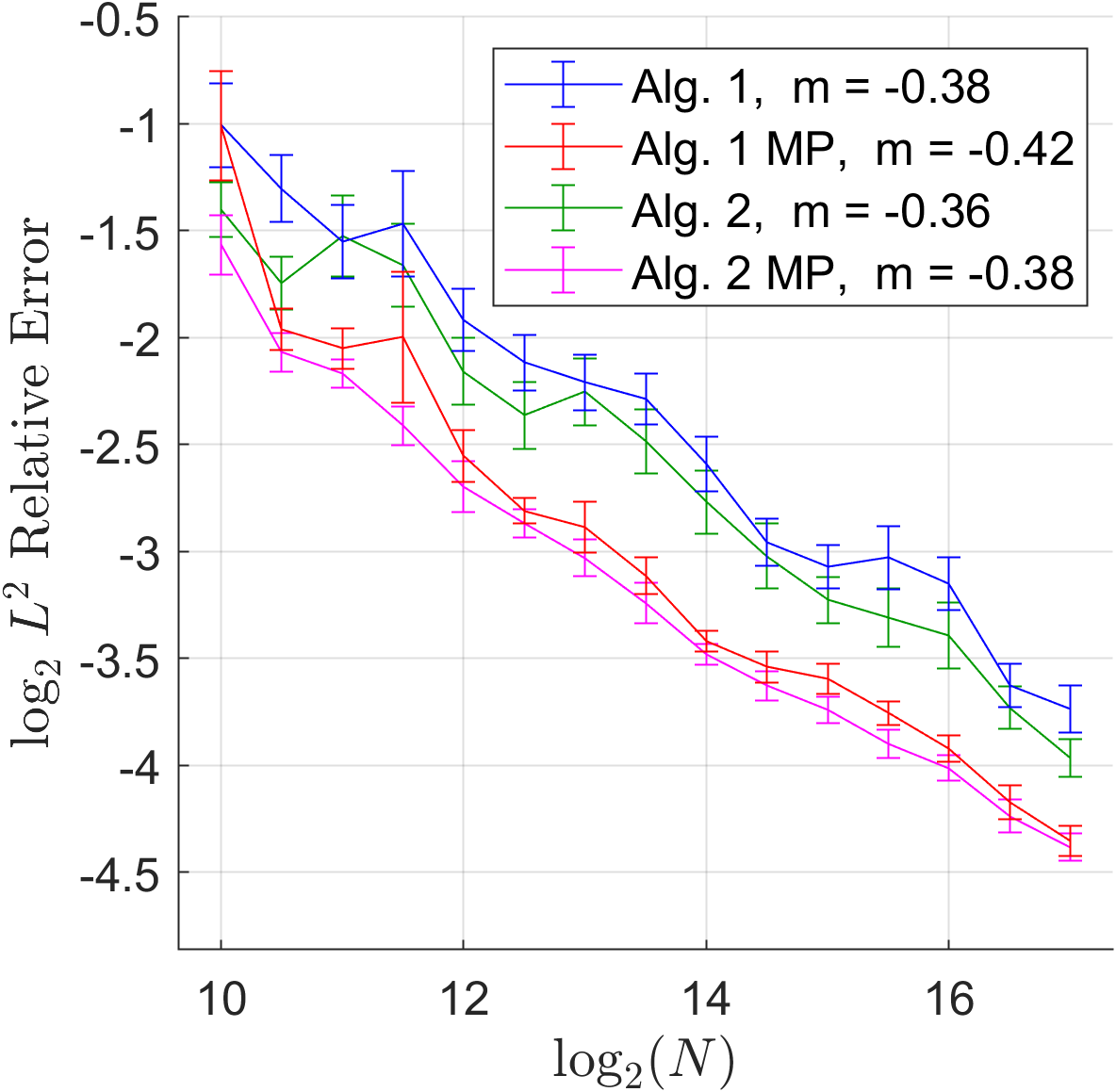}
		\caption{Recovery of $ f_4 $.}
		\label{fig:sub4}
    \end{subfigure}
    
    \caption{Error as a function of $\sam$ for recovering the four hidden signals using Algorithms \ref{alg:mainalgFM} and \ref{alg:mainalg} and their `multi-path' extensions. Slopes are reported for lines of best fit for each algorithm. }
	\label{fig:changing_n_l5}
\end{figure}
Figure \ref{fig:changing_n_l5} shows the relative spatial $L^2$ error decay in estimating the hidden signal $f$ using Algorithms \ref{alg:mainalgFM} and \ref{alg:mainalg}, as well as their `multi-path' extensions discussed in Subsection \ref{ssec:algextensions}. In this experiment, $ \sigma = 1, \lambda = .1 $ were held fixed to situate the simulations in the high noise regime. As discussed in the Section \ref{ssec:algextensions}, an additional regularization constant of $r=10^{-3}$ was introduced, and the deconvolution bandwidth was chosen to be $h = 10^{-1}$ for all functions. $\sam$ was varied evenly in $ \log_2 $ space from $ 2^{10} $ to $ 2^{17} $. Reported slopes are for lines of best fit. We call attention to the fact that as the hidden signals vary from $ f_1 $ to $ f_4 $, the slopes of the best fit lines for all algorithms increase monotonically. This demonstrates the dependence on $ \beta $ for the relative error bounds in Theorem \ref{thm:ErrorBoundDeconvolutionEstimator}: as $ \beta $ increases, the relative error in recovery decreases and thus the slope of the recovery best fit line increases in absolute value as well.

We also call attention to the performance of the various algorithms relative to each other. For the slower decaying functions $ f_1,f_2 $, there is less of a difference between Algorithms \ref{alg:mainalgFM}, \ref{alg:mainalg} and their `multi-path' equivalents. However, for the functions $ f_3,f_4 $ with faster frequency decay, the `multi-path' algorithms outperform their `single-path' equivalents, giving lower relative errors at all values of $\sam$. In general, for all functions, Algorithm \ref{alg:mainalg} and its `multi-path' equivalent outperform Algorithm \ref{alg:mainalgFM} and its `multi-path' equivalent, respectively---albeit only slightly for the slower frequency decaying functions. This difference in performance is accentuated in Figure \ref{fig:changing_n_l3}. The experiment for this figure was conducted for $ f_1 $ and is equivalent to that of Figure \ref{fig:sub1} with the exception that the spatial sampling rate was changed to the coarser $ 2^{-3} $, rather than $ 2^{-5} $. Here, both the single and `multi-path' variants of Algorithm \ref{alg:mainalg} outperform Algorithm \ref{alg:mainalgFM} and its `multi-path' extension at all $\sam$ values, with the gap in performance closing as $\sam$ increases. Together with Figure \ref{fig:changing_n_l5}, this demonstrates that Algorithm
\ref{alg:mainalg} increasingly outperforms Algorithm \ref{alg:mainalgFM} as the
number of informative Fourier modes decreases, whether because of a faster
\(\beta\) decay rate or fewer sampled frequency points.
This is likely due to the fact that Algorithm \ref{alg:mainalgFM} only recovers integer frequencies to reconstruct a Fourier series periodic extension of the hidden signal, while Algorithm \ref{alg:mainalg} recovers a finer frequency grid of values for the hidden signal;  Algorithm \ref{alg:mainalg} is less sensitive to having fewer frequency values and is more discretization-agnostic.

\subsubsection{Varying Noise Intensity}
In Figures \ref{fig:changing_sigma_l3} and \ref{fig:changing_sigma_l5} we held the number of signal occurrences fixed at $ N = 2^{10} $ and studied the spatial relative $ L^2 $ error in recovery of the hidden signal $f_1$ from the estimators constructed via Algorithms \ref{alg:mainalgFM} and  \ref{alg:mainalg} and their `multi-path' extensions. The noise intensity $ \sigma $ was varied evenly in $ \log_2 $ spacing from $ 2^{-10} $ to $ 2^3 $, while the algorithm constants $ r=10^{-3}, h = 10^{-1} $ were held fixed, as in Figure \ref{fig:changing_n_l5}. In Figure \ref{fig:changing_sigma_l3}, the spatial sampling rate was fixed at $ 2^{-3} $, while in Figure \ref{fig:changing_sigma_l5}, the spatial sampling rate was fixed at the finer $ 2^{-5} $. In addition to our proposed algorithms, we also consider the spectral method for
bispectrum inversion introduced in \cite{spectralbispecinversion}. This is another
non-initialized approach to bispectrum inversion, but it was developed for the
discrete formulation of the MRA problem and relies on a sufficiently large spectral
gap in an eigenvalue decomposition associated with the bispectrum. While the spectral algorithm can be applied to the functional shifts-off-grid setting (see Section \ref{subsec:num_setup}), its performance is highly dependent on the sampling rate. In Figure \ref{fig:changing_sigma_l3}, the spectral algorithm performs only slightly worse than our proposed algorithms as $ \sigma $ increaes, but in Figure \ref{fig:changing_sigma_l5}, the spectral algorithm dramatically fails; the fine spatial grid leads to very small eigenvalue gaps in the spectral
decomposition of the algorithm's bispectrum matrix, making the method highly
sensitive to even minuscule noise. Thus, among uninitialized algorithms for
bispectrum inversion, our proposed algorithms are better suited to the
off-grid-shift setting and are more discretization-agnostic.

\begin{figure}[h]
	\centering
	\begin{subfigure}[b]{0.32\textwidth}
		\centering
		\includegraphics[width=\textwidth]{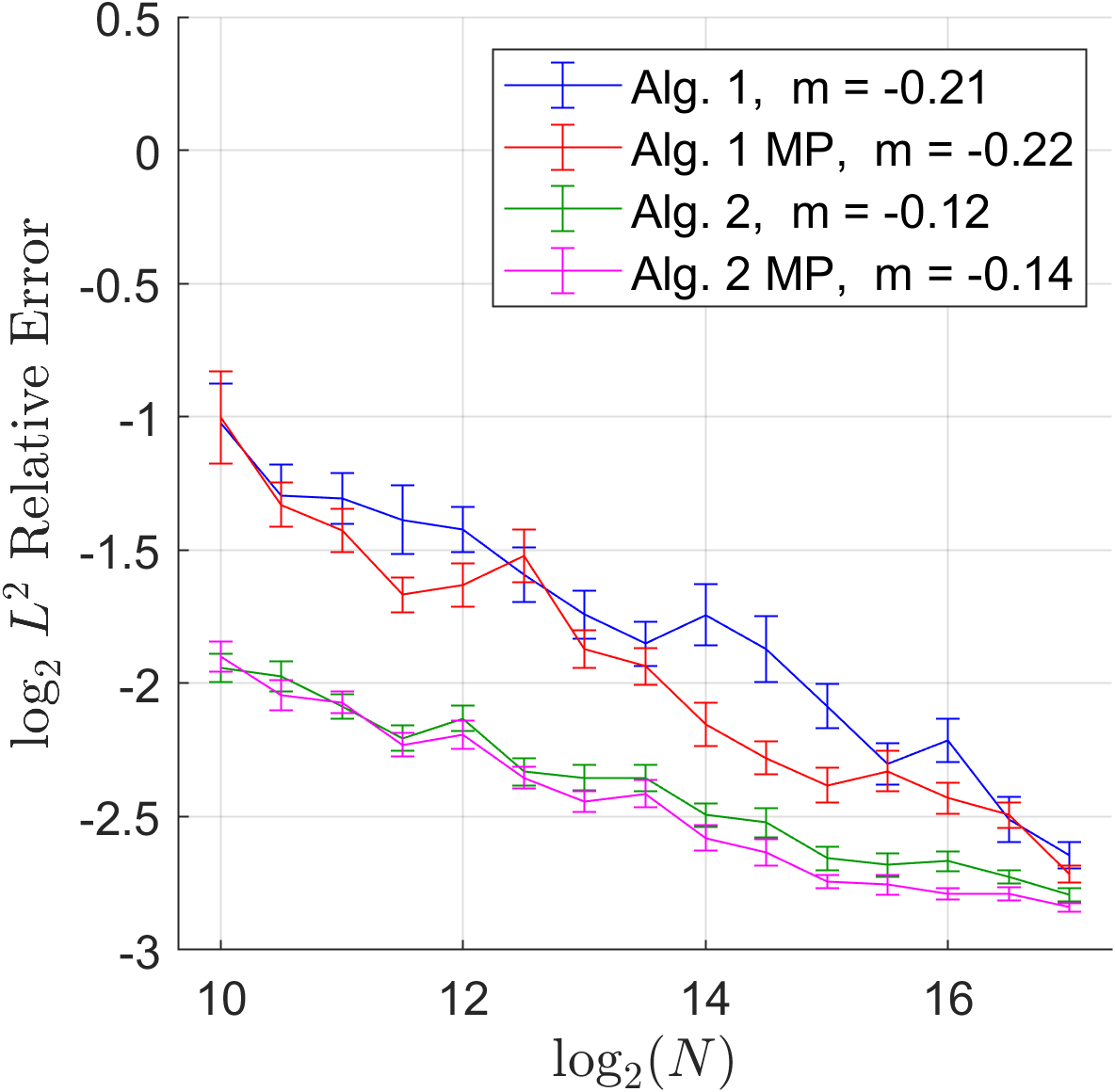}
		\caption{Recovery of $f_1$ as a function of $\sam$ for a spatial sampling rate of $2^{-3}$.}
		\label{fig:changing_n_l3}
	\end{subfigure}
	\hfill
	\begin{subfigure}[b]{0.32\textwidth}
		\centering
		\includegraphics[width=\textwidth]{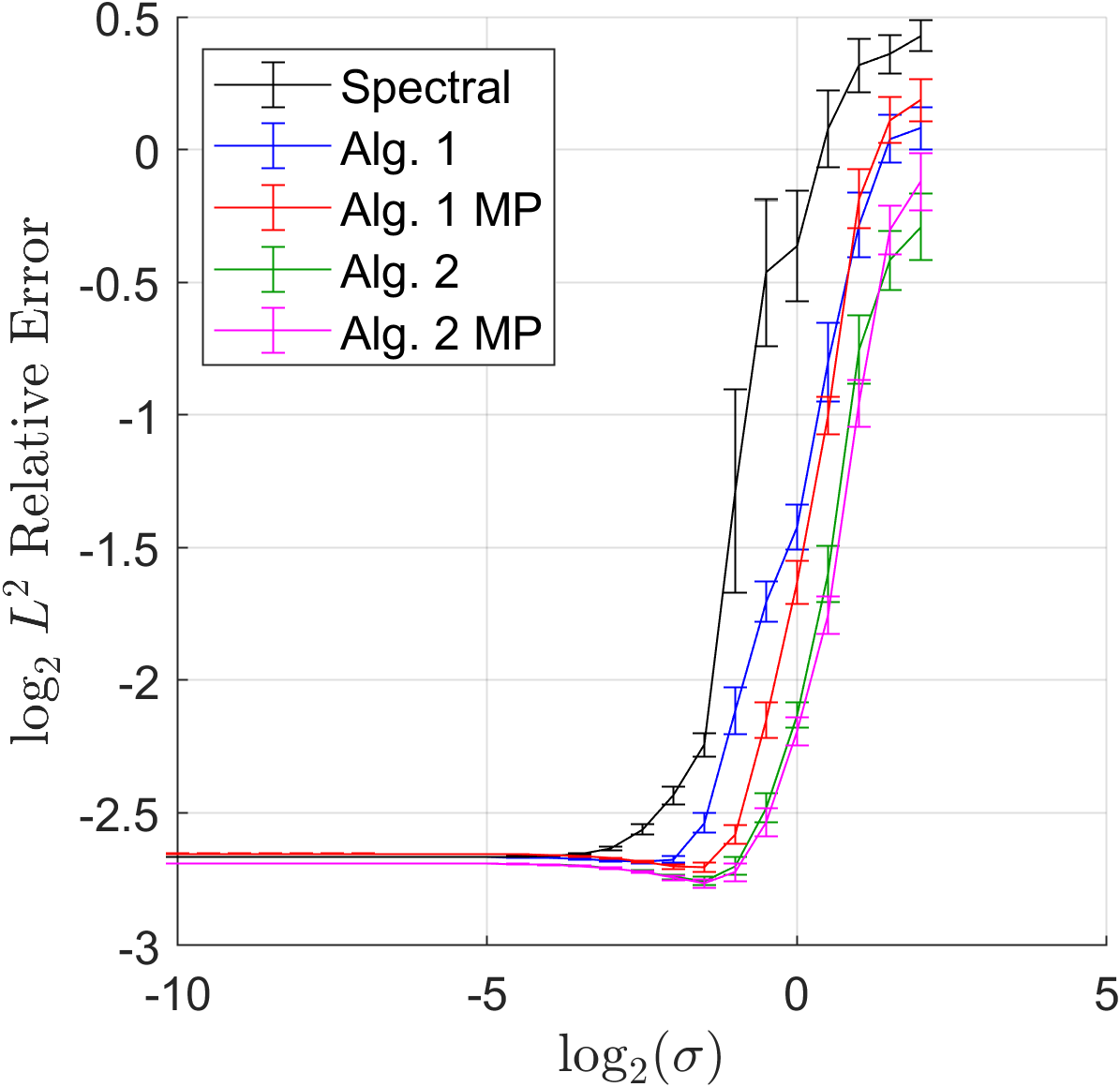}
		\caption{Error in recovery of $f_1$ as a function of $ \sigma $, for a spatial sampling rate of $2^{-3}$. }
		\label{fig:changing_sigma_l3}
	\end{subfigure}
	\hfill
	\begin{subfigure}[b]{0.32\textwidth}
	\centering
	\includegraphics[width=\textwidth]{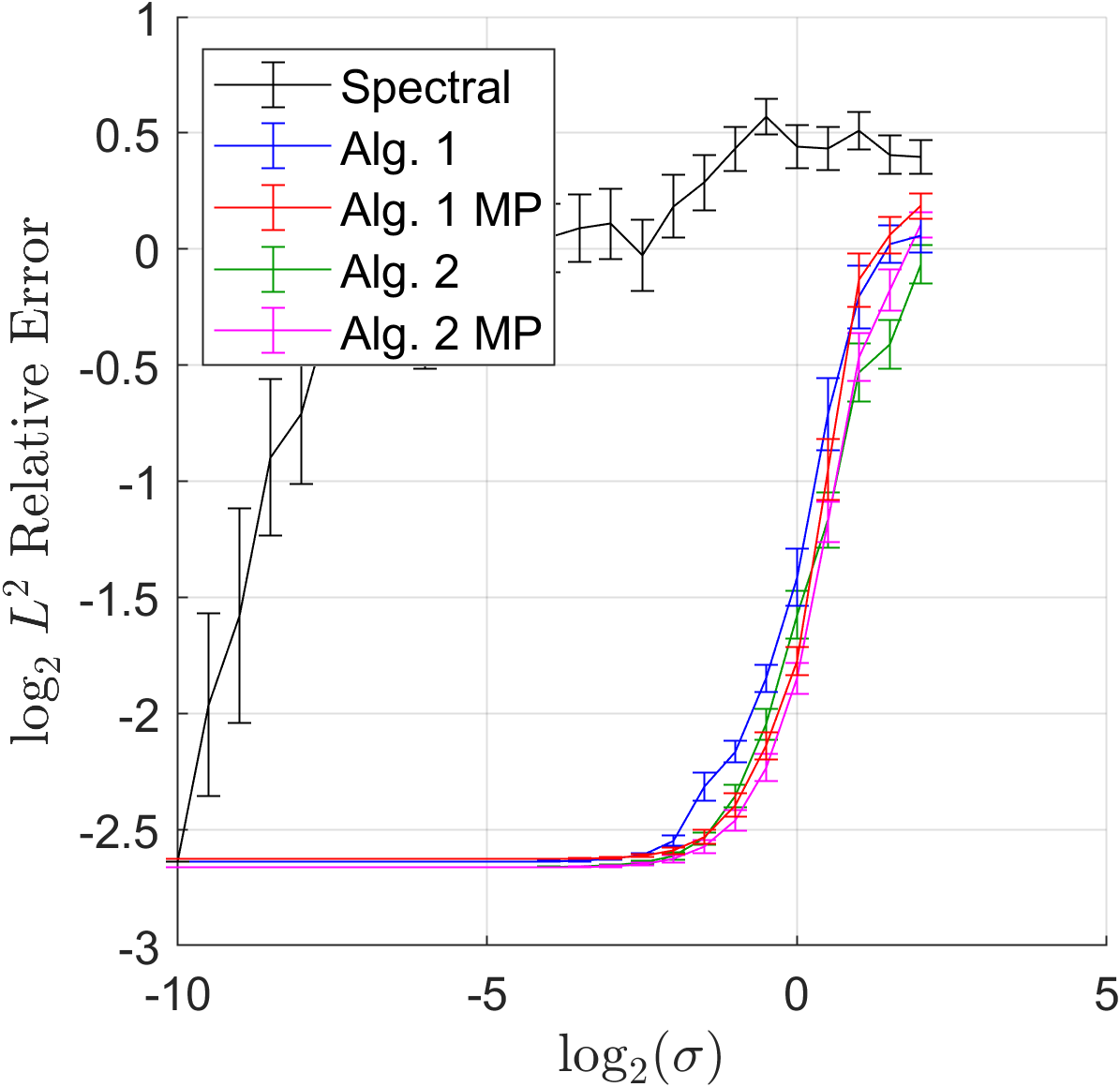}
	\caption{Error in recovery of $f_1$ as a function of $ \sigma $, for a spatial sampling rate of $2^{-5}$. }
	\label{fig:changing_sigma_l5}
	\end{subfigure}

	\caption{Three experiments comparing the performance of Algorithms \ref{alg:mainalgFM} and \ref{alg:mainalg}, their `multi-path' extensions, and the spectral algorithm from \cite{spectralbispecinversion} as functions of $\sam$ and noise intensity. The left panel is identical to the experiment in Figure \ref{fig:sub1} but with a coarser spatial sampling rate. The middle and right panels are equivalent except for a differing spatial sampling rate.   }
	\label{fig:bispec_est}
\end{figure}

\section{Conclusion}\label{sec:conclusions}
This paper introduced and analyzed two end-to-end algorithms for signal recovery in the functional MTD setting, providing the first explicit signal recovery guarantees in the MTD literature with finite sample concentration results. Our functional formulation accommodates off-grid translations, spatially correlated noise, and signals without band-limiting assumptions, with recovery rates depending explicitly on signal smoothness. A key theoretical byproduct is the identification of Kotlarski's formula as the continuous limit of frequency marching, unifying methods in deconvolution and bispectrum inversion.
Natural directions for future work include incorporating rotations into the functional MTD model, relaxing the well-separated assumption, and closing the gap between the upper bounds established here and the sample complexity lower bounds of \cite{abraham2025sample, SEorbitrec}.

\bibliographystyle{siam}
\bibliography{MTD, MRA}

\newpage
\begin{appendix}
    \section{Auxiliary Results}\label{sec:technical lemmas}

\begin{lem}\label{lem:bispectrum}
    Let $g\in L^1(\mathbb{R}^d)$ be compactly supported on $[-(\dom-2\pi),\dom-2\pi]^d$. Then
    \[ (A_3g)^{ft}(\omega_1,\omega_2) =\B g(\omega_1,\omega_2) \, ,\]
    where $A_3g$ is defined in \eqref{equ:A3_def} and $\B g$ in \eqref{equ:bispectrum_def}.
\end{lem}

\begin{proof}
By a straightforward calculation, we have 
    \begin{align*}
        (A_3g)^{ft}(\omega_1,\omega_2) &= \int_{\R^d}\int_{\R^d} \int_{[-\dom,\dom]^d} g(t)g(t+z_1)g(t-z_2) e^{-i\omega_1\cdot z_1-i\omega_2\cdot z_2}dtdz_1dz_2 \\
        &= \int_{\R^d} \left[ \int_{\R^d} g(t+z_1)e^{- i z_1\cdot\omega_1} \right] \left[ \int_{\R^d} g(t-z_2)e^{- i z_2\cdot\omega_2}\right] g(t) dt\\
        &= \int_{\R^d} \left[ \int_{\R^d} g(z_1)e^{- i \omega_1\cdot (z_1-t)}dz_1\right ]\left[ \int_{\R^d} g(z_2)e^{- i \omega_2\cdot (t-z_2)}dz_2\right ] g(t) dt\\
        &= \gft(\omega_1)\gft(-\omega_2)\int_{\R^d} g(t)e^{- i (\omega_2-\omega_1)\cdot t}dt\\
        &= \gft(\omega_1)\gft(-\omega_2)\gft(\omega_2-\omega_1) = \B g(\omega_1,\omega_2),
    \end{align*}
    as desired. 
\end{proof}

\begin{lem}\label{lemma:A3M}
Under Assumption \ref{assump:standing}, we have $A_3F=\sam A_3f$ and 
\begin{align*}
    \E\biggl[\frac{A_3\obs(z_1,z_2)}{\sam} \biggr] = A_3f(z_1,z_2)+\noisecov(z_2)+\noisecov(z_1)+\noisecov(z_1+z_2) .
\end{align*}
\end{lem}
\begin{proof}
By the separation condition in Assumption \ref{assump:standing}, we have
\begin{align*}
    A_3F(z_1,z_2)&= \int_{[-\dom,\dom]^d} \sum_{i=1}^\sam f(t-x_i) \sum_{j=1}^\sam f(t-x_j+z_1) \sum_{k=1}^\sam f(t-x_k-z_2)dt \\
    &=\int_{[-\dom,\dom]^d} \sum_{i=1}^\sam f(t-x_i)f(t-x_i+z_1)f(t-x_i-z_2)dt= \sam A_3f(z_1,z_2).
\end{align*}
A straightforward calculation gives:
\begin{align*}
    A_3\obs(z_1,z_2) &= A_3F(z_1,z_2) + A_3\epsilon(z_1,z_2) \\
    &\quad +A_2[\epsilon(t), F(t+z_1)](-z_2) + A_2[\epsilon(t),F(t-z_2)](z_1) + A_2[\epsilon(t+z_1),F(t)](-z_2-z_1) \\
    &\quad +A_1[\epsilon(t-z_2),F(t)F(t+z_1)] + A_1[\epsilon(t+z_1),F(t)F(t-z_2)] + A_1[\epsilon(t),F(t+z_1)F(t-z_2)] \, ,
\end{align*}   
where $A_2, A_1$ are the weighted autocorrelation functions defined in \eqref{equ:weighted_A2_A1_def}.
Taking expectation on both sides and noticing that all odd moments of the noise equal zero, we get 
\begin{align*}
    \E[A_3\obs(z_1,z_2)]& = A_3F(z_1,z_2) + \E\Big[A_2[\epsilon(t), F(t+z_1)](-z_2)\Big] \\
    &\quad +\E\Big[A_2[\epsilon(t),F(t-z_2)](z_1)\Big] + \E\Big[A_2[\epsilon(t+z_1),F(t)](-z_2-z_1)\Big]\\
    & = A_3F(z_1,z_2)+\E \left[\int_{[-\dom,\dom]^d} \epsilon(t)\epsilon(t-z_2)F(t+z_1)dt\right]\\
    &\quad +\E \left[\int_{[-\dom,\dom]^d}\epsilon(t)\epsilon(t+z_1)F(t-z_2)\right]+\E\left[\int_{[-\dom,\dom]^d}\epsilon(t+z_1)\epsilon(t-z_2)F(t)dt\right]\\
    &= A_3F(z_1,z_2) +[\noisecov(z_2)+\noisecov(z_1)+\noisecov(z_1+z_2)]\int_{[-\dom,\dom]^d} F(t)dt,
\end{align*}
where in the last step we have used the fact that the integration domain $[-\dom,\dom]^d$ covers the support of $F$ for shifts $z_1,z_2\in [-2\pi,2\pi]^{2d}$, so that all three integrals of $F$ above are equal to $\int_{[-\dom,\dom]^d}F(t)dt$.
The result then follows by noticing that $f^{ft}(0)=1$ implies $\int_{[-\dom,\dom]^d} F(t)dt = \sam  \int_Df(t)dt= \sam $. 

\end{proof}

\begin{defn}\label{def:error_decomp}
We explicitly decompose the error terms arising in \eqref{eq:concentration error} into the following random processes: 
\begin{align}
    \xxx(z)& := \frac{1}{\sqrt{\sam}}\int_{[-\dom,\dom]^d} \epsilon(t)\epsilon(t-z_1)\epsilon(t-z_2)dt\label{eq:type 3 error}\\
    \xxone(z) & := \frac{1}{\sqrt{\sam}}\int_{[-\dom,\dom]^d} \left[\sum_{i=1}^\sam  f(t-x_i)\right]\Bigg[\epsilon(t+z_1)\epsilon(t-z_2)-\noisecov(z_1+z_2)\Bigg]dt\label{eq:type 2 error}\\
    \xxtwo(z) &:= \frac{1}{\sqrt{\sam}}\int_{[-\dom,\dom]^d} \left[\sum_{i=1}^\sam  f(t-x_i)\right]\Bigg[\epsilon(t+z_1)\epsilon(t+z_1+z_2)-\noisecov(z_2)\Bigg]dt\nonumber\\
    \xxthree(z) &:= \frac{1}{\sqrt{\sam}}\int_{[-\dom,\dom]^d} \left[\sum_{i=1}^\sam  f(t-x_i)\right]\Bigg[\epsilon(t-z_1)\epsilon(t-z_1-z_2)-\noisecov(z_2)\Bigg]dt\nonumber\\
    \xone(z) & := \frac{1}{\sqrt{\sam}}\int_{[-\dom,\dom]^d}\left[ \sum_{i=1}^\sam  f(t-x_i+z_1)f(t-x_i-z_2)\right] \epsilon(t)dt\label{eq:type 1 error}\\
    \xtwo(z) &:=  \frac{1}{\sqrt{\sam}}\int_{[-\dom,\dom]^d} \left[\sum_{i=1}^\sam  f(t-x_i+z_2)f(t-x_i+z_1+z_2)\right]\epsilon(t)dt \nonumber\\
    \xthree(z) & :=  \frac{1}{\sqrt{\sam}}\int_{[-\dom,\dom]^d}\left[\sum_{i=1}^\sam  f(t-x_i-z_1)f(t-x_i-z_1-z_2)\right]\epsilon(t)dt\nonumber
\end{align}
\end{defn}

\begin{defn}\label{def:Gaussian_chaos}
Let $r$ be a pseudo-metric on a set $D$. A centered random field $X$ on $D$ is said to be a sub-$k$th-Gaussian chaos field with respect to the metric $r$ if 
\begin{align}\label{eq:sub-alpha exponential assumption}
    \E \left[\exp\left(\frac{X(z)-X(w)}{r(z,w)}\right)^{2/k}\right]\leq 2,\qquad \forall z,w\in D.
\end{align}
The canonical metric of $X$ is defined as $s(z,w)=\sqrt{\E|X(z)-X(w)|^2}$. 
\end{defn}

\begin{lem}\label{lemma:moments gaussian chaos}
Under Assumption \ref{assump:standing},
up to a rescaling, the random fields $\{\xone,\xtwo,\xthree\}$,  $\{\xxone,\xxtwo,\xxthree\}$, and $\{\xxx\}$ in Definition \ref{def:error_decomp} are sub-$kth$ Gaussian chaos fields with respect to their canonical metrics with $k=1,2,3$ respectively.  
\end{lem}
\begin{proof}
Let $\mu$ be the spectral measure of the real-valued stationary Gaussian field
$\epsilon$. By the spectral representation theorem
\cite[Proposition 14.19]{kallenberg1997foundations}, there exists a complex
Gaussian random spectral measure $W_\mu$, satisfying the Hermitian symmetry
condition
\[
W_\mu(-d\xi)=\overline{W_\mu(d\xi)},
\]
such that
\[
    \epsilon(t)=\int_{\R^d} e^{i\xi\cdot t}\,dW_\mu(\xi).
\]
For \(t\in\R^d\), define
\[
    f_t(\xi):=e^{i\xi\cdot t}.
\]
Then
\[
    \epsilon(t)=I_1(f_t),
\]
where \(I_1\) denotes the first-order Wiener--It\^o integral with respect to
\(W_\mu\). For this complex spectral representation, contractions are taken
with respect to the covariance pairing
\[
    \langle f,g\rangle_\mu
    :=
    \E[I_1(f)I_1(g)]
    =
    \int_{\R^d} f(\xi)g(-\xi)\,d\mu(\xi).
\]
In particular,
\[
\begin{aligned}
    \langle f_t,f_{t-z}\rangle_\mu
    &=
    \int_{\R^d} e^{i\xi\cdot t}e^{-i\xi\cdot(t-z)}\,d\mu(\xi)  \\
    &=
    \int_{\R^d} e^{i\xi\cdot z}\,d\mu(\xi).
\end{aligned}
\]
Since \(\epsilon\) is real-valued, the measure \(\mu\) is symmetric, and hence
the last display is real-valued and unchanged if \(z\) is replaced by \(-z\).

For \(F:(\R^d)^k\to\mathbb{C}\), write
\[
    I_k(F)
    :=
    \int_{(\R^d)^k}
    F(\xi_1,\ldots,\xi_k)\,
    dW_\mu(\xi_1)\cdots dW_\mu(\xi_k)
\]
for the \(k\)-th order Wiener--It\^o integral. Since
\(I_k(F)=I_k(\widetilde F)\), where
\[
    \widetilde F(\xi_1,\ldots,\xi_k)
    =
    \frac{1}{k!}
    \sum_{\pi\in\Pi_k}
    F(\xi_{\pi(1)},\ldots,\xi_{\pi(k)})
\]
is the symmetrization of \(F\), we work with symmetrized kernels below.

For functions \(f_1,\ldots,f_k\), we write
\[
f_1\widetilde\otimes \cdots \widetilde\otimes f_k
\]
for the symmetrized tensor product
\[
\left(f_1\widetilde\otimes \cdots \widetilde\otimes f_k\right)
(\xi_1,\ldots,\xi_k)
:=
\frac{1}{k!}
\sum_{\pi\in\Pi_k}
\prod_{j=1}^k f_j(\xi_{\pi(j)}),
\]
where \(\Pi_k\) denotes the set of permutations of \(\{1,\ldots,k\}\).

By the product formula for Wiener--It\^o integrals of symmetric functions
\cite[Proposition 1.1.2]{nualart2006malliavin}, and using the notation
above for symmetrized tensor products, we have
\[
\begin{aligned}
    I_1(f)I_1(g)I_1(h)
    &=
    I_3(f\widetilde\otimes g\widetilde\otimes h)  \\
    &\quad
    + I_1\left(
        \langle f,g\rangle_\mu h
        +
        \langle f,h\rangle_\mu g
        +
        \langle g,h\rangle_\mu f
    \right).
\end{aligned}
\]
Applying this identity with
\[
    f=f_t,\qquad g=f_{t-z_1},\qquad h=f_{t-z_2},
\]
we obtain
\[
    E_3(z)
    :=
    \frac{1}{\sqrt N}
    \int_{[-R,R]^d}
    \epsilon(t)\epsilon(t-z_1)\epsilon(t-z_2)\,dt
    =
    I_3(H_z)+I_1(G_z),
\]
where
\[
\begin{aligned}
    H_z(\xi_1,\xi_2,\xi_3)
    &:=
    \frac{1}{\sqrt N}
    \int_{[-R,R]^d}
    \left(
    f_t\widetilde\otimes f_{t-z_1}
    \widetilde\otimes f_{t-z_2}
    \right)(\xi_1,\xi_2,\xi_3)\,dt
\end{aligned}
\]
and
\[
\begin{aligned}
    G_z(\xi)
    &:=
    \frac{1}{\sqrt N}
    \int_{[-R,R]^d}
    \bigg[
    \left(\int_{\R^d} e^{ix\cdot z_1}\,d\mu(x)\right)
    e^{i\xi\cdot(t-z_2)}
    \\
    &\qquad\qquad\qquad
    +
    \left(\int_{\R^d} e^{ix\cdot z_2}\,d\mu(x)\right)
    e^{i\xi\cdot(t-z_1)}
    \\
    &\qquad\qquad\qquad
    +
    \left(\int_{\R^d} e^{ix\cdot(z_2-z_1)}\,d\mu(x)\right)
    e^{i\xi\cdot t}
    \bigg]\,dt .
\end{aligned}
\]
Equivalently, the three scalar coefficients in \(G_z\) are the contractions
\[
    \langle f_t,f_{t-z_1}\rangle_\mu,\qquad
    \langle f_t,f_{t-z_2}\rangle_\mu,\qquad
    \langle f_{t-z_1},f_{t-z_2}\rangle_\mu .
\]
Indeed,
\[
\begin{aligned}
    \langle f_t,f_{t-z_1}\rangle_\mu
    &=
    \int_{\R^d} e^{ix\cdot z_1}\,d\mu(x),\\
    \langle f_t,f_{t-z_2}\rangle_\mu
    &=
    \int_{\R^d} e^{ix\cdot z_2}\,d\mu(x),\\
    \langle f_{t-z_1},f_{t-z_2}\rangle_\mu
    &=
    \int_{\R^d} e^{ix\cdot(z_2-z_1)}\,d\mu(x).
\end{aligned}
\]
Since \(\mu\) is symmetric, the last expression is equal to
\[
    \int_{\R^d} e^{ix\cdot(z_1-z_2)}\,d\mu(x).
\]

Therefore,
\[
    E_3(z)-E_3(w)
    =
    I_3(H_z-H_w)+I_1(G_z-G_w)
    =:\mathcal I_3+\mathcal I_1 .
\]
By orthogonality of Wiener--It\^o integrals of different orders
\cite[page 9, eqn. (iii)]{nualart2006malliavin},
\[
\begin{aligned}
    s(z,w)^2
    &:=
    \E |E_3(z)-E_3(w)|^2  
    =
    \E |\mathcal I_3|^2+\E |\mathcal I_1|^2 .
\end{aligned}
\]
In particular,
\[
    \sqrt{\E|\mathcal I_3|^2}\le s(z,w),
    \qquad
    \sqrt{\E|\mathcal I_1|^2}\le s(z,w).
\]
Thus, by tail bounds for multiple Wiener--It\^o integrals
\cite[Theorem 6.6]{major2006multiple}, for every \(u>0\),
\[
\begin{aligned}
    \P\left(|E_3(z)-E_3(w)|>u\,s(z,w)\right)
    &\le
    \P\left(|\mathcal I_3|>\frac{u\,s(z,w)}{2}\right)
    +
    \P\left(|\mathcal I_1|>\frac{u\,s(z,w)}{2}\right)  \\
    &\le
    C_3\exp\left(-c_3 u^{2/3}\right)
    +
    C_1\exp\left(-c_1 u^{2}\right) \\
    &\le
    C\exp\left(-c u^{2/3}\right).
\end{aligned}
\]
Equivalently,
\[
    \P\left(
    \frac{|E_3(z)-E_3(w)|}{s(z,w)}>u
    \right)
    \le
    C\exp\left(-c u^{2/3}\right).
\]
After adjusting constants, or equivalently after rescaling the field, this gives
the desired sub-third Gaussian chaos increment bound.

The proof for the first- and second-order fields is analogous. Terms that are
linear in \(\epsilon\) are first-order Wiener--It\^o integrals. Terms that are
quadratic in \(\epsilon\) are handled by the product formula
\[
    I_1(f)I_1(g)
    =
    I_2(f\widetilde\otimes g)
    +
    \langle f,g\rangle_\mu .
\]
Thus the centered quadratic terms lie in the second Wiener chaos, while the
uncentered quadratic terms lie in the sum of the zeroth and second chaoses.
Consequently, up to rescaling, the random fields
\[
    \{\xone,\xtwo,\xthree\},\qquad
    \{\xxone,\xxtwo,\xxthree\},\qquad
    \{\xxx\}
\]
are sub-\(k\)th Gaussian chaos fields with \(k=1,2,3\), respectively.
\end{proof}

\begin{lem}\label{lemma:bound on canonical distance}
Under Assumption \ref{assump:standing}, the three types of errors in Definition \ref{def:error_decomp} satisfy 
\begin{enumerate}
    \item 
    \begin{align*}
        \E|\xxx(z)-\xxx(w)|^2 &\leq C(\boundcov^2\|\nabla\noisecov\|_1+\boundcov\Lipcov\|\noisecov\|_1)|z-w|_2.
    \end{align*}
    \item 
    \begin{align*}
        \E |\xxone(z)-\xxone(w)|^2 &\leq 4\sqrt{2}\Lipcov\|f\|^2_2 \|\noisecov\|_1 |z-w|_2 \, ,\\
    \E |\xxtwo(z)-\xxtwo(w)|^2 &\leq 8 \sqrt{2}\Lipcov  \|f\|^2_2\|\noisecov\|_1 |z-w|_2 \, ,\\
    \E |\xxthree(z)-\xxthree(w)|^2 &\leq 8 \sqrt{2}\Lipcov  \|f\|^2_2\|\noisecov\|_1 |z-w|_2 \, .
    \end{align*}

    \item 
    \begin{align*}
        \E |\xone(z)-\xone(w)|^2 &\leq 2\gamma^{-1}\Lipf^2\|\noisecov\|_1\|f\|^2_\infty |z-w|_2^2 \, ,\\
    \E|\xtwo(z)-\xtwo(w)|^2 &\leq 4\gamma^{-1} L_f^2 \|\noisecov\|_1\|f\|^2_\infty|z-w|^2_2 \, ,\\
    \E|\xthree(z)-\xthree(w)|^2 &\leq 4\gamma^{-1} L_f^2 \|\noisecov\|_1\|f\|^2_\infty |z-w|^2_2 \, , 
    \end{align*}
\end{enumerate}
where the constants are defined in Assumption \ref{assump:standing}.
\end{lem}
\begin{proof}
\textbf{Bound for $\xxx$:}
Notice that  
\begin{align*}
    \E |\xxx(z)-\xxx(w)|^2 &=\frac{1}{\sam} \E \left|\int_{[-\dom,\dom]^d}\underbrace{\epsilon(t)\epsilon(t-z_1)\epsilon(t-z_2)-\epsilon(t)\epsilon(t-w_1)\epsilon(t-w_2)}_{=:u(t)}dt\right|^2\\
    &=\frac{1}{\sam} \E \int_{-[\dom,\dom]^d} u(t)dt \int_{[-\dom,\dom]^d}u(s)ds\\
    & =\frac{1}{\sam} \int_{[-\dom,\dom]^d}  \int_{[-\dom,\dom]^d}  \E[u(t)u(s)]dtds\\
    & =\frac{1}{\sam} \int_{[-\dom,\dom]^d}  \int_{[-\dom,\dom]^d}  \cov_u(s-t)dtds\\
    & =\frac{1}{N}\int_{\R^d} \int_{[-\dom,\dom]^d\cap ([-\dom,\dom]^d-h)}   \cov_u(h) dvdh\\ 
    & \leq \frac{(2R)^d}{N} \int_{[-2\dom,2\dom]^d} |\cov_u(h)|dh= \frac{1}{\gamma}\int_{[-2\dom,2\dom]^d} |\cov_u(h)|dh.
\end{align*}
It suffices to bound $\int_{[-2\dom,2\dom]^d} |\cov_u(h)|dh$ in terms of $|z-w|_2$. 

Writing out explicitly, 
\begin{align}
    \cov_u(h) &= \E [u(0)u(h)] \nonumber\\
    &= \E[\epsilon(0)\epsilon(-z_1)\epsilon(-z_2)-\epsilon(0)\epsilon(-w_1)\epsilon(-w_2)][\epsilon(h)\epsilon(h-z_1)\epsilon(h-z_2)-\epsilon(h)\epsilon(h-w_1)\epsilon(h-w_2)]\nonumber\\
    &=\E [\epsilon(0)\epsilon(h)\epsilon(-z_1)\epsilon(-z_2)\epsilon(h-z_1)\epsilon(h-z_2)] 
    +\E[\epsilon(0)\epsilon(h)\epsilon(-w_1)\epsilon(-w_2)\epsilon(h-w_1)\epsilon(h-w_2)]\nonumber\\
    &\quad -\E[\epsilon(0)\epsilon(h)\epsilon(-w_1)\epsilon(-w_2)\epsilon(h-z_1)\epsilon(h-z_2)]-\E[\epsilon(0)\epsilon(h)\epsilon(-z_1)\epsilon(-z_2)\epsilon(h-w_1)\epsilon(h-w_2)]\nonumber\\
    &=: \E[P_1P_2 Q_1 Q_2 Z_1Z_2]  - \E[P_1P_2Q_1Q_2V_1V_2]+ \E[P_1P_2 W_1W_2V_1V_2]- \E[P_1P_2W_1W_2Z_1Z_2], \label{eq:cov of u expansion}
\end{align}
where 
\begin{align*}
    P_1=\epsilon(0),\quad Q_1=\epsilon(-z_1),\quad Z_1&=\epsilon(h-z_1),\quad W_1=\epsilon(-w_1),\quad V_1=\epsilon(h-w_1),\\
    P_2=\epsilon(h),\quad Q_2=\epsilon(-z_2),\quad Z_2&=\epsilon(h-z_2),\quad W_2=\epsilon(-w_2),\quad V_2=\epsilon(h-w_2).
\end{align*}
Isserlis's theorem on the product of six Gaussian random variables states that each term in \eqref{eq:cov of u expansion} can be written as a sum of 15 triple products of expectations that involve only a pair of Gaussians:
\begin{align*}
    \E[P_1P_2 Q_1 Q_2 Z_1Z_2] = &\E[P_1P_2]\E[Q_1Q_2]\E[Z_1Z_2] + \E[P_1P_2]\E[Q_1Z_1]\E[Q_2Z_2] +\E[P_1P_2]\E[Q_1Z_2]\E[Q_2Z_1]+\ldots\\ 
    \E[P_1P_2Q_1Q_2V_1V_2]= & \E[P_1P_2]\E[Q_1Q_2]\E[V_1V_2]+\E[P_1P_2]\E[Q_1V_1]\E[Q_2V_2]+\E[P_1P_2]\E[Q_1V_2]\E[Q_2V_1] +\ldots
\end{align*}
We can control the differences one-by-one as follows 
\begin{align*}
   &\E[P_1P_2]\E[Q_1Q_2]\E[Z_1Z_2] - \E[P_1P_2]\E[Q_1Q_2]\E[V_1V_2] \\
   &= \noisecov(h)\noisecov(z_1-z_2) \big[\noisecov(z_1-z_2)-\noisecov(w_1-w_2)\big]\\
   &\leq \noisecov(h)\boundcov \Lipcov |(z_1-z_2)-(w_1-w_2)|_2\\
   &\leq \noisecov(h)\boundcov \Lipcov (|z_1-w_1|_2+|z_2-w_2|_2) \leq  \sqrt{2}\noisecov(h)\boundcov \Lipcov  |z-w|_2,
\end{align*}
and so 
\begin{align*}
    \int_{[-2\dom,2\dom]^d} \Big|\E[P_1P_2]\E[Q_1Q_2]\E[Z_1Z_2] - \E[P_1P_2]\E[Q_1Q_2]\E[V_1V_2]\Big| dh \leq \|\noisecov\|_1\boundcov \Lipcov |z_1-z_2|_2.
\end{align*}
Half of the terms can be bounded similarly in this way. 
The other half involve differences like 
\begin{align*}
    &\int_{[-2\dom,2\dom]^d} \big|\noisecov(z_1)\noisecov(z_1)\noisecov(h) - \noisecov(z_1)\noisecov(w_1)\noisecov(h-w_2+z_2)\big|dh\\
    &\leq \noisecov(z_1) \int_{[-2\dom,2\dom]^d} \big|\noisecov(z_1)\noisecov(h) -\noisecov(z_1)\noisecov(h-w_2+z_2)\big|  + \big|\noisecov(z_1)\noisecov(h-w_2+z_2)-\noisecov(w_1)\noisecov(h-w_2+z_2)\big|dh\\
    &\leq \boundcov^2 \int_{[-2\dom,2\dom]^d} |\noisecov(h)-\noisecov(h-w_2+z_2)|dh + \boundcov\Lipcov|z_1-w_1|_2\int_{[-2\dom,2\dom]^d} |\noisecov(h-w_2+z_2)| dh \\
    & \leq \boundcov^2 \int_{[-2\dom,2\dom]^d} |\noisecov(h)-\noisecov(h-w_2+z_2)|dh + \|\noisecov\|_1 \boundcov\Lipcov |z_1-w_1|_2.
\end{align*}
The first term can be bounded by noticing that 
\begin{align*}
\noisecov(h)-\noisecov(h-w_2+z_2) = \int_{0}^1 \nabla\noisecov(h+t(-w_2+z_2))^T (-w_2+z_2) dt
\end{align*}
so that 
\begin{align*}
    \int_{[-2\dom,2\dom]^d} |\noisecov(h)-\noisecov(h-w_2+z_2)|dh &\leq |z_2-w_2|_2\int_0^1 \int_{[-2\dom,2\dom]^d} |\nabla\noisecov(h+t(-w_2+z_2))|_2 dh dt  \\
    &\leq \|\nabla\noisecov\|_1|z_2-w_2|_2.
\end{align*}
Therefore we eventually have errors of the form 
\begin{align*}
    \int_{[-2\dom,2\dom]^d} |\cov_u(h)|dh \leq C (\boundcov^2 \|\nabla\noisecov\|_1 + \boundcov\Lipcov\|\noisecov\|_1)|z-w|_2,
\end{align*}
where $C$ is a universal constant, and hence 
\begin{align*}
   \E |\xxx(z)-\xxx(w)|^2 \leq C (\boundcov^2 \|\nabla\noisecov\|_1 + \boundcov\Lipcov\|\noisecov\|_1)|z-w|_2.
\end{align*}   

\textbf{Bound for $\xxone$, $\xxtwo$, $\xxthree$:}
Notice from \eqref{eq:type 2 error} that for each $i$, the canonical distance of the type II error processes all take the form
\begin{align}\label{eq:E2 general}
    \E |E_2^{(i)}(z)-E_2^{(i)}(w)|^2 = \frac{1}{\sam}\E \left|\int_{[-\dom,\dom]^d}F(t)v^{(i)}(t)\right|^2,
\end{align}
where 
\begin{align*}
    v^{(1)}(t)& = \epsilon(t+z_1)\epsilon(t-z_2)-\noisecov(z_1+z_2)-\epsilon(t+w_1)\epsilon(t-w_2)+\noisecov(w_1+w_2)\\
    v^{(2)}(t)& = \epsilon(t+z_2)\epsilon(t+z_1+z_2)-\noisecov(z_1)-\epsilon(t+w_2)\epsilon(t+w_1+w_2)+\noisecov(w_1)\\
    v^{(3)}(t)& = \epsilon(t-z_1)\epsilon(t-z_1-z_2)-\noisecov(z_2)-\epsilon(t-w_1)\epsilon(t-w_1-w_2)+\noisecov(w_2) .
\end{align*}
We have 
\begin{align*}
    \eqref{eq:E2 general}
    &= \frac{1}{\sam}\E\int_{[-\dom,\dom]^d} \int_{[-\dom,\dom]^d} v^{(i)}(t)v^{(i)}(s) F(t)F(s) dtds\\
    & = \frac{1}{\sam}\int_{[-\dom,\dom]^d} \int_{[-\dom,\dom]^d} F(t)F(s) \cov_{v^{(i)}}(t-s)dtds\\
    & = \frac{1}{\sam}\int_{\R^d} \left[\int_{[-\dom,\dom]^d\cap ([-\dom,\dom]^d-h)} F(s+h)F(s)ds\right]\cov_{v^{(i)}}(h)  dh,
\end{align*}
where $[-\dom,\dom]^d-h=\{s:s+h\in [-\dom,\dom]^d\}$. 
Now by Cauchy-Schwarz inequality,
\begin{align*}
    \int_{[-\dom,\dom]^d\cap ([-\dom,\dom]^d-h)} F(s+h)F(s)ds &\leq \|F\|_2\|F(\cdot-h)\|_2 \\
    &= \|F\|_2^2 = \int_{\R^d} F(t)^2 dt = \int_{\R^d} \sum_{i=1}^\sam  f(t-x_i)^2 dt  = \sam\|f\|^2_2.
\end{align*}
Therefore 
\begin{align*}
    \E |E_2^{(i)}(z)-E_2^{(i)}(w)|^2  \leq \|f\|^2_2 \int_{\R^d} |\cov_{v^{(i)}}(h)|dh.
\end{align*}
Now it suffices to bound the $L^1$ norm of each $\cov_{v^{(i)}}$. 
We first compute 
\begin{small}
\begin{align*}
    &\cov_{v^{(1)}}(h) \\
    &= \E v^{(1)}(0)v^{(1)}(h)\\
    &= \E [\epsilon(z_1)\epsilon(-z_2)-\noisecov(z_1+z_2)-\epsilon(w_1)\epsilon(-w_2)+\noisecov(w_1+w_2)]\\
    &\quad \cdot[\epsilon(h+z_1)\epsilon(h-z_2)-\noisecov(z_1+z_2)-\epsilon(h+w_1)\epsilon(h-w_2)+\noisecov(w_1+w_2)]\\
    & = \E[\epsilon(z_1)\epsilon(-z_2)\epsilon(h+z_1)\epsilon(h-z_2)]-\noisecov(z_1+z_2)\E[\epsilon(z_1)\epsilon(-z_2)]\\
    &\quad -\E[\epsilon(z_1)\epsilon(-z_2)\epsilon(h+w_1)\epsilon(h-w_2)]+\noisecov(w_1+w_2)\E[\epsilon(z_1)\epsilon(-z_2)]\\
    &\quad -\noisecov(z_1+z_2)\E[\epsilon(h+z_1)\epsilon(h-z_2)]+\noisecov(z_1+z_2)^2\\
    &\quad +\noisecov(z_1+z_2)\E[\epsilon(h+w_1)\epsilon(h-w_2)]-\noisecov(z_1+z_2)\noisecov(w_1+w_2)\\
    &\quad -\E[\epsilon(w_1)\epsilon(-w_2)\epsilon(h+z_1)\epsilon(h-z_2)]+\noisecov(z_1+z_2)\E[\epsilon(w_1)\epsilon(-w_2)]\\
    &\quad +\E[\epsilon(w_1)\epsilon(-w_2)\epsilon(h+w_1)\epsilon(h-w_2)]-\noisecov(w_1+w_2)\E [\epsilon(w_1)\epsilon(-w_2)]\\
    &\quad +\noisecov(w_1+w_2)\E[\epsilon(h+z_1)\epsilon(h-z_2)]-\noisecov(w_1+w_2)\noisecov(z_1+z_2)\\
    &\quad -\noisecov(w_1+w_2)\E[\epsilon(h+w_1)\epsilon(h-w_2)]+\noisecov(w_1+w_2)^2\\
    &=\E[\epsilon(z_1)\epsilon(-z_2)\epsilon(h+z_1)\epsilon(h-z_2)]- \noisecov(z_1+z_2)^2 -\E[\epsilon(z_1)\epsilon(-z_2)\epsilon(h+w_1)\epsilon(h-w_2)] +\noisecov(w_1+w_2)\noisecov(z_1+z_2)\\
    &\quad -\cancel{\noisecov(z_1+z_2)^2} +\cancel{\noisecov(z_1+z_2)^2} +\cancel{\noisecov(z_1+z_2)\noisecov(w_1+w_2)}-\cancel{\noisecov(z_1+z_2)\noisecov(w_1+w_2)}\\
    &\quad -\E[\epsilon(w_1)\epsilon(-w_2)\epsilon(h+z_1)\epsilon(h-z_2)]+\noisecov(z_1+z_2)\noisecov(w_1+w_2)+\E[\epsilon(w_1)\epsilon(-w_2)\epsilon(h+w_1)\epsilon(h-w_2)] - \noisecov(w_1+w_2)^2 \\
    &\quad +\cancel{\noisecov(w_1+w_2)\noisecov(z_1+z_2)}-\cancel{\noisecov(w_1+w_2)\noisecov(z_1+z_2)}-\cancel{\noisecov(w_1+w_2)^2}+\cancel{\noisecov(w_1+w_2)^2}\\
    &=\E[\epsilon(z_1)\epsilon(-z_2)\epsilon(h+z_1)\epsilon(h-z_2)]- \noisecov(z_1+z_2)^2-\E[\epsilon(z_1)\epsilon(-z_2)\epsilon(h+w_1)\epsilon(h-w_2)] +\noisecov(w_1+w_2)\noisecov(z_1+z_2)\\
    &\quad -\E[\epsilon(w_1)\epsilon(-w_2)\epsilon(h+z_1)\epsilon(h-z_2)]+\noisecov(z_1+z_2)\noisecov(w_1+w_2)+\E[\epsilon(w_1)\epsilon(-w_2)\epsilon(h+w_1)\epsilon(h-w_2)] - \noisecov(w_1+w_2)^2.
\end{align*}
\end{small}
By Isserlis's theorem, we can further expand 
\begin{small}
\begin{align*}
    &\cov_v(h)\\
    & = \E[\epsilon(z_1)\epsilon(-z_2)]\E [\epsilon(h+z_1)\epsilon(h-z_2)] + \E[\epsilon(z_1)\epsilon(h+z_1)]\E[\epsilon(-z_2)\epsilon(h-z_2)]+ \E[\epsilon(z_1)\epsilon(h-z_2)]\E[\epsilon(-z_2)\epsilon(h+z_1)]\\
    &\quad -\Big(\E[\epsilon(z_1)\epsilon(-z_2)]\E[\epsilon(h+w_1)\epsilon(h-w_2)]+\E[\epsilon(z_1)\epsilon(h+w_1)]\E[\epsilon(-z_2)\epsilon(h-w_2)]+\E[\epsilon(z_1)\epsilon(h-w_2)]\E[\epsilon(-z_2)\epsilon(h+w_1)]\Big)\\
    &\quad -\Big(\E[\epsilon(w_1)\epsilon(-w_2)]\E[\epsilon(h+z_1)\epsilon(h-z_2)]+\E[\epsilon(w_1)\epsilon(h+z_1)]\E[\epsilon(-w_2)\epsilon(h-z_2)]+\E[\epsilon(w_1)\epsilon(h-z_2)]\E[\epsilon(-w_2)\epsilon(h+z_1)]\Big)\\
    &\quad +\E[\epsilon(w_1)\epsilon(-w_2)]\E[\epsilon(h+w_1)\epsilon(h-w_2)]+\E[\epsilon(w_1)\epsilon(h+w_1)]\E[\epsilon(-w_2)\epsilon(h-w_2)]+\E[\epsilon(w_1)\epsilon(h-w_2)]\E[\epsilon(-w_2)\epsilon(h+w_1)]\\
    &\quad -\noisecov(z_1+z_2)^2+2\noisecov(w_1+w_2)\noisecov(z_1+z_2)-\noisecov(w_1+w_2)^2\\
    &=\cancel{\noisecov(z_1+z_2)^2} +\noisecov(h)^2 +\noisecov(h-z_2-z_1)\noisecov(h+z_1+z_2)\\
    &\quad -\cancel{\noisecov(z_1+z_2)\noisecov(w_1+w_2)}-\noisecov(h+w_1-z_1)\noisecov(h-w_2+z_2)-\noisecov(h-w_2-z_1)\noisecov(h+w_1+z_2)\\
    &\quad -\cancel{\noisecov(w_1+w_2)\noisecov(z_1+z_2)}-\noisecov(h+z_1-w_1)\noisecov(h-z_2+w_2)-\noisecov(h-z_2-w_1)\noisecov(h+z_1+w_2)\\
    &\quad +\cancel{\noisecov(w_1+w_2)^2}+\noisecov(h)^2+\noisecov(h-w_2-w_1)\noisecov(h+w_1+w_2)\\
    &\quad -\cancel{\noisecov(z_1+z_2)^2}+\cancel{2\noisecov(w_1+w_2)\noisecov(z_1+z_2)}-\cancel{\noisecov(w_1+w_2)^2}\\
    &=\big[\noisecov(h)^2 - \noisecov(h+w_1-z_1)\noisecov(h-w_2+z_2)\big] + \big[\noisecov(h-z_2-z_1)\noisecov(h+z_1+z_2)-\noisecov(h-w_2-z_1)\noisecov(h+w_1+z_2)\big] \\
    &\quad +\big[\noisecov(h)^2-\noisecov(h+z_1-w_1)\noisecov(h-z_2+w_2)\big] + \big[\noisecov(h-w_2-w_1)\noisecov(h+w_1+w_2)-\noisecov(h-z_2-w_1)\noisecov(h+z_1+w_2)\big]\\
    &\leq \Lipcov\Big(|\noisecov(h)||z_1-w_1|+|\noisecov(h+w_1-z_1)||z_2-w_2|+|\noisecov(h-z_1-w_2)||z_1-w_1|+|\noisecov(h+z_1+z_2)||z_2-w_2|\\
    &\quad |\noisecov(h)||z_1-w_1|_2 + |\noisecov(h+z_1-w_1)||z_2-w_2|_2+ |\noisecov(h-z_2-w_1)||z_1-w_1|_2+ |\noisecov(h+w_1+w_2)||z_2-w_2|_2\Big).
\end{align*}
\end{small}
Hence 
\begin{small}
\begin{align*}
    \int_{\R^d} |\cov_{v^{(1)}}(h)|dh \leq 4\Lipcov\|\noisecov\|_1(|z_1-w_1|_2+|z_2-w_2|_2)
    \leq 4\sqrt{2}\Lipcov\|\noisecov\|_1 |z-w|_2,
\end{align*}
\end{small}
so that 
\begin{align*}
     \E |\xxone(z)-\xxone(w)|^2 \leq 4\sqrt{2}\Lipcov\|f\|^2_2 \|\noisecov\|_1 |z-w|_2.
\end{align*}

For $v^{(2)}$, we have 
\begin{small}
\begin{align*}
    &\cov_{v^{(2)}}(h)\\
    &=\E v^{(2)}(0)v^{(2)}(h)\\
    &=\E[\epsilon(z_1)\epsilon(z_1+z_2)-\noisecov(z_2)-\epsilon(w_1)\epsilon(w_1+w_2)+\noisecov(w_2)]\\
    &\quad \cdot [\epsilon(h+z_1)\epsilon(h+z_1+z_2)-\noisecov(z_2)-\epsilon(h+w_1)\epsilon(h+w_1+w_2)+\noisecov(w_2)]\\
    &=\E[\epsilon
    (z_1)\epsilon(z_1+z_2)\epsilon(h+z_1)\epsilon(h+z_1+z_2)]- \noisecov(z_2)\E[\epsilon(z_1)\epsilon(z_1+z_2)]\\
    &\quad -\E[\epsilon(z_1)\epsilon(z_1+z_2)\epsilon(h+w_1)\epsilon(h+w_1+w_2)]+\noisecov(w_2)\E[\epsilon(z_1)\epsilon(z_1+z_2)]\\
    &\quad-\noisecov(z_2)\E[\epsilon(h+z_1)\epsilon(h+z_1+z_2)]+\noisecov(z_2)^2\\
    &\quad +\noisecov(z_2)\E[\epsilon(h+w_1)\epsilon(h+w_1+w_2)]-\noisecov(z_2)\noisecov(w_2) \\
    &\quad -\E[\epsilon(w_1)\epsilon(w_1+w_2)\epsilon(h+z_1)\epsilon(h+z_1+z_2)]+\noisecov(z_2)\E[\epsilon(w_1)\epsilon(w_1+w_2)] \\
    & \quad +\E[\epsilon(w_1)\epsilon(w_1+w_2)\epsilon(h+w_1)\epsilon(h+w_1+w_2)] - \noisecov(w_2)\E[\epsilon(w_1)\epsilon(w_1+w_2)]\\
    & \quad +\noisecov(w_2)\E[\epsilon(h+z_1)\epsilon(h+z_1+z_2)]-\noisecov(w_2)\noisecov(z_2)\\
    & \quad -\noisecov(w_2)\E[\epsilon(h+w_1)\epsilon(h+w_1+w_2)]+\noisecov(w_2)^2\\
    & =\E[\epsilon
    (z_1)\epsilon(z_1+z_2)\epsilon(h+z_1)\epsilon(h+z_1+z_2)] - \noisecov(z_2)^2 -\E[\epsilon(z_1)\epsilon(z_1+z_2)\epsilon(h+w_1)\epsilon(h+w_1+w_2)]+\noisecov(w_2)\noisecov(z_2)\\
    & \quad -\cancel{\noisecov(z_2)^2}+\cancel{\noisecov(z_2)^2} + \cancel{\noisecov(z_2)\noisecov(w_2)}- \cancel{\noisecov(z_2)\noisecov(w_2)}\\
    &\quad -\E[\epsilon(w_1)\epsilon(w_1+w_2)\epsilon(h+z_1)\epsilon(h+z_1+z_2)]+\noisecov(z_2)\noisecov(w_2)+\E[\epsilon(w_1)\epsilon(w_1+w_2)\epsilon(h+w_1)\epsilon(h+w_1+w_2)]-\noisecov(w_2)^2
    \\&\quad +\cancel{\noisecov(w_2)\noisecov(z_2)}-\cancel{\noisecov(w_2)\noisecov(z_2)} -\cancel{\noisecov(w_2)^2}+\cancel{\noisecov(w_2)^2}\\
    & =\E[\epsilon(z_1)\epsilon(z_1+z_2)]\E{\epsilon(h+z_1)]\epsilon(h+z_1+z_2)}+ \E[\epsilon(z_1)\epsilon(h+z_1)]\E[\epsilon(z_1+z_2)\epsilon(h+z_1+z_2)]\\
    &\quad + \E[\epsilon(z_1)\epsilon(h+z_1+z_2)]\E[\epsilon(z_1+z_2)\epsilon(h+z_1)] - \noisecov(z_2)^2 \\
    &\quad -\E[\epsilon(z_1)\epsilon(z_1+z_2)]\E[\epsilon(h+w_1)\epsilon(h+w_1+w_2)] - \E[\epsilon(z_1)\epsilon(h+w_1)]\E[\epsilon(z_1+z_2)\epsilon(h+w_1+w_2)] \\
    &\quad -\E[\epsilon(z_1)\epsilon(h+w_1+w_2)]\E[\epsilon(z_1+z_2)\epsilon(h+w_1)] +\noisecov(w_2)\noisecov(z_2)\\
    &\quad -\E[\epsilon(w_1)\epsilon(w_1+w_2)]\E[\epsilon(h+z_1)\epsilon(h+z_1+z_2)] - \E[\epsilon(w_1)\epsilon(h+z_1)]\E[\epsilon(w_1+w_2)\epsilon
    (h+z_1+z_2)]\\
    &\quad - \E[\epsilon(w_1)\epsilon(h+z_1+z_2)]\E[\epsilon(w_1+w_2)\epsilon(h+z_1)]+\noisecov(z_2)\noisecov(w_2)\\
    & \quad +\E[\epsilon(w_1)\epsilon(w_1+w_2)]\E[\epsilon(h+w_1)\epsilon(h+w_1+w_2)] + \E[\epsilon(w_1)\epsilon(h+w_1)]\E[\epsilon(w_1+w_2)\epsilon(h+w_1+w_2)] \\
    &\quad + \E[\epsilon(w_1)\epsilon(h+w_1+w_2)]\E[\epsilon(w_1+w_2)\epsilon(h+w_1)] - \noisecov(w_2)^2 \\
    & =\cancel{\noisecov(z_2)^2} + \noisecov(h)^2+ \noisecov(h+z_2)^2 - \cancel{\noisecov(z_2)^2} \\
    &\quad - \cancel{\noisecov(z_2)\noisecov(w_2)}-\noisecov(h+w_1-z_1)\noisecov(h+w_1-z_1+w_2-z_2)- \noisecov(h+w_1-z_1+w_2)\noisecov(h+w_1-z_1-z_2) + \cancel{\noisecov(w_2)\noisecov(z_2)}\\
    &\quad -\cancel{\noisecov(w_2)\noisecov(z_2)}-\noisecov(h+z_1-w_1)\noisecov(h+z_1-w_1+z_2-w_2)-\noisecov(h+z_1-w_1+z_2)\noisecov(h+z_1-w_1-w_2)+\cancel{\noisecov(z_2)\noisecov(w_2)}\\
    &\quad +\cancel{\noisecov(w_2)^2} +\noisecov(h)^2+\noisecov(h+w_2)^2-\cancel{\noisecov(w_2)^2}\\
    & = [\noisecov(h)^2-\noisecov(h+w_1-z_1)\noisecov(h+w_1-z_1+w_2-z_2)]+ [\noisecov(h+w_2)^2-\noisecov(h+w_1-z_1+w_2)\noisecov(h+w_1-z_1-z_2)]\\
    &\quad +[\noisecov(h)^2- \noisecov(h+z_1-w_1)\noisecov(h+z_1-w_1+z_2-w_2)] + [\noisecov(h+z_2)^2-\noisecov(h+z_1-w_1+z_2)\noisecov(h+z_1-w_1-w_2)]\\
    & \leq L|\noisecov(h)||z_1-w_1|_2 + L|\noisecov(h+w_1-z_1)| (|z_1-w_1|_2+|z_2-w_2|_2) \\
    &\quad + L|\noisecov(h+w_2)||z_1-w_1|_2+L|\noisecov(h+w_1-z_1+w_2)|(|z_1-w_1|_2+|z_2-w_2|_2)\\
    &\quad +L|\noisecov(h)||z_1-w_1|_2+ L|\noisecov(h+w_1-z_1)|(|z_1-w_1|_2+|z_2-w_2|_2) \\
    &\quad +L|\noisecov(h+w_2)||z_1-w_1|_2 + L |\noisecov(h+z_1-w_1+z_2)|(|z_1-w_1|_2+|z_2-w_2|_2). 
\end{align*}
\end{small}
Therefore, 
\begin{align*}
    \int_{\R^d} |\cov_{v^{(2)}}(h)|dh \leq 8\sqrt{2} \Lipcov \|\noisecov\|_1 |z-w|_2,%
\end{align*}
and
\begin{align*}
     \E |\xxtwo(z)-\xxtwo(w)|^2 \leq 8 \sqrt{2}\Lipcov  \|f\|^2_2\|\noisecov\|_1 |z-w|_2. 
\end{align*}
The bound for $v^{(3)}$ is proved similarly as for $v^{(2)}$.

\textbf{Bound for $\xone$, $\xtwo$, $\xthree$:}
Notice from \eqref{eq:type 1 error} that for each $i$, the type I error process takes the form 
\begin{align}\label{eq:type 1 general}
    \E |E_1^{(i)}(z)-E_1^{(i)}(w)|^2 = \frac{1}{\sam}\E \left|\int_{[-\dom,\dom]^d} G^{(i)}(t)\epsilon(t)\right|^2,
\end{align}
where 
\begin{align*}
    G^{(1)}(t)&=\sum_{i=1}^\sam  f(t-x_i+z_1)f(t-x_i-z_2) -  f(t-x_i+w_1)f(t-x_i-w_2)\\
    G^{(2)}(t)&=\sum_{i=1}^\sam  f(t-x_i+z_2)f(t-x_i+z_1+z_2)-f(t-x_i+w_2)f(t-x_i+w_1+w_2)\\
    G^{(3)}(t)&=\sum_{i=1}^\sam  f(t-x_i-z_1)f(t-x_i-z_1-z_2)-f(t-x_i-w_1)f(t-x_i-w_1-w_2)
\end{align*}
We have 
\begin{align*}
    \eqref{eq:type 1 general}
    & = \frac{1}{\sam}\int_{[-\dom,\dom]^d}\int_{[-\dom,\dom]^d}G^{(i)}(t)G^{(i)}(s)\noisecov(t-s)dtds\\
    & = \frac{1}{\sam}\int_{\R^d} \int_{[-\dom,\dom]^d\cap ([-\dom,\dom]^d-u)} G^{(i)}(s+h)G^{(i)}(s)\noisecov(h)dsdh,
\end{align*}
where $[-\dom,\dom]^d-h=\{s:s+h\in [-\dom,\dom]^d\}$. 
By Cauchy-Schwarz, we have 
\begin{align*}
    \frac{1}{\sam}\int_{[-\dom,\dom]^d\cap ([-\dom,\dom]^d-u)} G^{(i)}(s+h)G^{(i)}(s) ds \leq \frac{1}{\sam}\int_{[-\dom,\dom]^d} G^{(i)}(s)^2 ds.
\end{align*}
It suffices to bound this last integral. 

Since $f$ is supported on $[-\pi,\pi]$ and $z,w\in [-2\pi,2\pi]^{2d}$, for each $t$, only one term remains in the definition of $G^{(i)}$ (which we denote as the $i_t$-$th$ term) and 
\begin{align*}
    |G^{(1)}(t)| &= |f(t-x_{i_t}+z_1)f(t-x_{i_t}-z_2) -  f(t-x_{i_t}+w_1)f(t-x_{i_t}-w_2)|\\
    & \leq |f(t-x_{i_t}+z_1)f(t-x_{i_t}-z_2)-f(t-x_{i_t}+z_1)f(t-x_{i_t}-w_2)|\\
    &\quad + |f(t-x_{i_t}+z_1)f(t-x_{i_t}-w_2)-f(t-x_{i_t}+w_1)f(t-x_{i_t}-w_2)|\\
    &\leq \|f\|_\infty L_f |z_2-w_2|_2+ \|f\|_\infty L_f |z_1-w_1|_2\\
    &\leq  \sqrt{2}L_f \|f\|_\infty|z-w|_2.
\end{align*}
Therefore 
\begin{align*}
    \E |\xone(z)-\xone(w)|^2 \leq 2\gamma^{-1} L_f^2 \|\noisecov\|_1\|f\|_\infty^2 |z-w|_2^2.
\end{align*}

Similarly,  
\begin{align*}
    |G^{(2)}(t)| &= |f(t-x_{i_t}+z_2)f(t-x_{i_t}+z_1+z_2)-f(t-x_{i_t}+w_2)f(t-x_{i_t}+w_1+w_2)|\\
    &\leq |f(t-x_{i_t}+z_2)f(t-x_{i_t}+z_1+z_2)-f(t-x_{t_t}+w_2)f(t-x_{i_t}+z_1+z_2)| \\
    &\quad + |f(t-x_{t_t}+w_2)f(t-x_{i_t}+z_1+z_2)-f(t-x_{i_t}+w_2)f(t-x_{i_t}+w_1+w_2)|\\
    & \leq \|f\|_\infty L_f|z_2-w_2|_2 + \|f\|_\infty L_f (|z_1-w_1|_2+|z_2-w_2|_2) \leq 2\sqrt{2} L_f \|f\|_\infty |z-w|_2.
\end{align*}
Therefore 
\begin{align*}
    \E|\xtwo(z)-\xtwo(w)|^2 \leq 4\gamma^{-1} L_f^2 \|\noisecov\|_1\|f\|^2_\infty |z-w|^2_2. 
\end{align*}
The bound for $\xthree$ is proved similarly to that for $\xtwo$.

\end{proof}

\begin{rmk}\label{rmk:n_est}
Under Assumption \ref{assump:standing}, we can estimate the number of signal occurrences $N$ from the MTD observation $\obs$ in the following way. We present this result informally and comment that the concentration results proven in Section \ref{subsec:bispectrum_conc} extend to this estimator as well. First observe that $\E[A_1\obs] = \sam \fft(0)$. Then,
\begin{align*}
 \int_{[-2\pi,2\pi]^d} A_2\obs(z)\ dz &= \int_{[-2\pi,2\pi]^d} \int_{[-\dom,\dom]^d} (F(t)+\epsilon(t))(F(t+z)+\epsilon(t+z))\ dt \ dz  \\
 &= \int_{[-2\pi,2\pi]^d} \int_{[-\dom,\dom]^d} (F(t)F(t+z)+\epsilon(t)F(t+z) + F(t)\epsilon(t+z)+\epsilon(t)\epsilon(t+z))\ dt \ dz  \\
 &=  \int_{[-2\pi,2\pi]^d} A_2F(z)\ dz + \text{mean zero cross terms} + \int_{[-2\pi,2\pi]^d} A_2\epsilon(z)\ dz,
\end{align*}
and taking the expected value of both sides yields:
\[ \E\left[ \int_{[-2\pi,2\pi]^d} A_2\obs(z)\ dz\right] = \sam \fft(0)^2 + \E\left[ \int_{[-2\pi,2\pi]^d} A_2\epsilon(z)\ dz\right]. \]
Since $\int_{[-2\pi,2\pi]^d} A_2\obs(z)\ dz$ concentrates around its mean, we can thus estimate $\sam$ by
\[ \widehat{\sam}:= \frac{(A_1\obs)^2 }{\int_{[-2\pi,2\pi]^d} A_2\obs(z)\ dz - \E\left[ \int_{[-2\pi,2\pi]^d} A_2\epsilon(z)\ dz\right] } = \frac{(A_1\obs)^2 }{\int_{[-2\pi,2\pi]^d} A_2 \obs(z)\ dz - (2\dom)^d \int_{[-2\pi,2\pi]^d} \noisecov(z)\ dz }\, .\]
Note $\fft(0)$ can then be estimated by $\widehat{\sam}/A_1Y$.
\end{rmk}
\begin{lem}\label{lem:kot_mp}
Under Assumption \ref{assump:standing}, let $g(\omega)=\log \fft(\omega)$ for $f:\R\rightarrow\R$, so that $g(0)=0$ and $g$ is well defined since $\fft$ is non-vanishing. The following recursion formula holds for any $0 \le \tau \le k$:
\begin{equation}
\label{equ:kot_recursion}
  g(k) = g(c) +\int_0^{k-c} \partial_1 \log \Bf(\xi+c,c)\ d\xi - \overline{g(k-c)}.  
\end{equation}
This gives a family of \textit{horizontal} integrals in the $\omega_1, \omega_2$ plane, starting on the diagonal point $(\tau,\tau)$ and ending at $(k,\tau)$.
\end{lem}

\begin{proof}
From the definition of the bispectrum,
\[
\Bf(\omega_1,\omega_2)
=
\fft(\omega_1)\fft(-\omega_2)\fft(\omega_2-\omega_1),
\]
taking logarithms gives
\[
\log \Bf(\omega_1,\omega_2)
=
g(\omega_1)+g(-\omega_2)+g(\omega_2-\omega_1).
\]
Differentiating with respect to $\omega_1$ yields
\[
\partial_1 \log \Bf(\omega_1,\omega_2)
=
g'(\omega_1)-g'(\omega_2-\omega_1).
\]

Fix $\tau$ and set
\[
\omega_1 = \xi + \tau, \qquad \omega_2 = \tau.
\]
Then
\[
\omega_2-\omega_1 = \tau-(\xi+\tau) = -\xi,
\]
so
\[
\partial_1 \log \Bf(\xi+\tau,\tau)
=
g'(\xi+\tau)-g'(-\xi).
\]

Integrating over $\xi \in [0,k-\tau]$ gives
\[
\int_0^{k-\tau} \partial_1 \log \Bf(\xi+\tau,\tau)\, d\xi
=
\int_0^{k-\tau} g'(\xi+\tau)\, d\xi
-
\int_0^{k-\tau} g'(-\xi)\, d\xi.
\]

The first term becomes
\[
\int_0^{k-\tau} g'(\xi+\tau)\, d\xi
=
\int_c^k g'(t)\, dt
=
g(k)-g(\tau).
\]

For the second term, substitute $s=-\xi$ to obtain
\[
-
\int_0^{k-\tau} g'(-\xi)\, d\xi =
\int_0^{\tau-k} g'(s)\, ds = g(\tau-k) - g(0) = g(\tau-k) \, .
\]

Therefore,
\[ \int_0^{k-\tau} \partial_1 \log \Bf(\xi+\tau,\tau)\, d\xi
= g(k)-g(\tau) + g(\tau-k) \, .\]
Rearranging and utilizing $g(\tau-k)=\overline{g(k-\tau)}$:
\[
g(k)
= g(\tau) +
\int_0^{k-\tau} \partial_1 \log \Bf(\xi+\tau,\tau)\, d\xi
-
\overline{g(k-\tau)}.
\]
which gives the recursion formula.
\end{proof}

\end{appendix}

\end{document}